\documentclass[twoside]{IEEEtran}
\usepackage{amsmath}
\usepackage{amssymb}
\usepackage{amsthm}
\usepackage{enumerate}

\usepackage{epsfig}
\usepackage{subfigure}
\usepackage{psfrag}
\usepackage{cite}
\usepackage{stfloats}
\usepackage[mathscr]{euscript}
\usepackage{hyperref}
\hypersetup{colorlinks=true, linkcolor=blue}

\newtheorem{thm}{Theorem}

\newtheorem{lemma}{Lemma}

\newtheorem{defn}{Definition}
\newtheorem{property}{Property}

\theoremstyle{remark}
\newtheorem{remark}{Remark}

\newcommand{\bigo}[1]{O\left(#1\right)}

\newcommand{\Qinv}[1]{Q^{-1}\left(#1\right)}
\newcommand{\Var}[1]{{\rm Var}\left[#1\right]}
\newcommand{\E}[1]{\mathbb E\left[#1\right]}
\newcommand{\Prob}[1]{\mathbb P\left[#1\right]}
\newcommand{\Qrob}[1]{\mathbb Q\left[#1\right]}
\newcommand{\beq}{\begin{equation}}
\newcommand{\eeq}{\end{equation}}

\def\binosum#1#2{\left\langle{#1\atopwithdelims..#2}\right\rangle}

\graphicspath{{figures/}}

\psfrag{Rd}{\scriptsize{$R(d)$}}
\psfrag{n}{\scriptsize{$n$}}
\psfrag{R}{\scriptsize{$R$}}
\psfrag{d}{\scriptsize{$d$}}
\psfrag{Rlim}{\scriptsize{$\dfrac{C}{R(d)}$}}
\psfrag{Vjscc}{\scriptsize{$ \mathscr V(d, \alpha)$}}

\psfrag{AEBMS}{\scriptsize{Achievability, SSCC \eqref{eq:AgSeparate}}}
\psfrag{AjEBMS}{\scriptsize{Achievability, JSCC \eqref{eq:ABMS-BSC}}}
\psfrag{CEBMS}{\scriptsize{Converse \eqref{eq:CBMS-BSC}}}
\psfrag{ClistEBMS}{\scriptsize{Converse \eqref{eq:ChtEBMS-BSC}}}
\psfrag{NEBMS}{\scriptsize{Approximation, JSCC = SSCC \eqref{eq:2order}}}

\psfrag{ABMS}{\scriptsize{Achievability, JSCC \eqref{eq:ABMS-BSC}}}
\psfrag{NBMS}{\scriptsize{Approximation, JSCC \eqref{eq:2order}}}
\psfrag{NdBMS}{\scriptsize{Approximation, JSCC \eqref{eq:2orderD}}}
\psfrag{CgBMS}{\scriptsize{Converse \eqref{eq:CBMS-BSC}}}
\psfrag{CBMS}{\scriptsize{Converse \eqref{eq:ChtBMS-BSC}}}
\psfrag{ClistBMS}{\scriptsize{Converse \eqref{eq:ChtBMS-BSC}}}
\psfrag{AsBMS}{\scriptsize{Achievability, SSCC \eqref{eq:AgSeparate}}}
\psfrag{NsBMS}{\scriptsize{Approximation, SSCC \eqref{eq:2orderSeparate}}}

\psfrag{Alossless}{\scriptsize{Achievability, JSCC \eqref{eq:Alossless}}}
\psfrag{Aslossless}{\scriptsize{Achievability, SSCC \eqref{eq:AgSeparate}}}

\psfrag{AGMS}{\scriptsize{Achievability, JSCC \eqref{eq:AGMS-AWGN}}}
\psfrag{CGMS}{\scriptsize{Converse \eqref{eq:C1GMS-AWGN}}}
\psfrag{ChtGMS}{\scriptsize{Converse \eqref{eq:ChtGMS-AWGN}}}
\psfrag{NGMS}{\scriptsize{Approximation, JSCC \eqref{eq:2order}}}
\psfrag{NdGMS}{\scriptsize{Approximation, JSCC \eqref{eq:2orderD}}}
\psfrag{AsGMS}{\scriptsize{Achievability, SSCC \eqref{eq:AgSeparate}}}
\psfrag{NsGMS}{\scriptsize{Approximation, SSCC \eqref{eq:2orderSeparate}}}

\psfrag{A1BMS}{\scriptsize{No coding \eqref{eq:A1_BMS-BSC}}}
\psfrag{AC1EBMS}{\scriptsize{No coding = converse \eqref{eq:A1_BMS-BSC}}}
\psfrag{A1GMS}{\scriptsize{No coding \eqref{eq:d1GMS-AWGN}}}
\psfrag{N1GMS}{\scriptsize{Approximation, no coding \eqref{eq:dispersion1}}}

\title{Lossy joint source-channel coding\\ in the finite blocklength regime}

\author{
Victoria Kostina,  \IEEEmembership{Student Member, IEEE}, \and Sergio Verd\'{u},
\IEEEmembership{Fellow, IEEE} 
\thanks{
This work was supported in part by the National Science Foundation (NSF)
under Grant CCF-1016625 and by the Center for Science of Information
(CSoI), an NSF Science and Technology Center, under Grant CCF-0939370. The work of V. Kostina was supported in part by the Natural Sciences and Engineering Research Council of Canada.

Portions
of this paper were presented at the 2012 IEEE International Symposium on Information Theory \cite{kostina2012ISITjscc}, and at the 2012 IEEE Information Theory Workshop \cite{kostina2012tocodeornot}.

The authors are with the Department of Electrical Engineering, Princeton
University, NJ 08544 USA (e-mail: vkostina@princeton.edu; verdu@princeton.
edu).
}
}

\IEEEoverridecommandlockouts
\begin{document}

\maketitle
\begin{abstract}
This paper finds new tight finite-blocklength bounds for the best achievable lossy joint source-channel code rate, and demonstrates that joint source-channel  code design brings considerable performance advantage over a separate one in the non-asymptotic regime. A joint source-channel code maps a block of $k$ source symbols onto a length$-n$ channel codeword, and the fidelity of reproduction at the receiver end is measured by the probability $\epsilon$ that the distortion exceeds a given threshold $d$.   For memoryless sources and channels, it is demonstrated that the parameters of the best joint source-channel code must satisfy
$nC - kR(d) \approx \sqrt{nV + k \mathcal V(d)} \Qinv{\epsilon}$,
where $C$ and $V$ are the channel capacity and channel dispersion, respectively; $R(d)$ and $\mathcal V(d)$ are the source rate-distortion and rate-dispersion functions; and $Q$ is the standard Gaussian complementary cdf.
Symbol-by-symbol (uncoded) transmission is known to achieve the Shannon limit when the source and channel satisfy a certain probabilistic matching condition. In this paper we show that even  when this condition is not satisfied, symbol-by-symbol transmission is, in some cases, the best known strategy in the non-asymptotic regime. 
\end{abstract}


\section{Introduction}
\label{sec:intro}
In the limit of infinite blocklengths, the optimal achievable coding rates in channel coding and lossy data compression are characterized by the channel capacity $C$ and the source rate-distortion function $R(d)$, respectively \cite{shannon1948mathematical}. For a large class of sources and channels, in the limit of large blocklength, the maximum achievable joint source-channel coding (JSCC) rate compatible with vanishing excess distortion probability is characterized  by the ratio $\frac C {R(d)}$ \cite{shannon1959coding}. A perennial question in information theory is how relevant the asymptotic fundamental limits are when the communication system is forced to operate at a given fixed blocklength. The finite blocklength (delay) constraint is inherent to all communication scenarios. In fact, in many systems of current interest, such as real-time multimedia communication, delays are strictly constrained, while in packetized data communication, packets are frequently on the order of 1000 bits. While computable formulas for the channel capacity and the source rate-distortion function are available for a wide class of channels and sources, the luxury of being able to compute exactly (in polynomial time) the non-asymptotic fundamental limit of interest is rarely affordable. Notable exceptions where the non-asymptotic fundamental limit is indeed computable are almost lossless source coding \cite{verdu2012losslessCISS,verdu2012lossless}, and JSCC over matched source-channel pairs \cite{gastpar2003tocodeornot}. In general, however, one can at most hope to obtain bounds and approximations to the information-theoretic non-asymptotic fundamental limits.  
 
Although non-asymptotic bounds can be distilled from classical proofs of coding theorems, these bounds are rarely satisfyingly tight in the non-asymptotic regime, as studied in \cite{polyanskiy2010channel,kostina2011fixed} in the contexts of channel coding and lossy source coding, respectively. For the JSCC problem, the classical converse is based on the mutual information data processing inequality, while the classical achievability scheme uses separate source/channel coding (SSCC), in which the channel coding block and the source coding block are optimized separately without knowledge of each other. These conventional approaches lead to disappointingly weak non-asymptotic bounds. In particular, SSCC can be rather suboptimal non-asymptotically. An accurate finite blocklength analysis therefore calls for novel upper and lower bounds that sandwich tightly the non-asymptotic fundamental limit. Such bounds were shown in \cite{polyanskiy2010channel} for the channel coding problem and in \cite{kostina2011fixed} for the source coding problem. In this paper, we derive new tight bounds for the JSCC problem, which hold in full generality, without any assumptions on the source alphabet, stationarity or memorylessness.
 
While numerical evaluation of the non-asymptotic upper and lower bounds bears great practical interest (for example, to decide how suboptimal with respect to the information-theoretic limit a given blocklength-$n$ code is), such bounds usually involve cumbersome expressions that offer scant conceptual insight. Somewhat ironically, to get an elegant, insightful approximation of the non-asymptotic fundamental limit, one must resort to an asymptotic analysis of these non-asymptotic bounds. Such asymptotic analysis must be finer than that based on the law of large numbers, which suffices to obtain the asymptotic fundamental limit but fails to provide any estimate of the speed of convergence to that limit. There are two complementary approaches to a finer asymptotic analysis: the large deviations analysis which leads to error exponents, and the Gaussian approximation analysis which leads to dispersion. The error exponent approximation and the Gaussian approximation to the non-asymptotic fundamental limit are tight in different operational regimes. In the former, a rate which is strictly suboptimal with respect to the asymptotic fundamental limit is fixed, and the error exponent measures the exponential decay of the error probability to $0$ as the blocklength increases. The error exponent approximation is tight if the error probability a system can tolerate is extremely small. However, already for probability of error as low as $10^{-6}$ to $10^{-1}$, which is the operational regime for many high data rate applications, the Gaussian approximation, which gives the optimal rate achievable at a given error probability as a function of blocklength, is tight \cite{polyanskiy2010channel,kostina2011fixed}. In the channel coding problem, the Gaussian approximation of $R^\star (n, \epsilon)$, the maximum achievable finite blocklength coding rate at blocklength $n$ and error probability $\epsilon$, is given by, for finite alphabet stationary memoryless channels \cite{polyanskiy2010channel}, 
\begin{equation}
 n R^\star (n, \epsilon) = nC - \sqrt{nV} \Qinv{\epsilon} + \bigo{\log n}  \label{eq:2orderCC}
\end{equation}
where $C$ and $V$ are the channel capacity and dispersion, respectively. 
In the lossy source coding problem, the Gaussian approximation of $R^\star (k, d, \epsilon)$, the minimum achievable finite blocklength coding rate at blocklength $k$ and probability $\epsilon$ of exceeding fidelity $d$, is given by, for stationary memoryless sources \cite{kostina2011fixed}, 
\begin{equation}
 k R^\star (k, d, \epsilon) = k R(d) + \sqrt{k \mathcal V(d)} \Qinv{\epsilon} + \bigo{\log k} \label{eq:2orderSC}
\end{equation}
where $R(d)$ and $\mathcal V(d)$ are the rate-distortion and the rate-dispersion functions, respectively. 

For a given code, the excess distortion constraint, which is the figure of merit in this paper as well as in \cite{kostina2011fixed}, is, in a way, more fundamental than the average distortion constraint, because varying $d$ over its entire range and evaluating the probability of exceeding $d$ gives full information about the distribution (and not just its mean) of the distortion incurred at the decoder output. Following the philosophy of \cite{polyanskiy2010channel,kostina2011fixed}, in this paper we perform the Gaussian approximation analysis of our new bounds to show that $k$, the maximum number of source symbols transmissible using a given channel blocklength $n$, must satisfy
\begin{equation}
nC - kR(d) = \sqrt{nV + k \mathcal V(d)} \Qinv{\epsilon}+ \bigo{\log n} \label{eq:2orderJSCCintro}
\end{equation}
 under the fidelity constraint of exceeding a given distortion level $d$ with probability $\epsilon$. In contrast, if, following the SSCC paradigm, we just concatenate the channel code in \eqref{eq:2orderCC} and the source code in \eqref{eq:2orderSC}, we obtain
\begin{align}
 nC - kR(d) &\leq \min_{\eta + \zeta \leq \epsilon} \left\{ \sqrt{nV} \Qinv{\eta} + \sqrt{k\mathcal V(d)} \Qinv{\zeta}\right\} \notag\\
 &+ \bigo{\log n} \label{eq:2orderSSCCintro}
\end{align} 
which is usually strictly suboptimal with respect to \eqref{eq:2orderJSCCintro}. 


In addition to deriving new general achievability and converse bounds for JSCC and performing their Gaussian approximation analysis, in this paper we revisit the dilemma of whether one should or
should not code when operating under delay constraints. Gastpar et al. \cite{gastpar2003tocodeornot} gave a set of necessary and sufficient conditions on the source, its distortion measure, the channel and its cost function in order for symbol-by-symbol transmission to attain the minimum average distortion. In these curious cases, the source and the channel are probabilistically matched.
In the absence of channel cost constraints, we show that whenever the source and the channel are probabilistically matched so that symbol-by-symbol coding achieves the minimum average distortion,  it also
achieves the dispersion of joint source-channel coding. Moreover,
even in the absence of such a match between the
source and the channel, symbol-by-symbol transmission, though
asymptotically suboptimal, might outperform in the non-asymptotic regime not only separate
source-channel coding but also our random-coding
achievability bound.

Prior research relating to finite blocklength analysis of JSCC includes the work of Csisz{\'a}r  \cite{csiszar1980joint,csiszar1982error} who demonstrated that the error exponent of joint source-channel coding outperforms that of separate source-channel coding. For discrete source-channel pairs with average distortion criterion, Pilc's achievability bound \cite{pilc1967coding,pilc1967transmission} applies. For the transmission of a Gaussian source over a discrete channel under the average mean square error constraint, Wyner's achievability bound \cite{wyner1968communication,wyner1972transmission} applies. Non-asymptotic achievability and converse bounds for a graph-theoretic model of JSCC have been obtained by Csisz{\'a}r \cite{csiszar1983abstract}. Most recently, Tauste Campo et al. \cite{campo2011random} showed a number of finite-blocklength random-coding bounds applicable to the almost-lossless JSCC setup, while Wang et al. \cite{wang2011dispersion}  found the dispersion of JSCC for sources and channels with finite alphabets. 


The rest of the paper is organized as follows. Section \ref{sec:prelim} summarizes basic definitions and notation. Sections \ref{sec:C} and \ref{sec:A} introduce the new converse and achievability bounds to the maximum achievable coding rate, respectively.  A Gaussian approximation analysis of the new bounds is presented in Section \ref{sec:2order}. The evaluation of the bounds and the approximation is performed for two important special cases:  the transmission of a  binary memoryless source (BMS) over a binary symmetric channel (BSC) with bit error rate distortion (Section \ref{sec:BMS_BSC}) and the transmission of a Gaussian memoryless source (GMS) with mean-square error distortion over an AWGN channel with a total power constraint (Section \ref{sec:GMS_AWGN}). Section \ref{sec:tocodeornot} focuses on symbol-by-symbol transmission.  

\section{Definitions}
\label{sec:prelim}
A lossy source-channel code is a pair of (possibly randomized) mappings $\mathsf f \colon \mathcal M \mapsto \mathcal X$ and $\mathsf g \colon \mathcal Y \mapsto \widehat{\mathcal M}$. A distortion measure $\mathsf d \colon \mathcal M \times \widehat{\mathcal M}  \mapsto [0, +\infty]$ is used to quantify the performance of the lossy code. A cost function $\mathsf c \colon \mathcal X \mapsto [0, +\infty]$ may be imposed on the channel inputs. The channel is used without feedback.

\label{sec:prelim}
\begin{defn}
The pair $(\mathsf f, \mathsf g)$ is a $(d, \epsilon, \alpha)$ lossy source-channel code for 
$\{\mathcal M, \ \mathcal X,\ \mathcal Y, \ \widehat{\mathcal M}, 
\ P_S, \mathsf d, 
\ P_{Y|X}, \ \mathsf c\}$ 
if 
$\mathbb P \left[ \mathsf d\left( S, \mathsf g(Y)\right) > d\right] \leq \epsilon$
and either 
$\E{\mathsf c(X)} \leq \alpha$ (average cost constraint) or 
$\mathsf c(X) \leq \alpha$ a.s. (maximal cost constraint),
 where $\mathsf{f} (S) = X$ (see Fig. \ref{fig:jscc}). In the absence of an input cost constraint we simplify the terminology and refer to the code as $(d, \epsilon)$ lossy source-channel code.
\label{defn:jscc}
\end{defn}

\begin{figure}[htbp]
\begin{center}
    \epsfig{file=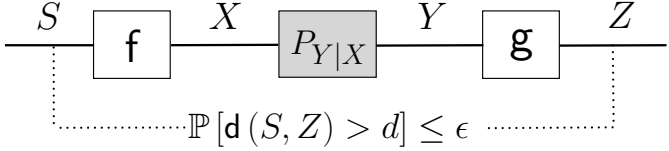,width=1\linewidth}
\caption{A $(d, \epsilon)$ joint source-channel code. }
\label{fig:jscc}
\end{center}
\end{figure}

The special case $d = 0$ and $\mathsf d(s, z) = 1 \left\{ s \neq z \right\}$ corresponds to almost-lossless compression. If, in addition, $P_S$ is equiprobable on an alphabet of cardinality 
$\left| \mathcal{M} \right| = | \widehat{\mathcal{M}} | = M$, 
a $(0, \epsilon, \alpha)$ code in Definition \ref{defn:jscc} corresponds to an $(M, \epsilon, \alpha)$ channel code (i.e. a code with $M$ codewords and average error probability $\epsilon$ and cost $\alpha$). On the other hand, if $P_{Y|X}$ is an identity mapping on an alphabet of cardinality $M$ without cost constraints, a $(d, \epsilon)$ code in Definition \ref{defn:jscc} corresponds to an $(M, d, \epsilon)$ lossy compression code (as e.g. defined in \cite{kostina2011fixed}). 

As our bounds in Sections \ref{sec:C} and \ref{sec:A} do not foist a Cartesian structure on the  underlying alphabets, we state them in the one-shot paradigm of Definition \ref{defn:jscc}. When we apply those bounds to the block coding setting, transmitted objects indeed become vectors, and Definition \ref{defn:jscckn} below comes into play.    
\begin{defn}
In the conventional fixed-to-fixed (or block) setting in which $\mathcal X$ and $\mathcal Y$ are the $n-$fold Cartesian products of alphabets $\mathcal A$ and $\mathcal B$, $\mathcal M$ and $\widehat {\mathcal M}$ are the $k-$fold Cartesian products of alphabets $\mathcal S$ and $\hat{\mathcal S}$, and $\mathsf d_k \colon \mathcal S^k \times \hat{\mathcal S}^k \mapsto  [0, +\infty]$, 
$\mathsf c_n \colon \mathcal A^n \mapsto [0, +\infty]$,  
a $(d, \epsilon, \alpha)$ code for 
$\{\mathcal S^k, \ \mathcal A^n,\ \mathcal B^n, \ \hat{ \mathcal S}^k, \ P_{S^k}, \ \mathsf d_k, \ P_{Y^n|X^n}, \ \mathsf c_n\}$ is called a $(k, n, d, \epsilon, \alpha)$ code (or a $(k, n, d, \epsilon)$ code if there is no cost constraint).
\label{defn:jscckn}
\end{defn}

\begin{defn}
Fix $\epsilon$, $d$, $\alpha$ and the channel blocklength $n$. The maximum achievable source blocklength and coding rate (source symbols per channel use) are defined by, respectively
\begin{align}
k^\star(n, d, \epsilon, \alpha) &=  \sup\left\{ k \colon \exists (k, n, d, \epsilon, \alpha) \text{ code}\right\} \label{eq:kstar}\\
R(n, d, \epsilon, \alpha) &= \frac 1 n k^\star(n, d, \epsilon, \alpha) \label{eq:Rstar}
\end{align}
Alternatively, fix $\epsilon$, $\alpha$, source blocklength $k$ and channel blocklength $n$. The minimum achievable excess distortion is defined by
\begin{equation}
D(k, n, \epsilon, \alpha) =  \inf\left\{ d \colon \exists (k, n, d, \epsilon, \alpha) \text{ code}\right\} \label{eq:Dstar}
\end{equation}
\end{defn}

Denote, for a given $P_{Y|X}$ and a cost function 
$\mathsf c \colon \mathcal X \mapsto [0, +\infty]$,
\begin{equation}
\mathbb C(\alpha) = \sup_{\substack{P_{X} \colon\\
\E{ \mathsf c(X)} \leq \alpha
}} I(X; Y) \label{eq:Capacity}
\end{equation}
and, for a given $P_S$ and a distortion measure $\mathsf d \colon \mathcal M \times \widehat {\mathcal M}  \mapsto [0, +\infty]$,
\begin{equation}
\mathbb R_S(d) =  \inf_{\substack{P_{Z|S}\colon\\ \E{\mathsf d(S, Z)} \leq d}} I(S; Z) \label{eq:RR(d)}
\end{equation}
We impose the following basic restrictions on $P_{Y|X}$, $P_S$, the input-cost function and the distortion measure:
 \begin{enumerate}[(a)]
  \item  \label{item:a}
 $\mathbb R_S(d)$ is finite for some $d$, i.e. $ d_{\min} < \infty$, where
\begin{equation}
 d_{\min} = \inf \left\{ d\colon ~ \mathbb R_S(d) < \infty \right\}; \label{eq:dmin}
\end{equation}
\item \label{item:b} The infimum in \eqref{eq:RR(d)} is achieved by a unique $P_{Z^\star|S}$;
\item \label{item:c} The supremum in \eqref{eq:Capacity} is achieved by a unique $P_{X^\star}$.
\end{enumerate}

The dispersion, which serves to quantify the penalty on the rate of the best JSCC code induced by the finite blocklength, is defined as follows.
\begin{defn}
Fix $\alpha$ and $d
\geq d_{\min}$. The rate-dispersion function of joint source-channel coding (source samples squared per channel use) is defined as
\begin{equation}
 \mathscr V(d, \alpha)= \lim_{\epsilon \rightarrow 0} \limsup_{n \rightarrow \infty}
 \frac { n \left( \frac{C(\alpha)}{R(d)} - R(n,d,\epsilon, \alpha) \right)^2}{2 \log_e \frac 1 \epsilon} \label{eq:dispersiondef}
\end{equation}
where $C(\alpha)$ and $R(d)$ are the channel capacity-cost and source rate-distortion functions, respectively.\footnote{While for memoryless sources and channels, $C(\alpha) = \mathbb C(\alpha)$ and $R(d) = \mathbb R_{\mathsf S}(d)$ given by \eqref{eq:Capacity} and \eqref{eq:RR(d)} evaluated with single-letter distributions, it is important to distinguish between the operational definitions and the extremal mutual information quantities, since the core results in this paper allow for memory. 
}

The distortion-dispersion function of joint source-channel coding is defined as
\begin{equation}
 \mathscr W(R, \alpha)= \lim_{\epsilon \rightarrow 0} \limsup_{n \rightarrow \infty}
\frac {  n \left( D\left(\frac {C(\alpha)}{R}\right) - D(nR, n,\epsilon, \alpha) \right)^2 }{ 2 \log_e \frac 1 \epsilon }\label{eq:distortiondispersiondef}
\end{equation}
where $D(\cdot)$ is the distortion-rate function of the source. 
\label{defn:dispersion}
\end{defn}
If there is no cost constraint, we will simplify notation by dropping $\alpha$ from \eqref{eq:kstar}, \eqref{eq:Rstar}, \eqref{eq:Dstar}, \eqref{eq:Capacity}, \eqref{eq:dispersiondef} and \eqref{eq:distortiondispersiondef}.

\begin{defn}[$\mathsf d-$tilted information \cite{kostina2011fixed}] \label{defn:id}
For $d > d_{\min}$, the $\mathsf d-$tilted information in $s$ is defined as\footnote{All $\log$'s and $\exp$'s are in an arbitrary common base. }
 \begin{equation}
\jmath_{S}(s, d) =\log \frac 1 {\E{ \exp\left( \lambda^\star d - \lambda^\star \mathsf d(s, Z^\star)\right)}} \label{eq:idball}
\end{equation}
where the expectation is with respect to $P_{Z^\star}$, i.e. the unconditional distribution
of the reproduction random variable that achieves the infimum in \eqref{eq:RR(d)}, and
\begin{equation}
 \lambda^\star = -\mathbb R_S^\prime (d) \label{eq:lambdastar}
\end{equation}
\end{defn}
The following properties of $\mathsf d-$tilted information, proven in \cite{csiszar1974extremum}, are used in the sequel.
\begin{align}
 & \jmath_{S}(s, d) = \imath_{S; Z^\star}(s; z) + \lambda^\star \mathsf d(s, z) - \lambda^\star d \label{eq:jddensity}\\
 & \E{\jmath_{S}(s, d)} = \mathbb R_S(d) \label{eq:Ejd}\\
  &\E{ \exp\left( \lambda^\star d - \lambda^\star \mathsf d(S, z) + \jmath_S(S, d)\right) } \leq 1 \label{eq:csiszar}
\end{align}
where \eqref{eq:jddensity} holds for $P_{Z}^\star$-almost every $z$, while \eqref{eq:csiszar} holds for all $z \in \widehat {\mathcal M}$, and

\begin{equation}
\imath_{S; Z}(s; z) = \log \frac{dP_{Z|S = s}}{d P_Z}(z) \label{eq:isz}
\end{equation}
denotes the information density of the joint distribution $P_{SZ}$ at $(s, z)$. We can define the right side of \eqref{eq:isz} for a given $(P_{Z|S}, P_Z)$ even if there is no $P_S$ such that the
marginal of $P_S P_{Z|S}$ is $P_Z$. 
We use the same notation $\imath_{S;Z}$ for that more general function. 
To extend Definition \ref{defn:id} to the lossless case, for discrete random variables we define $0$-tilted information as
\begin{equation}
 \jmath_S(s, 0) = \imath_S(s) \label{eq:j0}
\end{equation}
where
\begin{equation}
 \imath_S(s) = \log \frac 1 {P_S(s)} \label{eq:i}
\end{equation}
is the information in outcome $s \in \mathcal M$. 

The distortion $d$-ball centered at $s \in \mathcal M$ is denoted by
\beq
B_d(s) = \{ z \in \widehat {\mathcal M}\colon \mathsf d(s, z) \leq d \}.
\eeq

Given $(P_X, P_{Y|X})$, we write $P_{X} \to P_{Y|X} \to P_{Y}$ to indicate that $P_Y$ is the marginal of $P_X P_{Y|X}$, i.e. $P_{Y}(y) = \sum_{x \in \mathcal X} P_{Y|X}(y|x)P_{X}(x)$. \footnote{We write summations over alphabets for simplicity. Unless stated otherwise, all our results hold for abstract probability spaces.} 
 
So as not to clutter notation, in Sections \ref{sec:C} and \ref{sec:A} we assume that there are no cost constraints. However, all results in those sections generalize to the case of a maximal cost constraint by considering $X$ whose distribution is supported on the subset of allowable channel inputs: 
\begin{equation}
\mathcal F(\alpha) = \left\{x \in \mathcal X \colon \mathsf c(x) \leq \alpha \right\} \label{eq:F}
\end{equation}
rather than the entire channel input alphabet $\mathcal X$. 
\section{Converses}
\label{sec:C}
\subsection{Converses via $\mathsf d$-tilted information}
\label{ssec:Ctilted}
Our first result is a general converse bound. 
\begin{thm}[Converse]
\label{thm:C1}
The existence of a $(d, \epsilon)$ code for $S$ and $P_{Y|X}$ requires that
\begin{align}
 \epsilon \geq   &~\inf_{P_{X|S} } \sup_{\gamma > 0} \bigg\{\sup_{P_{\bar Y}}
  \Prob{ \jmath_S(S, d) - \imath_{X; \bar Y}(X; Y) \geq \gamma} 
  \notag\\ &~
 - \exp\left( -\gamma \right) \bigg\}\label{eq:C1}\\
\geq &~\sup_{\gamma > 0}\bigg\{ \sup_{P_{\bar Y}} \E{\inf_{x \in \mathcal X}  \Prob{ \jmath_S(S, d) -  \imath_{X; \bar Y}(x; Y) \geq \gamma \mid  S} } 
  \notag\\ &~
 - \exp\left( -\gamma \right) \bigg\}\label{eq:C11}
\end{align}
where in \eqref{eq:C1}, $S - X - Y$, and the conditional probability in \eqref{eq:C11} is with respect to $Y$ distributed
according to $P_{Y|X=x}$ (independent of $S$), and
\begin{equation}
     \imath_{X; \bar Y} (x; y)  = \log \frac{dP_{Y|X = x}}{dP_{\bar Y}}(y)
\end{equation}
\end{thm}
\begin{proof}
Fix $\gamma$ and the $(d, \epsilon)$ code $(P_{X|S}, P_{Z|Y})$.  Fix an arbitrary probability measure $P_{\bar Y}$ on $\mathcal Y$. Let  $P_{\bar Y} \to P_{Z|Y} \to P_{\bar Z}$. We can write the probability in the right side of \eqref{eq:C1} as 
\begin{align}
&~ \Prob{ \jmath_S(S, d) - \imath_{X; \bar Y}(X; Y) \geq \gamma } \notag\\
=&~  \Prob{ \jmath_S(S, d) - \imath_{X; \bar Y}(X; Y) \geq \gamma, \mathsf d(S; Z) > d } \notag\\
+&~ \Prob{ \jmath_S(S, d) - \imath_{X; \bar Y}(X; Y) \geq \gamma, \mathsf d(S; Z) \leq d } \label{eq:-Cga}\\
\leq&~ \epsilon 
\notag\\
+
&~ 
\sum_{s \in \mathcal M} P_S(s) \sum_{x \in \mathcal X} P_{X|S}(x|s) \sum_{y \in \mathcal Y} \sum_{z \in B_d(s)} P_{Z|Y}(z|y) 
\notag\\
\cdot&~  P_{Y|X}(y|x) 1\left\{ P_{Y|X}(y|x)  \leq  P_{\bar Y}(y) \exp\left( \jmath_S(s, d) - \gamma\right)\right\} \\
\leq&~ \epsilon 
\notag
+ \exp\left( -\gamma \right)  \sum_{s \in \mathcal M} P_S(s) \exp\left( \jmath_S(s, d) \right)  \sum_{y \in \mathcal Y}  P_{\bar Y}(y)  
\notag\\
\cdot&~
\sum_{z \in B_d(s)} P_{Z|Y}(z|y) \sum_{x \in \mathcal X} P_{X|S}(x|s)  \\
=&~ \epsilon 
+ \exp\left( -\gamma \right) \sum_{s \in \mathcal M} P_S(s) \exp\left( \jmath_S(s, d) \right)  \sum_{y \in \mathcal Y}  P_{\bar Y}(y)  
\notag\\
\cdot&~
\sum_{z \in B_d(s)} P_{Z|Y}(z|y)   \\
=&~ \epsilon + \exp\left( -\gamma \right) \sum_{s \in \mathcal M} P_S(s) \exp\left( \jmath_S(s, d) \right)  P_{\bar Z}(B_d(s))  \\
\leq&~ \epsilon + \exp\left( -\gamma \right) \sum_{z \in \widehat {\mathcal M}} P_{\bar Z}(z)\sum_{s \in \mathcal M} P_S(s) 
\notag\\
\cdot&~
\exp\left( \jmath_S(s, d) + \lambda^\star d - \lambda^\star \mathsf d(s, z)\right) \label{eq:-Cga1}\\
\leq&~ \epsilon + \exp\left( -\gamma \right) \label{eq:-Cgb}
\end{align}
where \eqref{eq:-Cgb} is due to \eqref{eq:csiszar}. Optimizing over $\gamma > 0$ and $P_{\bar Y}$, we get the best possible bound for a given encoder $P_{X|S}$. To obtain a code-independent converse, we simply choose $P_{X|S}$ that gives the weakest bound, and \eqref{eq:C1} follows. To show \eqref{eq:C11}, we weaken \eqref{eq:C1} as
\begin{align}
 \epsilon  \geq  &~\sup_{\gamma > 0} \bigg\{ \sup_{P_{\bar Y}} \inf_{P_{X|S}} 
  \Prob{ \jmath_S(S, d) - \imath_{X; \bar Y}(X; Y) \geq \gamma} 
  \notag\\ &~
 - \exp\left( -\gamma \right) \bigg\}\label{eq:C1w}
\end{align}
and observe that for any $P_{\bar Y}$,
\begin{align}
&~\inf_{P_{X|S}} \Prob{ \jmath_S(S, d) - \imath_{X; \bar Y}(X; Y) \geq \gamma} \notag\\
= &~\sum_{s \in \mathcal M} P_S(s) \inf_{P_{X|S = s}} \sum_{x \in \mathcal X} P_{X|S}(x|s) 
\notag\\ 
\cdot&~
\sum_{y \in \mathcal Y} P_{Y|X}(y|x) 1 \left\{ \jmath_S(s, d) - \imath_{X; \bar Y}(x; y) \geq \gamma\right\}\\
=&~ \sum_{s \in \mathcal M} P_S(s) 
\notag\\
\cdot&~
\inf_{x \in \mathcal X} \sum_{y \in \mathcal Y} P_{Y|X}(y|x) 1 \left\{ \jmath_S(s, d) - \imath_{X; \bar Y}(x; y) \geq \gamma\right\} \label{eq:-C1a}\\
=&~  \E{\inf_{x \in \mathcal X}  \Prob{ \jmath_S(S, d) -  \imath_{X; \bar Y}(x; Y) \geq \gamma \mid  S}}
\end{align}
\end{proof}
\pagebreak
An immediate corollary to Theorem \ref{thm:C1} is the following result.
\begin{thm}[Converse]
Assume that there exists a distribution $P_{\bar Y}$ such that the distribution of $\imath_{X; \bar Y}(x; Y)$ (according to $P_{Y|X = x}$) does not depend on the choice of $x \in \mathcal X$. If a $(d, \epsilon)$ code for $S$ and $P_{Y|X}$ exists, then
\begin{align}
 \epsilon \geq  \sup_{\gamma > 0}\bigg\{ \Prob{ \jmath_S(S, d) -  \imath_{X; \bar Y}(x; Y) \geq \gamma } 
 - \exp\left(-\gamma\right) \bigg\}\label{eq:C1sym}
\end{align}
for an arbitrary $x \in \mathcal X$. The probability measure $\mathbb P$ in \eqref{eq:C1sym} is generated by $ P_S P_{Y|X = x} $. 
\label{thm:C1sym}
\end{thm}
\begin{proof}
Under the assumption, the conditional probability in the right side of \eqref{eq:C11} is the same regardless of the choice of $x \in \mathcal X$. 
\end{proof}

The next result generalizes Theorem \ref{thm:C1}. When we apply Theorem \ref{thm:CT} in Section \ref{sec:2order} to find the dispersion of JSCC, we will let $T$ be the number of channel input types, and we will let $W$ be the type of the channel input block. If $T = 1$, Theorem \ref{thm:CT} reduces to Theorem \ref{thm:C1}. 

\begin{thm}[Converse]
\label{thm:CT}
The existence of a $(d, \epsilon)$ code for $S$ and $P_{Y|X}$ requires that
 \begin{align}
 \epsilon \geq
 &~  \inf_{P_{X|S}} \max_{\gamma > 0, T} \bigg\{ - T \exp\left( -\gamma \right) 
 \notag \\
 + &~
\sup_{
   \substack{
 \bar Y, W \colon\\
S - (X,W) - Y
}
} 
\mathbb P \big [ \jmath_S(S, d) - \imath_{X; \bar Y | W}(X; Y | W) \geq \gamma \big ]
 \bigg\}\label{eq:CTa}\\
   \geq &~\max_{\gamma > 0, T}\bigg\{ - T\exp\left( -\gamma \right) 
\notag \\
+ &~
\sup_{
\bar Y, W
}    
  \mathbb E \Big[ \inf_{x \in \mathcal X}  
  \mathbb P \big[ \jmath_S(S, d)  -  \imath_{X; \bar Y |W}(x; Y | W) \geq \gamma \mid S\big] \Big]
 \bigg\}
 \label{eq:CT}
\end{align}
where $T$ is a positive integer, the random variable $W$ takes values on 
$\left\{1, \ldots, T\right\}$, 
and
\begin{equation}
 \imath_{X; \bar Y | W}(x; y | t) = \log \frac{ P_{Y|X = x, W = t}}{P_{\bar Y | W = t}} (y)
\end{equation}
and in \eqref{eq:CT}, the probability measure is generated by $P_S P_{W|X = x} P_{Y|X = x, W}$.
\end{thm}

\begin{proof}
Fix a possibly randomized $(d, \epsilon)$ code $\{P_{X|S}, P_{Z|Y}\}$, a positive scalar $\gamma$, a positive integer $T$, an auxiliary random variable $W$ that satisfies $S - (X, W) - Y$, and a conditional probability distribution $P_{\bar Y | W} \colon \{1, \ldots T\} \mapsto \mathcal Y$. Let  
$P_{\bar Y|W = t} \to P_{Z|Y} \to P_{\bar Z| W = t}$, i.e. 
$P_{\bar Z|W = t}(z) = \sum_{y \in \mathcal Y} P_{Z|Y}(z|y)P_{\bar Y|W = t}(y)$, for all $t$. Write
\begin{align}
&~ \Prob{ \jmath_S(S, d) - \imath_{X; Y | W }(X; Y | W ) \geq \gamma } \notag\\
\leq&~ \epsilon 
+
\sum_{s \in \mathcal M} P_S(s) \sum_{t = 1}^T P_{W|S}(t|s) 
\sum_{x \in \mathcal X} P_{X|S, W}(x|s, t) 
\notag\\
\cdot&~
\sum_{y \in \mathcal Y} P_{Y|X, W}(y|x, t) \sum_{z \in B_d(s)} P_{Z|Y}(z|y) \notag\\
\cdot&~   1\left\{ P_{Y|X, W}(y|x, t)  \leq  P_{\bar Y | W = t}(y) \exp\left( \jmath_S(s, d) - \gamma\right)\right\} \\
\leq&~ \epsilon 
+ \exp\left( -\gamma \right) \sum_{s \in \mathcal M} P_S(s) \exp\left( \jmath_S(s, d) \right) \sum_{t = 1}^T P_{W|S}(t|s)   
\notag\\
\cdot&~
 \sum_{y \in \mathcal Y}  P_{\bar Y | W}(y|t)  
\sum_{z \in B_d(s)} P_{Z|Y}(z|y) \sum_{x \in \mathcal X}  P_{X|S, W}(x|s, t)  \\
\leq&~ \epsilon 
+ \exp\left( -\gamma \right) \sum_{t = 1}^T \sum_{s \in \mathcal M} P_S(s) \exp\left( \jmath_S(s, d) \right)  
  \sum_{y \in \mathcal Y}  P_{\bar Y | W}(y|t) 
\notag\\
\cdot&~
\sum_{z \in B_d(s)} P_{Z|Y}(z|y)   \\
\leq&~ \epsilon + \exp\left( -\gamma \right) \sum_{t = 1}^T \sum_{s \in \mathcal M} P_S(s) \exp\left( \jmath_S(s, d) \right)  P_{\bar Z|W = t}(B_d(s)) \\
\leq&~ \epsilon + \exp\left( -\gamma \right) \sum_{t = 1}^T \sum_{s \in \mathcal M} P_S(s) \sum_{z \in \widehat {\mathcal M}} P_{\bar Z | W= t}(z)
\notag\\
\cdot&~
\exp\left( \jmath_S(s, d) + \lambda^\star d - \lambda^\star \mathsf d(s, z)\right) \label{eq:-CT}\\
\leq&~ \epsilon + T\exp\left( -\gamma \right) \label{eq:-CTa}
\end{align}
where \eqref{eq:-CTa} is due to \eqref{eq:csiszar}. Optimizing over $\gamma$, $T$ and the distributions of the auxiliary random variables $\bar Y$ and $W$, we obtain the best possible bound for a given encoder $P_{X|S}$. To obtain a code-independent converse, we simply choose $P_{X|S}$ that gives the weakest bound, and \eqref{eq:CTa} follows. To show \eqref{eq:CT}, we weaken \eqref{eq:CTa} by restricting the $\sup$ to $W$ satisfying $S - X - W$ and changing the order of $\inf$ and $\sup$ as follows:
\begin{equation}
 \max_{\gamma > 0, T}  \sup_{
   \substack{
\bar Y,  W \colon\\
S - (X,W) - Y\\
S - X - W
}
}   \inf_{P_{X|S}} 
\end{equation}
Observe that for any legitimate choice of $\bar Y$ and $W$,
\begin{align}
&~\inf_{P_{X|S}} \Prob{ \jmath_S(S, d) - \imath_{X; \bar Y | W}(X; Y | W) \geq \gamma}\\
= &~\sum_{s \in \mathcal M} P_S(s)  \inf_{P_{X|S = s}} \sum_{x \in \mathcal X} P_{X|S}(x|s) \sum_{t = 1}^T P_{W|X}(t|x)
\notag\\
\cdot &~
 \sum_{y \in \mathcal Y} P_{Y|X, W}(y|x, t) 1 \left\{ \jmath_S(s, d) - \imath_{X; \bar Y | W}(x; y | t) \geq \gamma\right\}\\
=&~ \sum_{s \in \mathcal M} P_S(s) \inf_{x \in \mathcal X} \sum_{t = 1}^T P_{W|X}(t|x) \sum_{y \in \mathcal Y} P_{Y|X, W}(y|x, t) 
\notag \\
\cdot &~
1 \left\{ \jmath_S(s, d) - \imath_{X; \bar Y |W }(x; y | t) \geq \gamma\right\} \label{eq:-C1a}
\end{align}
which is equal to the expectation on the right side of \eqref{eq:CT}. 
\end{proof}

\begin{remark}
Theorems \ref{thm:C1}, \ref{thm:C1sym} and \ref{thm:CT} still hold in the case  $d = 0$ and $d(x, y) = 1 \left\{ x \neq y \right\}$, which corresponds to almost-lossless data compression. Indeed, recalling \eqref{eq:j0}, it is easy to see that the proof of Theorem \ref{thm:C1} applies, skipping the now unnecessary step \eqref{eq:-Cga1}, and, therefore, \eqref{eq:C1} reduces to 
\begin{align}
 \epsilon \geq  \inf_{P_{X|S} } \sup_{\gamma > 0} \bigg\{&~ \sup_{P_{\bar Y}}
  \Prob{ \imath_S(S) - \imath_{X; \bar Y}(X; Y) \geq \gamma} 
  \notag\\ &~
 - \exp\left( -\gamma \right) \bigg\}\label{eq:C1lossless}
\end{align}
Similar modification can be applied to the proof of Theorem \ref{thm:CT}. 
\label{rem_Clossless}
\end{remark}
\begin{remark}
Our converse for lossy source coding in \cite[Theorem 7]{kostina2011fixed} can be viewed as a particular case of the result in Theorem \ref{thm:C1sym}. Indeed, if $\mathcal X = \mathcal Y = \{1, \ldots, M\}$ and $P_{Y|X}(m|m) = 1$, $P_Y(1) = \ldots = P_Y(M) = \frac 1 M$, then \eqref{eq:C1sym} becomes
\begin{equation}
 \epsilon \geq \sup_{\gamma>0}  \Prob{  \jmath_S(S, d) \geq \log M + \gamma } - \exp\left(-\gamma\right)
\end{equation}
which is precisely \cite[Theorem 7]{kostina2011fixed}. 
\end{remark}

\subsection{Converses via hypothesis testing and list decoding}
\label{sec:Cht}

To show a joint source-channel converse in  \cite{csiszar1982error}, Csisz{\'a}r used a list decoder, which outputs a list of $L$ elements drawn from $\mathcal M$. While traditionally list decoding has only been considered in the context of finite alphabet sources, we generalize the setting to sources with abstract alphabets. 
In our setup, the encoder is the random transformation $P_{X|S}$, and the decoder is defined as follows.
\begin{defn}[List decoder]
Let $L$ be a positive real number, and let $Q_S$ be a measure on $\mathcal M$. An $(L, Q_S)$ list decoder is a random transformation $P_{\tilde S|Y}$, where $\tilde S$ takes values on $Q_S$-measurable sets with $Q_S$-measure not exceeding $L$:
\begin{equation}
 Q_S\left(\tilde S\right) \leq L \label{eq:listsize}
\end{equation}
\end{defn}
Even though we keep the standard ``list'' terminology, the decoder output need not be a finite or countably infinite set. The error probability with this type of list decoding is the probability that the source outcome $S$ does not belong to the decoder output list for $Y$:
\begin{align}
1 - \sum_{x \in \mathcal X} \sum_{y \in \mathcal Y} \sum_{\tilde s \in \mathcal M^{(L)} }  \sum_{s \in \tilde s} P_{\tilde S|Y} (\tilde s | y) P_{Y|X}(y|x) 
P_{X|S}(x|s) P_S(s) \label{eq:listerror}
\end{align}
where $\mathcal M^{(L)}$ is the set of all $Q_S$-measurable subsets of $\mathcal M$ with $Q_S$-measure not exceeding $L$.

\begin{defn}[List code]
 An $(\epsilon, L, Q_S)$ list code is a pair of random transformations $(P_{X|S}, P_{\tilde S|Y})$ such that \eqref{eq:listsize} holds and the list error probability \eqref{eq:listerror} does not exceed $\epsilon$. 
\end{defn}
Of course, letting $Q_S = U_S$, where $U_S$ is the counting measure on $\mathcal M$, we recover the conventional list decoder definition where the smallest scalar that satisfies \eqref{eq:listsize} is an integer. The almost-lossless JSCC setting ($d = 0$) in Definition \ref{defn:jscc} corresponds to $L = 1$, $Q_S = U_S$. If the source is analog (has a continuous distribution), it is reasonable to let $Q_S$ be the Lebesgue measure.

Any converse for list decoding implies a converse for conventional decoding. To see why, observe that any $(d, \epsilon)$ lossy code can be converted to a list code with list error probability not exceeding $\epsilon$ by feeding the lossy decoder output to a function that outputs the set of all source outcomes $s$ within distortion $d$ from the output $z \in \widehat {\mathcal M}$ of the original lossy decoder. In this sense, the set of all $(d, \epsilon)$ lossy codes is included in the set of all list codes with list error probability $\leq \epsilon$ and list size
\begin{equation}
 L = \max_{z \in \widehat {\mathcal M}} Q_S\left( \left\{ s \colon \mathsf d(s, z) \leq d \right\}\right) \label{eq:lossylistsize} 
\end{equation}

Denote by
\begin{equation}
\label{eq:beta}
\beta_{\alpha}(P, Q) = \min_{\substack{P_{W|X}\colon \\ \Prob{W = 1} \geq \alpha}} \mathbb Q \left[ W = 1\right]
\end{equation}
the optimal performance achievable among all randomized tests $P_{W|X}\colon \mathcal X \rightarrow \left\{ 0, 1\right\}$ between probability distributions $P$ and $Q$ on $\mathcal X$ ($1$ indicates that the test chooses $P$).\footnote{Throughout, $P$, $Q$ denote distributions, whereas $\mathbb P$, $\mathbb Q$ are used for the corresponding probabilities of events on the underlying probability space.} In fact, $Q$ need not be a probability measure, it just needs to be $\sigma$-finite in order for the Neyman-Pearson lemma and related results to hold. 

The hypothesis testing converse for channel coding \cite[Theorem 27]{polyanskiy2010channel} can be generalized to joint source-channel coding with list decoding as follows.

\begin{thm}[Converse]
Fix $P_S$ and $P_{Y|X}$, and let $Q_{S}$ be a $\sigma$-finite measure. The existence of an $(\epsilon, L, Q_S)$ list code requires that
 \beq
 \inf_{P_{X|S}} \sup_{P_{\bar Y}}\beta_{1 - \epsilon} (P_S P_{X|S} P_{Y|X}, Q_S P_{X|S} P_{\bar Y}) \leq L  \label{eq:Cht}
 \eeq
  where the supremum is over all probability measures $P_{\bar Y}$ defined on the channel output alphabet $\mathcal Y$.
 \label{thm:Cht}
 \end{thm}
\begin{proof}
Fix $Q_S$, the encoder $P_{X|S}$, and an auxiliary $\sigma$-finite conditional measure $Q_{Y|XS}$. Consider the (not necessarily optimal) test for deciding between $P_{SXY} = P_S P_{X|S} P_{Y|X}$ and $Q_{SXY} = Q_S P_{X|S} Q_{Y|XS}$ which chooses $P_{SXY}$ if $S$ belongs to the decoder output list. Note that this is a hypothetical test, which has access to both the source outcome and the decoder output. 

According to $\mathbb P$, the probability measure generated by $P_{SXY}$, the probability that the test chooses $P_{SXY}$ is given by
\begin{equation}
  \Prob{S \in \tilde S} \geq 1 - \epsilon
  \end{equation}
Since $\Qrob{ S \in \tilde S }$ is the measure of the event that the test chooses $P_{SXY}$ when $Q_{SXY}$ is true, and the optimal test cannot perform worse than the possibly suboptimal one that we selected, it follows that
  \beq
  \beta_{1 - \epsilon} (P_S P_{X|S} P_{Y|X}, Q_S P_{X|S} Q_{Y|XS}) \leq \Qrob{ S \in \tilde S } \label{eq:Cmeta}
  \eeq
Now, fix an arbitrary probability measure $P_{\bar Y}$ on $\mathcal Y$. Choosing $Q_{Y|XS} = P_{\bar Y}$, the inequality in \eqref{eq:Cmeta} can be weakened as follows. 
\begin{align}
  &~
  \Qrob{ S \in \tilde S } 
  \notag\\
  =
  &~
    \sum_{y \in \mathcal Y} P_{\bar Y}(y) \sum_{\tilde s \in \mathcal M^{(L)}} P_{\tilde S|Y} ( \tilde s | y) \sum_{s \in \tilde s} Q_S(s)\sum_{x \in \mathcal X} P_{X|S}(x|s)\\
  =&~  \sum_{y \in \mathcal Y} P_{\bar Y}(y) \sum_{\tilde s \in \mathcal M^{(L)}} P_{\tilde S|Y} ( \tilde s | y) \sum_{s \in \tilde s} Q_S(s)\\
  \leq&~ \sum_{y \in \mathcal Y} P_{\bar Y}(y) \sum_{\tilde s \in \mathcal M^{(L)}} P_{\tilde S|Y} ( \tilde s | y) L\\
  =&~ L
\end{align}
Optimizing the bound over $P_{\bar Y}$ and choosing $P_{X|S}$ that yields the weakest bound in order to obtain a code-independent converse,  \eqref{eq:Cht} follows. 
\end{proof}

\begin{remark}
Similar to how Wolfowitz's converse for channel coding can be obtained from the meta-converse for channel coding \cite{polyanskiy2010channel}, the converse for almost-lossless joint source-channel coding in \eqref{eq:C1lossless} can be obtained by appropriately weakening \eqref{eq:Cht} with $L = 1$. Indeed, invoking \cite{polyanskiy2010channel} 
\beq
\beta_{\alpha}(P, Q) \geq \frac 1 {\gamma} \left( \alpha - \Prob{\frac{dP}{dQ} > \gamma}\right)
\eeq
and letting $Q_S = U_S$ in \eqref{eq:Cht}, where $U_S$ is the counting measure on $\mathcal M$, we have
\begin{align}
1 &\geq \inf_{P_{X|S}} \sup_{P_{\bar Y}}\beta_{1 - \epsilon} (P_S P_{X|S} P_{Y|X}, U_S P_{X|S} P_{\bar Y})\label{eq:-Cmeta}\\
&\geq \inf_{P_{X|S}} \sup_{P_{\bar Y}}\sup_{\gamma > 0}\frac 1 {\gamma} \left( 1 - \epsilon - \Prob{  \imath_{X; \bar Y}(X; Y)\! -\! \imath_S(S) \!>\! \log \gamma}\right)
\end{align}
which upon rearranging yields \eqref{eq:C1lossless}. 
\end{remark}

In general, computing the infimum in \eqref{eq:Cht} is challenging. However, if the channel is symmetric (in a sense formalized in the next result), $\beta_{1 - \epsilon} (P_S P_{X|S} P_{Y|X}, U_S P_{X|S} P_{\bar Y})$ is independent of $P_{X|S}$.

  \begin{thm}[Converse]
 Fix a probability measure $P_{\bar Y}$. Assume that the distribution of  $\imath_{X; \bar Y}(x; Y)$ does not depend on $x \in \mathcal X$ under either $P_{Y|X = x}$ or $P_{\bar Y}$. 
  Then, the existence of an $(\epsilon, L, Q_S)$ list code requires that
  \begin{equation}
  \beta_{1 - \epsilon} (P_S P_{Y|X = x}, Q_S P_{\bar Y}) \leq L \label{eq:Chtsym}
  \end{equation}
  where $x \in \mathcal X$ is arbitrary.
  \label{thm:Chtsym}
  \end{thm}

\begin{proof}
The Neyman-Pearson lemma (e.g. \cite{poor1994introduction}) implies that the outcome of the optimum binary hypothesis test between $P$ and $Q$ only depends on the observation through $\frac {dP}{dQ}$. In particular, the optimum binary hypothesis test $W^\star$ for deciding between $P_S P_{X|S} P_{Y|X}$ and $Q_S P_{X|S} P_{\bar Y}$ satisfies
\begin{equation}
W^\star - (S,  \imath_{X; \bar Y}(X; Y) ) - (S, X, Y) \label{eq:-CsymW}
\end{equation}
For all $s \in \mathcal M$, $x \in \mathcal X$, we have
   \begin{align}
   &~\Prob{W^\star = 1|  S = s, X = x}
   \notag \\
   =&~  \E{\Prob{W^\star = 1| X = x, S = s, Y}} \\
   =&~  \E{\Prob{W^\star = 1|  S = s,  \imath_{X; \bar Y}(X; Y) = \imath_{X; \bar Y}(x; Y)}}\label{eq:-Csym1} \\
   =&~ \sum_{y \in \mathcal Y} P_{Y|X}(y|x) P_{W^\star|S, \, \imath_{X; \bar Y}(X; Y)}(1 |s, \imath_{X; \bar Y}(x; y) ) \label{eq:-Csym2}\\
   =&~ \Prob{W^\star = 1|  S = s} \label{eq:-Csym3}
   \end{align}
   and
   \begin{equation}
    \mathbb Q\left[W^\star = 1|  S = s, X = x\right] = \mathbb Q\left[W^\star = 1|  S = s \right] \label{eq:-Csym3a}
   \end{equation}
   where
\begin{itemize}
 \item \eqref{eq:-Csym1} is due to \eqref{eq:-CsymW},
 \item \eqref{eq:-Csym2} uses the Markov property $S - X - Y$,
 \item \eqref{eq:-Csym3} follows from the symmetry assumption on the distribution of $\imath_{X; \bar Y}(x, Y)$,
  \item \eqref{eq:-Csym3a} is obtained similarly to \eqref{eq:-Csym2}.
\end{itemize}
 Since \eqref{eq:-Csym3}, \eqref{eq:-Csym3a} imply that the optimal test achieves the same performance (that is, the same $\Prob{W^\star = 1}$ and $\mathbb Q\left[W^\star = 1\right]$) regardless of $P_{X|S}$, we choose $P_{X|S} = 1_X(x)$ for some $x \in\mathcal X$ in the left side of \eqref{eq:Cht} to obtain \eqref{eq:Chtsym}. 
\end{proof}
\begin{remark}
 In the case of finite channel input and output alphabets, the channel symmetry assumption of Theorem \ref{thm:Chtsym} holds, in particular, if the rows of the channel transition probability matrix are permutations of each other, and $P_{\bar Y^n}$ is the equiprobable distribution on the ($n$-dimensional) channel output alphabet,  which, coincidentally, is also the capacity-achieving output distribution. For Gaussian channels with equal power constraint, which corresponds to requiring the channel inputs to lie on the power sphere, any spherically-symmetric $P_{\bar Y^n}$ satisfies the assumption of Theorem \ref{thm:Chtsym}. 
\end{remark}

\section{Achievability}
\label{sec:A}
Given a source code $(\mathsf f_{\mathrm s}^{(M)}, \mathsf g_{\mathrm s}^{(M)})$ of size $M$, and a channel code  $(\mathsf f_{\mathrm c}^{(M)}, \mathsf g_{\mathrm c}^{(M)})$ of size $M$, we may concatenate them to obtain the following sub-class of the source-channel codes introduced in Definition \ref{defn:jscc}:
\begin{defn}
 An $(M, d, \epsilon)$ source-channel code is a $(d, \epsilon)$ source-channel code such that the encoder and decoder mappings satisfy
\begin{align}
 \mathsf f &= \mathsf f_{\mathrm c}^{(M)} \circ \mathsf f_{\mathrm s}^{(M)}\\
 \mathsf g &= \mathsf g_{\mathrm c}^{(M)} \circ \mathsf g_{\mathrm s}^{(M)}
\end{align}
where
\begin{align}
 \mathsf f_{\mathrm s}^{(M)} &\colon \mathcal M \mapsto \left\{1, \ldots, M \right\}\\
 \mathsf f_{\mathrm c}^{(M)} &\colon \left\{1, \ldots, M \right\} \mapsto \mathcal X\\
 \mathsf g_{\mathrm c}^{(M)} &\colon \mathcal Y \mapsto \left\{1, \ldots, M \right\}\\
 \mathsf g_{\mathrm s}^{(M)} &\colon \left\{1, \ldots, M \right\} \mapsto \widehat {\mathcal M}
\end{align}
(see Fig. \ref{fig:Mjscc}).
\label{defn:jsccM}
\end{defn}
\begin{figure}[htbp]
\begin{center}
    \epsfig{file=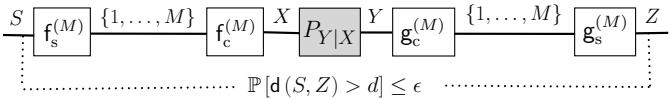,width= 1\linewidth}
\caption{An $(M, d, \epsilon)$ joint source-channel code. }
\label{fig:Mjscc}
\end{center}
\end{figure}
Note that an $(M, d, \epsilon)$ code is an $(M+1, d, \epsilon)$ code. 

The conventional separate source-channel coding paradigm corresponds to the special case of Definition \ref{defn:jsccM} in which the source code $(\mathsf f_{\mathrm s}^{(M)}, \mathsf g_{\mathrm s}^{(M)})$ is chosen without knowledge of $P_{Y|X}$ and the channel code $(\mathsf f_{\mathrm c}^{(M)}, \mathsf g_{\mathrm c}^{(M)})$ is chosen without knowledge of $P_{S}$ and the distortion measure $\mathsf d$. A pair of source and channel codes is separation-optimal if the source code is chosen so as to minimize the distortion (average or excess) when there is no channel, whereas the channel code is chosen so as to minimize the worst-case (over source distributions) average error probability:
\begin{equation}
 \max_{P_U} \Prob{ U \neq \mathsf g_{\mathrm c}^{(M)}(Y)} \label{eq:worstcaseerror}
\end{equation}
where $X = \mathsf f_{\mathrm c}^{(M)}(U)$ and $U$ takes values on $\left\{1, \ldots, M\right\}$. 
If both the source and the channel code are chosen separation-optimally for their given sizes, the separation principle guarantees that under certain quite general conditions (which encompass the memoryless setting, see \cite{vembu1995sourcechannel}) the asymptotic fundamental limit of joint source-channel coding is achievable. In the finite blocklength regime, however, such SSCC construction is, in general, only suboptimal. Within the SSCC paradigm, we can obtain an achievability result by further optimizing with respect to the choice of $M$:

\begin{thm}[Achievability, SSCC]
Fix $P_{Y|X}$, $\mathsf{d}$ and $P_S$.
Denote by $\epsilon^\star ( M )$ the minimum achievable worst-case average error probability among all transmission codes of size $M$, and the minimum achievable probability of exceeding distortion $d$ with a source code of size $M$ by
$\epsilon^\star ( M, d)$.

Then, there exists a $(d, \epsilon)$ source-channel code with
\begin{equation}
\epsilon \leq \min_M \{ \epsilon^\star (M) +\epsilon^\star ( M, d) \} \label{eq:ASeparate}
\end{equation}
\end{thm}
Bounds on $\epsilon^\star(M)$ and $\epsilon^\star ( M, d)$ have been obtained recently in \cite{polyanskiy2010channel} and \cite{kostina2011fixed}, respectively.\footnote{As the maximal (over source outputs) error probability cannot be lower than the worst-case error probability, the maximal error probability achievability bounds of \cite{polyanskiy2010channel} apply to bound $\epsilon^\star ( M )$. Moreover, the random coding union (RCU) bound on average error probability of \cite{polyanskiy2010channel}, although stated assuming equiprobable source, is oblivious to the distribution of the source and thus upper-bounds the worst-case average error probability 
$\epsilon^\star ( M )$ as well.  }

Definition \ref{defn:jsccM} does not rule out choosing the source code based on the knowledge of $P_{Y|X}$ or the channel code based on the knowledge of $P_S$, $\mathsf d$ and $d$. 
One of the interesting conclusions in the present paper is that the optimal dispersion of JSCC is achievable within the class of $(M, d, \epsilon)$ source-channel codes introduced in Definition \ref{defn:jsccM}. However, the dispersion achieved by the conventional SSCC approach is in fact suboptimal. 

To shed light on the reason behind the suboptimality of SSCC at finite blocklength despite its asymptotic optimality, we recall the reason SSCC achieves the asymptotic fundamental limit. The output of the optimum source encoder is, for large $k$,  approximately equiprobable over a set of roughly $\exp\left(k R(d)\right)$ distinct messages, which allow to represent most of the source outcomes within distortion $d$.  From the channel coding theorem we know that there exists a channel code that is capable of distinguishing, with high probability, $M = \exp\left(k R(d)\right) < \exp\left(n C\right)$ messages when equipped with the maximum likelihood decoder. Therefore, a simple concatenation of the source code and the channel code achieves vanishing probability of distortion exceeding $d$, for any $d > D\left(\frac{nC}{k}\right)$.  
However, at finite $n$, the output of the optimum source encoder need not be nearly equiprobable, so there is no reason to expect that a separated scheme employing a  maximum-likelihood channel decoder, which does not exploit unequal message probabilities, would achieve near-optimal non-asymptotic performance. Indeed, in the non-asymptotic regime the gain afforded by taking into account the residual encoded source redundancy at the channel decoder is appreciable. The following achievability result, obtained using independent random source codes and random channel codes within the paradigm of Definition \ref{defn:jsccM}, capitalizes on this intuition.

\begin{thm}[Achievability]
There exists a $(d, \epsilon)$ source-channel code with
\begin{align}
\epsilon \leq \inf_{P_{X}, P_{Z}, P_{W|S}} \bigg\{&
\E{ \exp\left( - \left| \imath_{X; Y}(X; Y )  - \log W
\right|^+ \right)}
\notag\\&
+  \E{  \left(1 - P_{Z}(B_d(S)) \right)^{W} } \bigg\}\label{eq:Ag}
\end{align}
where the expectations are with respect to 
$P_S P_{X} P_{Y|X} P_{Z} P_{W|S}$
defined on 
$\mathcal M \times \mathcal X \times \mathcal Y \times \widehat {\mathcal M} \times \mathbb N$, where $\mathbb N$ is the set of natural numbers.
\label{thm:Ag}
\end{thm}
\begin{proof}
Fix a positive integer $M$.  Fix a positive integer-valued random variable $W$ that depends on other random variables only through $S$ and that satisfies $ W \leq M$. We will construct a code with separate
encoders for source and channel and separate
decoders for source and channel as in Definition \ref{defn:jsccM}. 
We will perform a random coding analysis by choosing random independent source and channel codes which will lead to the conclusion that there exists an $(M, d, \epsilon)$ code with error probability $\epsilon$ guaranteed in \eqref{eq:Ag} with $W \leq M$. Observing that increasing $M$ can only tighten the bound in \eqref{eq:Ag} in which $W$ is restricted to not exceed $M$, we will let $M \to \infty$ and conclude, by invoking the bounded convergence theorem, that the support of $W$ in \eqref{eq:Ag} need not be bounded. 
\par
\textit{Source Encoder.}
Given an ordered list of
representation points 
$z^M = (z_1, \ldots , z_M ) \in \widehat {\mathcal M}^M$,
and having observed the source outcome $s$, the (probabilistic) source encoder generates $W$ from $P_{W|S = s}$ and
selects the lowest index 
$m \in \{1, \ldots , W\}$
such that  $s$ is within distance $d$ of $z_m$. If no such index can be found, the source encoder outputs a pre-selected arbitrary index, e.g. $M$. Therefore,
\begin{equation}
 \mathsf f_{\mathrm s}^{(M)}(s) = 
\begin{cases}
 \min\{m, W\} & \mathsf d(s, z_m) \leq d < \displaystyle{\min_{i = 1, \ldots, m -1} }\mathsf d(s, z_i)\\
M & d < \min_{i = 1, \ldots, W} \mathsf d(s, z_i)
\end{cases}
\label{eq:-Afs}
\end{equation}
In a good $(M, d, \epsilon)$ JSCC code, $M$ would be chosen so large that with overwhelming probability, a source outcome would be encoded successfully within distortion $d$. It might seem counterproductive to let the source encoder in \eqref{eq:-Afs} give up before reaching the end of the list of representation points, but in fact, such behavior helps the channel decoder by skewing the distribution of $\mathsf f_{\mathrm s}^{(M)}(S)$. 
\par
\textit{Channel Encoder.}
Given a codebook $(x_1, \ldots, x_M) \in \mathcal{X}^M $,
the channel encoder outputs $x_m$ if
$m$ is the output of the source encoder: 
\begin{equation}
 \mathsf f_{\mathrm c}^{(M)}(m) = x_m
\end{equation}
\par
\textit{Channel Decoder.}
Define the random variable $U \in \left\{ 1, \ldots, M+1 \right\}$ which is a function of $S$, $W$ and $z^{M}$ only:
\begin{equation}
 U =
\begin{cases}
  \mathsf f_{\mathrm s}^{(M)}(S) & \mathsf d(S, \mathsf g_{\mathrm s} (\mathsf f_{\mathrm s}(S) ) \leq d\\
  M+1 & \text{otherwise}
\end{cases}
\end{equation}
Having observed $y \in \mathcal Y$, the channel decoder chooses arbitrarily among the members of the set:
\beq
\mathsf g_{\mathrm c}^{(M)}(y) = m \in \arg\max_{j \in \{1, \ldots M\}} P_{U|Z^M}(j | z^M) P_{Y|X}(y|x_j) \label{eq:MAPapprox}
\eeq
A MAP decoder would multiply $P_{Y|X}(y|x_j)$ by $P_{X}(x_j)$. While that decoder would be too hard to analyze, the product in \eqref{eq:MAPapprox} is a good approximation because $P_{U|Z^M}(j|z^M)$ and $P_X(x_j)$ are related by
\begin{align}
P_X(x_j) 
&= \sum_{m \colon x_m = x_j} P_{U|Z^M}(m|z^M)
\notag \\ 
 &+ P_{U|Z^M}(M+1|z^M) 1\left\{ j = M\right\}
\end{align}
so the decoder in \eqref{eq:MAPapprox} differs from a MAP decoder only when either several $x_m$ are identical, or there is no representation point among the first $W$ points within distortion $d$ of the source, both unusual events. 
\par
\textit{Source Decoder.} The source decoder outputs $z_{m}$ if $m$ is the output of the channel decoder:
\begin{equation}
\mathsf g_{\mathrm s}^{(M)} (m) = z_m 
\end{equation}
\par
\textit{Error Probability Analysis.}
We now proceed to analyze the performance of the code described above.
If there were no source encoding error, a channel decoding error can occur if and only if
\begin{align}
&~\exists j \neq m \colon 
\notag \\
&~ P_{U|Z^M}(j | z^M) P_{Y|X}(Y|x_j) \geq  P_{U|Z^M}(m | z^M) P_{Y|X}(Y|x_{m}) 
\end{align}
 Let the channel codebook $(X_1, \ldots, X_{M})$ be drawn i.i.d. from $P_X$, and independent of the source codebook $(Z_1, \ldots, Z_{M}) $, which is drawn i.i.d. from $P_{Z}$.  Denote by $\epsilon(x^{M}, z^{M})$ the excess-distortion probability attained with the source codebook $z^{M}$ and the channel codebook $x^{M}$. 
 Conditioned on the event $\left\{\mathsf d(S, \mathsf g_{\mathrm s} (\mathsf f_{\mathrm s}(S) ) \leq d\right\} = \left\{ U \leq W \right\} = \left\{ U \neq M + 1\right\} $ (no failure at the source encoder), the probability of excess distortion is upper bounded by the probability that the channel decoder does not choose $\mathsf f_{\mathrm s}^{(M)}(S)$, so
\begin{align}
&~\epsilon(x^{M}, z^{M})  
\notag \\
\leq
&~
\sum_{m = 1}^M P_{U|Z^M} \left(m | z^m\right)
\notag \\
\cdot &~\Prob{\bigcup_{j \neq m} \left\{\frac{ P_{U|Z^M}(j|z^M) P_{Y|X}(Y|x_j)}{P_{U|Z^M}(m|z^M) P_{Y|X}(Y|x_m)} \geq 1 \right\}\mid X = x_m} 
\notag\\&
+ P_{U | Z^{M} }(U > W | z^M)  \label{eq:-Aa}
\end{align}
We now average \eqref{eq:-Aa} over the source and channel codebooks. 
Averaging the $m$-th term of the sum in \eqref{eq:-Aa} with respect to the channel codebook yields
 \beq
P_{U|Z^M} \left(m | z^m\right) \Prob{\bigcup_{j \neq m} \left\{\! \frac{P_{U|Z^M}(j | z^M) P_{Y|X}(Y|X_j)}{P_{U|Z^M}(m |z^M) P_{Y|X}(Y|X_m)} \geq 1 \!\right\}\!} \label{eq:-Ab}
 \eeq
 where $Y, X_1, \ldots, X_M$ are distributed according to
\begin{equation}
 P_{Y X_1 \ldots X_m}(y, x_1, \ldots, x_M) = P_{Y|X_m}(y|x_m) \prod_{j \neq m} P_X(x_j)
\end{equation}
 
  Letting $\bar X$ be an independent copy of $X$ and applying the union bound to the probability in \eqref{eq:-Ab}, we have that for any given $(m, z^M)$, 
 \begin{align}
 &~\Prob{\bigcup_{j \neq m} \left\{\frac{P_{U | Z^M}(j | z^M) P_{Y|X}(Y|X_j)}{P_{U | Z^M}(m | z^M) P_{Y|X}(Y|X_m)} \geq 1\right\} } \notag\\
 \leq&~ 
 \mathbb E \Bigg[ \min\Bigg\{ 1, 
 \notag\\
 &~ 
 \sum_{j = 1}^M \Prob{\frac{P_{U | Z^M}(j | z^M) P_{Y|X}(Y|\bar X)}{P_{U | Z^M}(m | z^M) P_{Y|X}(Y|X)} \geq 1 \mid X, Y}\Bigg\}
 \Bigg]\\
 \leq&~ \E{ \min\left\{ 1, \sum_{j = 1}^M \frac{P_{U | Z^M}(j | z^M)}{P_{U | Z^M}(m | z^M)} \frac{\E{ P_{Y|X}(Y|\bar X)|Y}}{ P_{Y|X}(Y|X)}\right\} } \label{eq:-Ac}\\
 =&~  \E{ \min\left\{ 1, \sum_{j = 1}^M \frac{P_{U | Z^M}(j | z^M)}{P_{U | Z^M}(m | z^M)} \frac{P_Y(Y)}{ P_{Y|X}(Y|X)}\right\}} \\
 =&~ \E{ \min\left\{ 1, \frac{ \Prob{U \leq W \mid Z^M = z^M}}{P_{U | Z^M}(m | z^M)} \frac{P_{Y}(Y)}{P_{Y|X}(Y|X)}\right\}} \label{eq:-Ach1}\\
  =&~ \E{ \min\left\{ 1, \frac{1}{P_{U | Z^M, 1\left\{U \leq W\right\} }(m | z^M, 1)} 
  \frac{P_{Y}(Y)}{P_{Y|X}(Y|X)}\right\}
  } \label{eq:-Ach}
 \end{align}
 where \eqref{eq:-Ac} is due to $1\{a \geq 1\} \leq a$.
 
 Applying \eqref{eq:-Ach} to \eqref{eq:-Aa} and averaging with respect to the source codebook, we may write
\begin{align}
&~\E{ \epsilon(X^{M}, Z^{M})} 
\leq
\E{ \min\left\{ 1, G \right\} 
} 
+
\Prob{U > W }  \label{eq:-Abasic}
\end{align}
where for brevity we denoted the random variable
\begin{equation}
 G = \frac{1}{P_{U | Z^M, 1\left\{U \leq W\right\} }(U | Z^M, 1)}   \frac{P_{Y}(Y)}{P_{Y|X}(Y|X)}   
 \end{equation}
The expectation in the right side of \eqref{eq:-Abasic} is with respect to $P_{Z^M}P_{U|Z^M}  P_{W|UZ^M} P_{X} P_{Y|X} $. It is equal to
\begin{align}
&~ \E{ \E{\min\left\{ 1, G \right\}  \mid X, Y, Z^M,  1 \left\{ U \leq W \right\}}  } \notag \\
\leq&~
\E{ \min \left\{ 1, \E{G \,\mid X, Y, Z^M, 1 \left\{ U \leq W \right\} \right\}   }  } \label{eq:-Aminconcave}\\
= &~ \E{ \min \left\{ 1, W \frac{P_{Y}(Y)}{P_{Y|X}(Y|X)}   \right\}}  \label{eq:-AUmarkov}\\
=&~ \E{\exp\left( - \left| \imath_{X; Y}(X; Y) - \log W \right|^+ \right) } \label{eq:-AE}
\end{align}
 where 
\begin{itemize}
 \item \eqref{eq:-Aminconcave} applies Jensen's inequality to the concave function $\min\{1, a\}$;
 \item \eqref{eq:-AUmarkov} uses 
$P_{U| X,Y, Z^M, 1\left\{U \leq W\right\}}
 = P_{U|Z^M, 1\left\{U \leq W\right\}}
 $;
  \item \eqref{eq:-AE} is due to $\min\{1, a\} = \exp \left(- \left| \log \frac 1 a \right|^+\right)$, where $a$ is nonnegative.
\end{itemize}
To evaluate the probability in the right side of \eqref{eq:-Abasic}, note that conditioned on $S = s$, $W = w$, $U$ is distributed as:
\begin{equation}
 P_{U|S, W}(m|s, w) = 
\begin{cases}
 \rho(s) (1 - \rho(s))^{m - 1} &m = 1, 2, \ldots, w\\
 (1 - \rho(s))^w &m = M + 1
\end{cases}
   \label{eq:-As} \\
\end{equation}
where we denoted for brevity
\beq
\rho(s) = P_{Z}(B_d(s)) \label{eq:rho}
\eeq
Therefore, 
\begin{align}
 \Prob{U > W } &= \E{ \Prob{U > W| S, W} }\\
 &= \E{  \left(1 - \rho(S) \right)^{W} } \label{eq:-AProb}
\end{align}
Applying \eqref{eq:-AE} and \eqref{eq:-AProb} to \eqref{eq:-Abasic} and invoking Shannon's random coding argument, \eqref{eq:Ag} follows.
 \end{proof}
\begin{remark}
 As we saw in the proof of Theorem \ref{thm:Ag}, if we restrict $W$ to take values on $\left\{1, \ldots, M\right\}$, then the bound on the error probability $\epsilon$ in \eqref{eq:Ag} is achieved in the class of $(M, d, \epsilon)$ codes. The code size $M$ that leads to tight achievability bounds following from Theorem \ref{thm:Ag} is in general much larger than the size that achieves the minimum in \eqref{eq:ASeparate}. In that case, $M$ is chosen so that $\log M$ lies  between $k R(d)$  and $n C$ so as to minimize the sum of source and channel decoding error probabilities without the benefit of
a channel decoder that exploits residual source redundancy. 
In contrast, Theorem \ref{thm:A} is obtained with an approximate MAP decoder that allows a larger choice for $\log M$, even beyond $nC$.
Still we can achieve a good $(d, \epsilon)$ tradeoff because the channel code employs unequal error protection: those codewords with higher probabilities are more reliably decoded. 
\label{rem:Mcodes}
\end{remark}
 
\begin{remark}
Had we used the ML channel decoder in lieu of \eqref{eq:MAPapprox} in the proof of Theorem \ref{thm:Ag}, we would conclude that a $(d, \epsilon)$ code exists with
 \begin{align}
 \epsilon \leq \inf_{P_{X}, P_{Z}, M} \bigg\{&
\E{ \exp\left( - \left| \imath_{X; Y}(X; Y )  - \log (M-1) 
\right|^+ \right)}
\notag\\&
+  \E{  \left(1 - P_{Z}(B_d(S)) \right)^M } \bigg\}\label{eq:AgSeparate}
\end{align}
 which corresponds to the SSCC bound in \eqref{eq:ASeparate} with  the worst-case average channel error probability $\epsilon^\star(M)$ upper bounded using the random coding union (RCU) bound of \cite{polyanskiy2010channel} and the source error probability $\epsilon^\star(M, d)$ upper bounded using the random coding achievability bound of \cite{kostina2011fixed}.
 \end{remark}
\begin{remark}
Weakening \eqref{eq:Ag} by letting $W = M$, we obtain a slightly looser version of \eqref{eq:AgSeparate} in which $M-1$ in the exponent is replaced by $M$. To get a generally tighter bound than that afforded by SSCC, a more intelligent choice of $W$ is needed, as detailed next in Theorem \ref{thm:A}. 
\label{rem:Aweaken}
\end{remark}
\begin{thm}[Achievability]
There exists a $(d, \epsilon)$ source-channel code with
\begin{align}
\epsilon \leq 
\notag \\
\inf_{P_{X}, P_{Z}, \gamma > 0} \bigg\{ &
\E{ \exp\left( - \left| \imath_{X; Y}(X; Y )  - \log \frac{\gamma}{P_{Z}(B_d(S))}
\right|^+ \right)}
 \notag\\&
 + e^{1-\gamma}
 \bigg\} \label{eq:A}
\end{align}
where the expectation is with respect to 
$P_S P_{X} P_{Y|X} P_{Z}$
defined on 
$\mathcal M \times \mathcal X \times \mathcal Y \times \widehat {\mathcal M}$.
\label{thm:A}
\end{thm}
\begin{proof}
We fix an arbitrary $\gamma > 0$ and choose
\begin{equation}
 W = \left \lfloor \frac{\gamma}{\rho\left(S\right)} \right \rfloor \label{eq:-Aw}
\end{equation}
where $\rho(\cdot)$ is defined in \eqref{eq:rho}. Observing that
 \begin{align}
  (1 - \rho(s))^{\left \lfloor \frac {\gamma}{\rho(s)}\right \rfloor} &\leq (1 - \rho(s))^{\frac {\gamma}{\rho(s)} - 1}\\
  &\leq  e^{- \rho(s) \left( \frac{\gamma}{\rho(s)} - 1\right)} \\
  &\leq e^{1 -\gamma} \label{eq:-Af}
  \end{align}
we obtain \eqref{eq:A} by weakening \eqref{eq:Ag} using \eqref{eq:-Aw} and \eqref{eq:-Af}. 
\end{proof}

 In the case of almost-lossless JSCC, the bound in Theorem \ref{thm:A} can be sharpened as shown recently by Tauste Campo et al. \cite{campo2011random}.

\begin{thm}[Achievability, almost-lossless JSCC \cite{campo2011random}]
 There exists a $(0, \epsilon)$ code with
\begin{equation}
\epsilon \leq \inf_{P_X } \E{\exp\left( -| \imath_{X; Y}(X; Y) - \imath_S (S)|^+ \right)} \label{eq:Alossless}
\end{equation}
where the expectation is with respect to $P_S P_X P_{Y|X}$ defined on $\mathcal M \times \mathcal X \times \mathcal Y$.
\label{thm:Alossless}
\end{thm}


\section{Gaussian Approximation}
\label{sec:2order}

In addition to the basic conditions \eqref{item:a}-\eqref{item:c} of Section \ref{sec:prelim}, in this section we impose the following restrictions.
\begin{enumerate}[(i)]
\item The channel is stationary and memoryless, $P_{Y^n | X^n}  = P_{\mathsf Y| \mathsf X} \times \ldots \times P_{\mathsf Y | \mathsf X}$. If the channel has an input cost function then it satisfies 
$\mathsf c_n(x^n) = \frac 1 n \sum_{i= 1}^n \mathsf c(x_i)$. 
 \label{item:ch}
\item The source is stationary and memoryless,  $P_{S^k}  = P_{\mathsf S} \times \ldots \times P_{\mathsf S}$, and the distortion measure is separable,  $\mathsf d_k(s^k, z^k) = \frac 1 k \sum_{i = 1}^k \mathsf d(s_i, z_i)$. \label{item:s}
\item The distortion level satisfies $d_{\min} < d < d_{\max}$, where $d_{\min}$ is defined in \eqref{eq:dmin}, and $d_{\max} =\inf_{\mathsf z \in \hat{\mathcal S}} \E{\mathsf d(\mathsf S, \mathsf z)}$, where the average is with respect to the unconditional distribution of $\mathsf S$. The excess-distortion probability satisfies $0 < \epsilon < 1$. \label{item:dminmax}
 \item $
 \E{\mathsf d^{9}(\mathsf S, \mathsf Z^\star)} < \infty
$
where the average is with respect to $P_{\mathsf S} \times P_{\mathsf Z^\star}$ and $P_{\mathsf Z^\star}$ is the output distribution corresponding to the minimizer in \eqref{eq:RR(d)}. \label{item:last} 
\end{enumerate}
The technical condition \eqref{item:last} ensures applicability of the Gaussian approximation in the following result. 


\begin{thm}[Gaussian approximation]
\label{thm:2order}
Under restrictions \eqref{item:ch}--\eqref{item:last}, the parameters of the optimal $(k, n, d, \epsilon)$ code satisfy
\begin{equation}
 nC - k R(d) = \sqrt{n V + k \mathcal V(d)}\Qinv{ \epsilon} + \theta\left(n\right) \label{eq:2order}
 \end{equation}
 where
 \begin{enumerate}[1.]
 \item $\mathcal V(d)$ is the source dispersion given by
 \begin{equation}
\mathcal V(d) = {\rm Var}\left[\jmath_{\mathsf S}(\mathsf S, d) \right] \label{eq:sdispersion}
\end{equation}
\item $V$ is the channel dispersion given by:
\begin{enumerate}[a)]
\item If $\mathcal A$ and $\mathcal B$ are finite and the channel has no cost constraints,
\begin{align}
 V &= {\rm Var}\left[\imath_{\mathsf X; \mathsf Y}^\star(\mathsf X^\star; \mathsf Y^\star) \right]\label{eq:cdispersion}\\
 \imath_{\mathsf X; \mathsf Y}^\star (x; y) &= \log \frac{dP_{\mathsf Y| \mathsf X = x}}{dP_{\mathsf Y^\star}}(y)
\end{align}
where $\mathsf X^\star$, $\mathsf Y^\star$ are the capacity-achieving input and output random variables.
\item If the channel is Gaussian with either equal or maximal power constraint, 
\begin{equation}
  V = \frac 1 2 \left( 1 - \frac 1 {\left(1+P\right)^2}\right) \log^2 e\label{eq:DispersionAWGN}
\end{equation}
where $P$ is the signal-to-noise ratio. 
\label{item:2orderGauss}
\end{enumerate}
\item The remainder term $\theta(n)$ satisfies:

\begin{enumerate}
\item  If $\mathcal A$ and $\mathcal B$ are finite, the channel has no cost constraints and $V > 0$, 
\begin{align}
-\underline c \log n + \bigo{1} &\leq \theta\left(n\right) \label{eq:Cremainder}\\
&\leq \bar c \log n + \log\log n + \bigo{1} 	\label{eq:Aremainder}
\end{align}
where 
 \begin{align}
 \underline c &= |\mathcal A| - \frac 1 2\\
 \bar c &= 1 + \frac{\Var{\Lambda^{\prime}_{\mathsf Z^\star}(\mathsf S, \lambda^\star) }}{ \E{\left| \Lambda^{\prime\prime}_{\mathsf Z^\star}(\mathsf S, \lambda^\star)\right|} \log e}
 \label{eq:Cbar}
\end{align}
In \eqref{eq:Cbar}, $(\cdot)^\prime$ denotes differentiation with respect to $\lambda$, $\Lambda_{\mathsf Z^\star}(\mathsf s, \lambda)$ is defined by 
\begin{equation}
 \Lambda_{\mathsf Z^\star}( \mathsf s, \lambda) = \log \frac 1 {\E{\exp\left( \lambda d -\lambda d(\mathsf s, \mathsf Z^\star) \right) }} \label{eq:jlambdagen}
\end{equation}
(cf.  Definition \ref{defn:id}) and $\lambda^\star = - R^\prime(d)$. 
\label{item:2orderV>0}
\item If $\mathcal A$ and $\mathcal B$ are finite, the channel has no cost constraints and $V = 0$, \eqref{eq:Aremainder} still holds, while \eqref{eq:Cremainder} is replaced with
\begin{equation}
\liminf_{n \to \infty} \frac{ \theta\left(n\right) }{\sqrt n} \geq 0 \label{eq:Cremainder0}
\end{equation}
\label{item:2orderV=0}
\item If the channel is such that the (conditional) distribution of $\imath_{\mathsf X; \mathsf Y}^\star( \mathsf x; \mathsf Y)$ does not depend on $x \in \mathcal A$ (no cost constraint), then $\underline c = \frac 1 2$. \footnote{
Note added in proof: we have shown recently that the symmetricity condition is actually superfluous.
}
\label{item:2ordersym}
\item If the channel is Gaussian with equal or maximal power constraint, \eqref{eq:Aremainder} still holds, and \eqref{eq:Cremainder} holds with $\underline c = \frac 1 2$. 
\label{item:2ordergauss}
\item In the almost-lossless case, $R(d) = H(\mathsf S)$,  and provided that the third absolute moment of $\imath_{\mathsf S}\mathsf (\mathsf S)$ is finite, \eqref{eq:2order} and \eqref{eq:Cremainder} still hold, while \eqref{eq:Aremainder} strengthens to
\begin{equation}
\theta\left(n \right)  \leq \frac 1 2 \log n + \bigo{1} \label{eq:Aremainderlossless}
\end{equation} \label{item:2orderlossless}
\end{enumerate}
\end{enumerate}
\end{thm}

\begin{proof}
\begin{itemize}
 \item Appendices \ref{appx:2orderV>0} and \ref{appx:2orderV=0} show the converses in \eqref{eq:Cremainder} and \eqref{eq:Cremainder0} for cases $V > 0$ and $V = 0$, respectively, using Theorem \ref{thm:CT}.
 \item Appendix \ref{appx:2orderCsym} shows the converse for the symmetric channel (\ref{item:2ordersym}) using Theorem \ref{thm:C1sym}.
 \item Appendix \ref{appx:2orderCgauss} shows the converse for the Gaussian channel (\ref{item:2ordergauss}) using Theorem \ref{thm:C1sym}.
 \item Appendix \ref{appx:2orderAlossless} shows the achievability result for almost lossless coding (\ref{item:2orderlossless}) using Theorem \ref{thm:Alossless}.
 \item Appendix \ref{appx:2orderAlossy} shows the achievability result in \eqref{eq:Aremainder} for the DMC using Theorem \ref{thm:A}. 
 \item Appendix \ref{appx:2orderAgauss} shows the achievability result for the Gaussian channel (\ref{item:2ordergauss}) using Theorem \ref{thm:A}.
\end{itemize}
\end{proof}

\begin{remark}
If the channel and the data compression codes are designed separately, we can invoke channel coding \cite{polyanskiy2010channel} and lossy compression \cite{kostina2011fixed} results in \eqref{eq:2orderCC} and \eqref{eq:2orderSC} to show that (cf. \eqref{eq:2orderSSCCintro})
\begin{align}
 nC - kR(d) &\leq \min_{\eta + \zeta \leq \epsilon} \left\{ \sqrt{nV} \Qinv{\eta} + \sqrt{k\mathcal V(d)} \Qinv{\zeta}\right\} \notag\\
 &+ \bigo{\log n} \label{eq:2orderSeparate}
\end{align} 
\label{rem_SSCC} 
\end{remark}
Comparing \eqref{eq:2orderSeparate} to \eqref{eq:2order}, observe that if either the channel or the source (or both) have zero dispersion,  the joint source-channel coding dispersion can be achieved by separate coding. In that special case, either the $\mathsf d$-tilted information or the channel information density are so close to being deterministic that there is no need to account for the true distributions of these random variables, as a good joint source-channel code would do. 

The Gaussian approximations of JSCC and SSCC in \eqref{eq:2order} and \eqref{eq:2orderSeparate}, respectively, admit the following heuristic interpretation when $n$ is large (and thus, so is $k$): since the source is stationary and memoryless, the normalized $\mathsf d$-tilted information
$
 J = \frac 1 n \jmath_{S^k}\left( S^k , d\right)
$
becomes approximately Gaussian with mean $\frac k n R(d)$ and variance $\frac{k}{n}\frac{\mathcal V(d)}{n}$. Likewise, the conditional normalized channel information density
$
 I = \frac 1 n \imath_{X^n; Y^n}^\star(x^{n}; Y^{n \star})
$
is, for large $k$, $n$, approximately Gaussian with mean $C$ and variance $\frac{V}{n}$ for all $x^n \in\mathcal A^n$ typical according to the capacity-achieving distribution. Since a good encoder chooses such inputs for (almost) all source realizations, and the source and the channel are independent, the random variable $ I -  J$ is approximately Gaussian with mean $C - \frac k n R(d)$ and variance $\frac 1 n \left( \frac k n \mathcal V(d) + V \right)$, and \eqref{eq:2order} reflects the intuition that under JSCC, the source is reconstructed successfully within distortion $d$ if and only if the channel information density exceeds the source $\mathsf d$-tilted information, that is, $\left\{ I >  J\right\}$. 
In contrast, in SSCC, the source is reconstructed successfully with high probability if $( I,  J)$ falls in the intersection of half-planes  
$
\left\{I > r \right\} \cap \left\{J < r \right\}
$
for some $r = \frac {\log M}{n}$, which is the capacity of the noiseless link between the source and the channel code block that can be chosen so as to maximize the probability of that intersection, as reflected in \eqref{eq:2orderSeparate}. Since in JSCC the successful transmission event is strictly larger than in SSCC, i.e. $\left\{I > r \right\} \cap \left\{J < r \right\}  \subset \left\{I > J\right\}$, separate source/channel code design incurs a performance loss. It is worth pointing out that $\left\{I > J\right\}$ leads to successful reconstruction even within the paradigm of the codes in Definition \ref{defn:jsccM} because, as explained in Remark \ref{rem:Mcodes}, unlike the SSCC case, it is not necessary that 
$\frac{\log M}{n}$ lie between $I$ and $J$ for successful reconstruction.  

\begin{remark}
 Using Theorem \ref{thm:2order}, it can be shown that
\begin{equation}
 R(n, d, \epsilon) = \frac{C}{R(d)} - \sqrt{\frac{\mathscr V(d)}{n}}\Qinv{\epsilon} - \frac 1 {R(d)}\frac{\theta(n)}{n}\label{eq:2orderR}
\end{equation}
where the rate-dispersion function of JSCC is found as (recall Definition \ref{defn:dispersion}),
\begin{equation}
 \mathscr V(d) = \frac{R(d) V + C \mathcal V(d) }{R^3(d)} \label{eq:JSCCdispersion}
\end{equation}
\end{remark}
\begin{remark}
Under regularity conditions similar to those in \cite[Theorem 14]{kostina2011fixed}, it can be shown that
\begin{equation}
 D(nR, n, \epsilon) = D\left(\! \frac{C}{R} \!\right) + \sqrt{\frac{\mathscr W(R)}{n}}\Qinv{\epsilon} + \frac{\partial}{\partial R} D\left( \!\frac{C}{R} \!\right) \frac{\theta(n)}{n} \label{eq:2orderD}
\end{equation}
 where the distortion-dispersion function of JSCC is given by
\begin{align}
\mathscr W(R) &= \left( \frac{\partial}{\partial R} D \left(\frac{C}{R} \right) \right)^2 \left( V + R \mathcal V\left( D\left(\frac{C}{R} \right)\right)  \right) 
\label{eq:JSCCdispersiondistortion}
\end{align}
\end{remark} 
\begin{remark}
If the basic conditions \eqref{item:b} and/or \eqref{item:c} fail so that there are several distributions $P_{\mathsf Z^\star | \mathsf S}$ and/or several $P_{\mathsf X^\star}$ that achieve the rate-distortion function and the capacity, respectively, then, for $\epsilon < \frac 1 2$,  
\begin{align}
 \mathscr V(d) &\leq \min \mathscr V_{\mathsf Z^\star; \mathsf X^\star}(d)\\
  \mathscr W(R) &\leq \min \mathscr W_{\mathsf Z^\star; \mathsf X^\star}(R)
\end{align}
where the minimum is taken over $P_{\mathsf Z^\star | \mathsf S}$ and $P_{\mathsf X^\star}$, and $\mathscr V_{\mathsf Z^\star; \mathsf X^\star}(d)$ (resp. $\mathscr W_{\mathsf Z^\star; \mathsf X^\star}(R)$) denotes \eqref{eq:JSCCdispersion} (resp. \eqref{eq:JSCCdispersiondistortion}) computed with $P_{\mathsf Z^\star | \mathsf S}$ and $P_{\mathsf X^\star}$. The reason for possibly lower achievable dispersion in this case is that we have the freedom to map the unlikely source realizations leading to high probability of failure to those codewords resulting in the maximum variance so as to increase the probability that the channel output escapes the decoding failure region. 
\end{remark}
\begin{remark}
The dispersion of the Gaussian channel is given by \eqref{eq:DispersionAWGN}, regardless of whether an equal or a maximal power constraint is imposed.
An equal power constraint corresponds to the subset of allowable channel inputs being the power sphere:
\begin{equation}
 \mathcal F \left(P\right) = \left\{ x^n \in \mathbb R^n \colon \frac{|x^n|^2}{\sigma_{\mathsf N}^2} = nP \right\} \label{eq:equalP}
\end{equation}
where $\sigma_{\mathsf N}^2$ is the noise power.  In a maximal power constraint, \eqref{eq:equalP} is relaxed  replacing `$=$' with `$\leq$'. 

Specifying the nature of the power constraint in the subscript, we remark that the bounds for the maximal constraint can be obtained from the bounds for the equal power constraint via the following relation
\begin{equation}
 k^\star_{\mathrm{eq}}(n, d, \epsilon) \leq k^\star_{\max}(n, d, \epsilon) \leq  k^\star_{\mathrm{eq}}(n+1, d, \epsilon) \label{eq:kstareqmax}
\end{equation}
where the right-most inequality is due to the following idea dating back to Shannon:  a $(k, n, d, \epsilon)$ code with a maximal power constraint can be converted to a $(k, n+1, d, \epsilon)$ code with an equal power constraint by appending an $(n+1)$-th coordinate to each codeword to equalize its total power to $n  \sigma_{\mathsf N}^2 P$. From \eqref{eq:kstareqmax} it is immediate that the channel dispersions for maximal or equal power constraints must be the same.
\label{rem:Gauss}
\end{remark}

\section{Lossy transmission of a BMS over a BSC}
\label{sec:BMS_BSC}
In this section we particularize the bounds in Sections \ref{sec:C}, \ref{sec:A} and the approximation in Section \ref{sec:2order} to the transmission of a BMS with bias $p$ over a BSC with crossover probability $\delta$. The target bit error rate satisfies $d \leq p$. 

The rate-distortion function of the source and the channel capacity are given by, respectively,
\begin{align}
 R(d) &= h(p) - h(d) \label{eq:R(d)BMS}\\
 C &= 1 - h(\delta) \label{eq:CapacityBSC}
\end{align}
The source and the channel dispersions are given by \cite{polyanskiy2010channel,kostina2011fixed}:
\begin{align}
 \mathcal V(d) &= p(1-p)\log^2 \frac{1 - p}{p} \label{eq:DispersionBMS}\\
 V &= \delta(1-\delta)\log^2 \frac{1 - \delta}{\delta} \label{eq:DispersionBSC}
\end{align}
where note that \eqref{eq:DispersionBMS} does not depend on $d$. The rate-dispersion function in \eqref{eq:JSCCdispersion} is plotted in Fig. \ref{fig:binarydispersion}. 
\begin{figure}[htbp]
    \epsfig{file=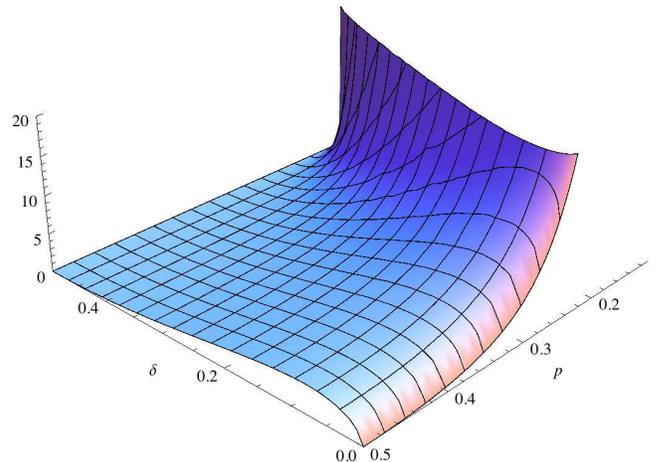,width=1\linewidth}
\caption{The rate-dispersion function for the transmission of a BMS over a BSC with $d = 0.11$ as a function of $(\delta, p)$ in $\left(0, \frac 1 2 \right) \times \left(d, \frac 1 2 \right)$. It increases unboundedly as $p \to d$, and vanishes as $\delta \to \frac 1 2$ or $(\delta, p) \to \left( 0, \frac 1 2 \right)$. }
\label{fig:binarydispersion}
\end{figure}

Throughout this section, $w(a^\ell)$ denotes the Hamming weight of the binary $\ell$-vector $a^\ell$, and $T_\alpha^\ell$ denotes a binomial random variable with parameters $\ell$ and $\alpha$, independent of all other random variables. 

For convenience, we define the discrete random variable $U_{\alpha, \beta}$ by
\begin{equation}
 U_{\alpha, \beta} =  \left( T_{\alpha}^k  - kp\right) \log \frac{1 - p}{p} +  \left( T_{\beta}^n  - n \delta \right) \log \frac{1 - \delta}{\delta}
\label{eq:binaryU}
\end{equation}
In particular, substituting $\alpha = p$ and $\beta = \delta$ in \eqref{eq:binaryU}, we observe that the terms in the right side of \eqref{eq:binaryU} are zero-mean random variables whose variances are equal to $k \mathcal V(d)$ and $n V$, respectively. 

Furthermore, the binomial sum is denoted by
\begin{equation}
 \binosum{k}{\ell } = \sum_{i = 0}^{\ell} {k \choose i}
\end{equation}

A straightforward particularization of the $\mathsf d$-tilted information converse in Theorem \ref{thm:C1sym} leads to the following result.  
\begin{thm}[Converse, BMS-BSC]
  Any $(k, n, d, \epsilon)$ code for transmission of a BMS with bias $p$ over a BSC with bias $\delta$ must satisfy
  \begin{align}
 \epsilon \geq  \sup_{\gamma \geq 0} \bigg\{ &~ \Prob{ U_{p, \delta} \geq nC - kR(d)+ \gamma} 
 - \exp\left(-\gamma\right) \bigg\}\label{eq:CBMS-BSC}
\end{align}
\label{thm:CBMS-BSC}
\end{thm}
\begin{proof}
Let $P_{\bar Y^n} = P_{Y^{n \star} }$, which is the equiprobable distribution on $\{0, 1\}^n$.  An easy exercise reveals that
\begin{align}
 \jmath_{S^k}(s^k, d) &= \imath_{S^k}(s^k) - k h(d) \label{eq:jBMS}\\
  \imath_{S^k}(s^k) &= k h(p) + \left( w(s^k) - k p \right)\log \frac{1 - p}{p}  \label{eq:iBMS}\\
 \imath_{X^n; Y^{n \star}}(x^n; y^n) &= n \left( \log 2 - h(\delta) \right) 
 \notag \\
 &- \left(  w(y^n - x^n) - n \delta\right)  \log \frac{1 - \delta}{\delta}\label{eq:ixyBSC}
\end{align}
Since $w(Y^n - x^n)$  is distributed as $T_{\delta}^n$ regardless of $x^n \in \{0, 1\}^n$, and  $w(S^k)$ is distributed as $T_{p}^k$, the condition in Theorem \ref{thm:C1sym} is satisfied, and \eqref{eq:C1sym} becomes \eqref{eq:CBMS-BSC}. 
\end{proof}
The hypothesis-testing converse in Theorem \ref{thm:Cht} particularizes to the following result:
  \begin{thm}[Converse, BMS-BSC]
 Any $(k, n, d, \epsilon)$ code for transmission of a BMS with bias $p$ over a BSC with bias $\delta$ must satisfy
  \beq
    \Prob{ U_{\frac 1 2, \frac 1 2} < r} + \lambda \Prob{U_{\frac 1 2, \frac 1 2} = r} \leq  \binosum{k}{\lfloor kd \rfloor }2^{-k} \label{eq:ChtBMS-BSC}
  \eeq
  where $0 \leq \lambda < 1$ and scalar $r$ are uniquely defined by
  \beq
 \Prob{ U_{p, \delta}  < r} + \lambda \Prob{  U_{p, \delta}  = r} = 1 - \epsilon
  \eeq
  \label{thm:ChtBMS-BSC}
  \end{thm}
\begin{proof}
As in the proof of Theorem \ref{thm:CBMS-BSC}, we let $P_{\bar Y^n}$ be the equiprobable distribution on $\{0, 1\}^n$, $P_{\bar Y^n} = P_{Y^{ n \star}}$. Since under $P_{Y^n|X^n = x^n}$, 
$w\left(Y^n - x^n\right)$
  is distributed as $T_{\delta}^n$, and under $P_{Y^{n\star}}$, 
  $w\left(Y^n - x^n\right)$
  is distributed as $T_{\frac 1 2}^n$, irrespective of the choice of $x^n \in A^n$, the distribution of the information density in \eqref{eq:ixyBSC} does not depend on the choice of $x^n$ under either measure, so Theorem \ref{thm:Chtsym} can be applied. Further, we choose $Q_{S^k}$ to be the equiprobable distribution on $\{0, 1\}^k$ and observe that under $P_{S^k}$, the random variable $w(S^k)$ in \eqref{eq:iBMS} has the same distribution as $T_{p}^{k}$, while under $Q_{S^k}$ it has the same distribution as $T_{\frac 1 2 }^k$. Therefore, the log-likelihood ratio for testing between $P_{S^k}P_{Y^n|X^n = x^n}$ and $Q_{S^k} P_{Y^{n \star}}$ has the same distribution as (`$\sim$' denotes equality in distribution)
\begin{align}
 &~
 \log \frac{P_{S^k}(S^k)P_{Y^n | X^n = x^n}(Y^n)}{Q_{S^k}(S^k) P_{Y^{n \star}}(Y^n)} 
 \notag \\
 =
 &~
 \imath_{X^n; Y^{n \star}}(x^n; Y^n) -  \imath_{S^k}(S^k)  + k \log 2  \\
\sim&~ 
n\log 2 - n h(\delta) - k h(p) 
\notag \\
-&~ 
\begin{cases}
  U_{p, \delta} & \text{under } P_{S^k}P_{Y^n|X^n = x^n}\\
  U_{\frac 1 2, \frac 1 2} & \text{under } Q_{S^k}P_{Y^{n \star}}
\end{cases}
\end{align}
so $\beta_{1-\epsilon}(P_{S^k}P_{Y^n|X^n = x^n}, Q_{S^k}P_{Y^{n \star}})$ is equal to the left side of \eqref{eq:ChtBMS-BSC}. Finally, matching the size of the list to the fidelity of reproduction using \eqref{eq:lossylistsize}, we find that $L$ is equal to the right side of \eqref{eq:ChtBMS-BSC}.
\end{proof}

If the source is equiprobable, the bound in Theorem \ref{thm:ChtBMS-BSC} becomes particularly simple, as the following result details.  
 \begin{thm}[Converse, EBMS-BSC]
For $p = \frac 1 2$, if there exists a $(k, n, d, \epsilon)$ joint source-channel code, then
\beq
\lambda {n \choose r^\star + 1} + \binosum{n}{r^\star} \leq \binosum{k}{ \lfloor kd \rfloor }2^{n-k} \label{eq:ChtEBMS-BSC}
\eeq
where
\beq
r^\star = {\max }\left\{ r: \ \sum_{t = 0}^r {n \choose t} \delta^t (1 - \delta)^{n-t} \leq 1 - \epsilon \right\}\label{eq:ChtEBMS-BSCrstar}
\eeq
and $\lambda \in [0, 1)$ is the solution to
\beq
\sum_{j = 0}^{r^\star} {n \choose t} \delta^t (1 - \delta)^{n-t} + \lambda  \delta^{r^\star + 1} (1 - \delta)^{n - r^\star - 1} {n \choose r^\star + 1} = 1 - \epsilon 
\eeq
\label{thm:ChtEBMS-BSC}
\end{thm}

The achievability result in Theorem \ref{thm:A} is particularized as follows. 
  \begin{thm}[Achievability, BMS-BSC]
  There exists an $(k, n, d, \epsilon)$ joint source-channel code with
\begin{align}
\epsilon \leq 
\inf_{\gamma >0} \bigg\{&
\E{ \exp\left( - \left| U - \log \gamma
\right|^+ \right)}
+   e^{1 -\gamma}
 \bigg\}
\label{eq:ABMS-BSC}
\end{align}
where
\begin{equation}
 U = n C - \left(T_\delta^n - n \delta \right) \log \frac{1 - \delta}{\delta}  - \log \frac{1}{ \rho(T_p^k)} 
\end{equation}
and $\rho \colon \{0, 1, \ldots, k\} \mapsto [0, 1]$ is defined as
\begin{equation}
 \rho(T) = \sum_{t = 0}^k  L(T, t) q^t (1 - q)^{k - t}
\end{equation}
with
\begin{align}
 L(T, t) &= 
\begin{cases}
 {T \choose t_0} { k - T \choose t - t_0} & t - kd \leq T \leq t + kd\\
 0 & \text{otherwise}
\end{cases}
\label{eq_ABMSL}\\
t_0 &= \left \lceil\frac{t+T-kd}{2}\right\rceil^+\\
q &= \frac {p - d}{1 - 2d} \label{eq_q}
\end{align}
\label{thm:ABMS-BSC}
 \end{thm}
 \begin{proof}
We weaken the infima over $P_{X^n}$ and $P_{Z^k}$ in \eqref{eq:A} by choosing them to be the product distributions generated by the capacity-achieving channel input distribution and the rate-distortion function-achieving reproduction distribution, respectively, i.e. $P_{X^n}$ is equiprobable on $\{0, 1\}^n$, and $P_{Z^k} = P_{\mathsf Z^\star} \times \ldots \times P_{\mathsf Z^\star}$, where $P_{\mathsf Z^\star} (1) = q$.  As shown in \cite[proof of Theorem 21]{kostina2011fixed}, 
\begin{equation}
 P_{Z^k}\left( B_d(s^k)\right) \geq \rho(w(s^k)) \label{eq:-ABMS-BSCdball}
\end{equation}
On the other hand, $|Y^n - X^n|_0$ is distributed as $T_{\delta}^n$,  so \eqref{eq:ABMS-BSC} follows by substituting \eqref{eq:ixyBSC} and \eqref{eq:-ABMS-BSCdball} into \eqref{eq:A}.
\end{proof}

In the special case of the BMS-BSC, Theorem \ref{thm:2order} can be strengthened as follows. 
\begin{thm}[Gaussian approximation, BMS-BSC]
\label{thm:2orderBMS-BSC}
The parameters of the optimal $(k, n, d, \epsilon)$ code satisfy \eqref{eq:2order}
 where $R(d)$, $C$, $\mathcal V(d)$, $V$ are given by \eqref{eq:R(d)BMS}, \eqref{eq:CapacityBSC}, \eqref{eq:DispersionBMS}, \eqref{eq:DispersionBSC}, respectively, and the remainder term in \eqref{eq:2order} satisfies
\begin{align}
 \bigo{1} &\leq \theta \left( n\right) \label{eq:CremainderBMS-BSC}\\
 &\leq \log n + \log \log n + \bigo{1} \label{eq:AremainderBMS-BSC}
\end{align}
if $0 < d < p$, and
\begin{align}
- \frac 1 2 \log n + \bigo{1} &\leq \theta \left( n\right) \label{eq:CremainderBMS-BSClossless}\\
 &\leq \frac 1 2 \log n + \bigo{1} \label{eq:AremainderBMS-BSClossless}
\end{align}
if $d = 0$.
\end{thm}
\begin{proof}
 An asymptotic analysis of the converse bound in Theorem \ref{thm:ChtBMS-BSC} akin to that found in \cite[proof of Theorem 23]{kostina2011fixed} leads to  \eqref{eq:CremainderBMS-BSC} and \eqref{eq:CremainderBMS-BSClossless}.  An asymptotic analysis of the achievability bound in Theorem \ref{thm:ABMS-BSC} similar to the one found in \cite[Appendix G]{kostina2011fixed} leads to \eqref{eq:AremainderBMS-BSC}. Finally,  \eqref{eq:AremainderBMS-BSClossless} is the same as \eqref{eq:Aremainderlossless}.
\end{proof}
The bounds and the Gaussian approximation (in which we take 
$\theta\left( n \right) = 0$) are plotted in Fig. \ref{fig:BMS-BSClossless} ($d = 0$), Fig. \ref{fig_EBMS-BSC} (fair binary source, $d > 0$) and Fig. \ref{fig_BMS-BSC} (biased binary source, $d > 0$). A source of fair coin flips has zero dispersion, and as anticipated in Remark \ref{rem_SSCC}, JSSC does not afford much gain in the finite blocklength regime (Fig. \ref{fig_EBMS-BSC}). Moreover, in that case the JSCC achievability bound in Theorem \ref{thm:A} is worse than the SSCC achievability bound. However, the more general achievability bound in Theorem \ref{thm:Ag} with the choice $W = M$, as detailed in Remark \ref{rem:Aweaken}, nearly coincides with the SSCC curve in Fig. \ref{fig_EBMS-BSC}, providing an improvement over Theorem \ref{thm:A}. The situation is different if the source is biased, with JSCC showing significant gain over SSCC (Figures \ref{fig:BMS-BSClossless} and \ref{fig_BMS-BSC}). 
\begin{figure}[htbp]
    \epsfig{file=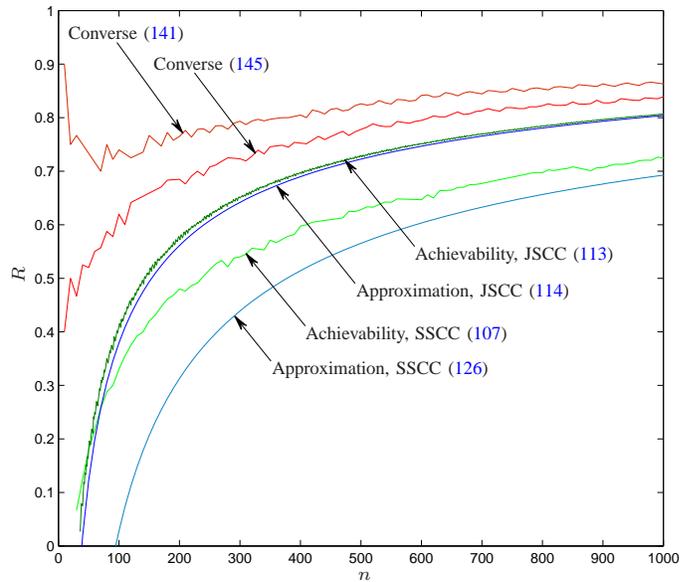,width=1\linewidth}
\caption{ Rate-blocklength tradeoff for the transmission of a BMS with bias $p = 0.11$ over a BSC with crossover probability $\delta = p = 0.11$ and $d = 0$, $\epsilon = 10^{-2}$. }
\label{fig:BMS-BSClossless}
\end{figure}

\begin{figure}[htbp]
    \epsfig{file=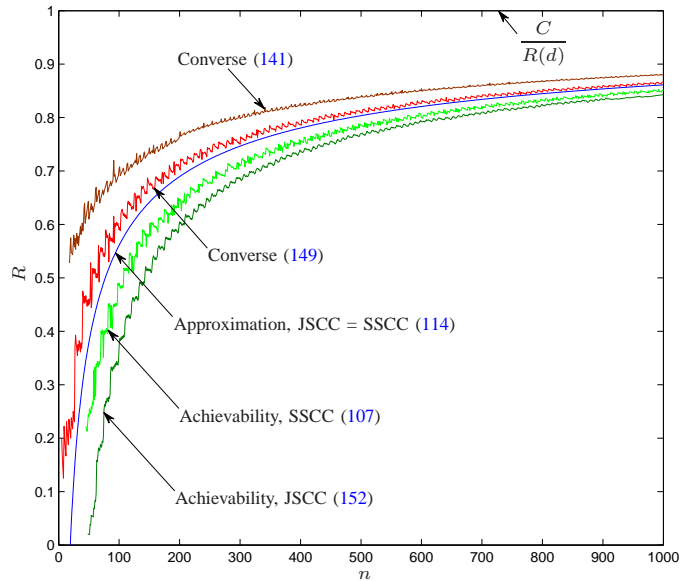,width=1\linewidth}
\caption{ Rate-blocklength tradeoff for the transmission of a fair BMS over a BSC with crossover probability $\delta = d = 0.11$ and $\epsilon = 10^{-2}$. }
\label{fig_EBMS-BSC}
\end{figure}
%


%

\begin{figure}[htbp]
    \epsfig{file=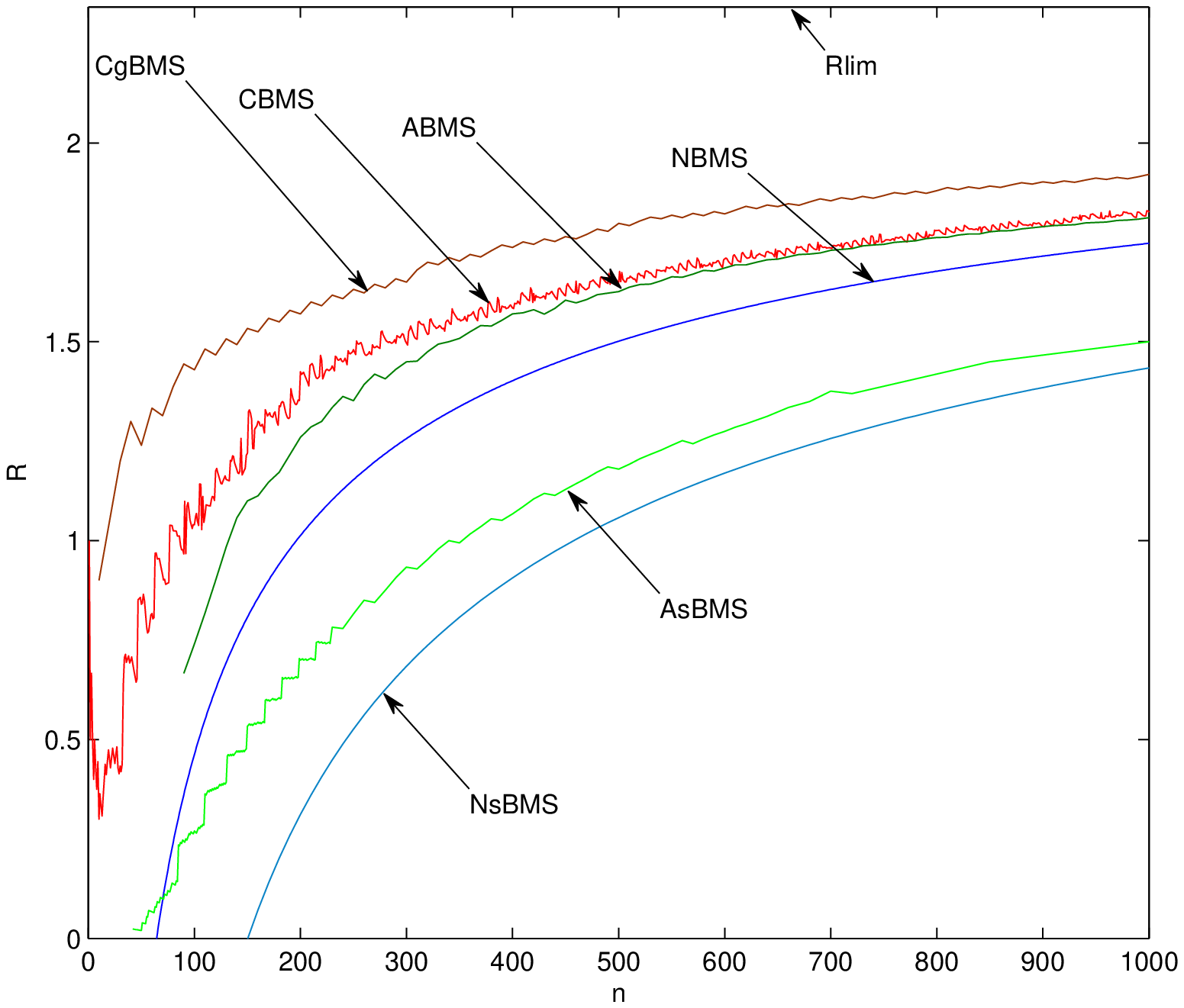,width=1\linewidth}
\caption{ Rate-blocklength tradeoff for the transmission of a BMS with bias $p = 0.11$ over a BSC with crossover probability $\delta = p = 0.11$ and $d = 0.05$, $\epsilon = 10^{-2}$. }
\label{fig_BMS-BSC}
\end{figure}

\section{Transmission of a GMS \\over an AWGN channel}
\label{sec:GMS_AWGN}

In this section we analyze the setup where the Gaussian memoryless source $S_i \sim \mathcal N(0, \sigma_{\mathsf S}^2)$ is transmitted over an AWGN channel, which, upon receiving an input $x^n$,  outputs $Y^n = x^n + N^n$, where $N^n \sim \mathcal N(0, \sigma_{\mathsf N}^2 \mathbf I)$. The encoder/decoder must satisfy two constraints, the fidelity constraint and the cost constraint:
\begin{itemize}
 \item the MSE distortion exceeds $0 \leq d \leq  \sigma_{\mathsf S}^2$ with probability no greater than $0 < \epsilon < 1$;
 \item each channel codeword satisfies the equal power constraint in \eqref{eq:equalP}.\footnote{See Remark \ref{rem:Gauss} in Section \ref{sec:2order} for a discussion of the close relation between an equal and a maximal power constraint.} 
 \end{itemize}

The capacity-cost function and the rate-distortion function are given by
\begin{align}
 R(d) &= \frac 1 2 \log \left(\frac {\sigma_\mathsf S^2}{d}\right) \label{eq:R(d)GMS}\\
  C(P) &= \frac 1 2 \log \left(1 + P\right) \label{eq:CapacityAWGN}
\end{align}
The source dispersion is given by \cite{kostina2011fixed}:
\begin{equation}
 \mathcal V(d) = \frac 1 2 \log^2 e\label{eq:DispersionGMS}
\end{equation}
while the channel dispersion is given by \eqref{eq:DispersionAWGN} \cite{polyanskiy2010channel}. 

In the rest of the section, $W^{\ell}_{\lambda}$ denotes a noncentral chi-square distributed random variable with $\ell$ degrees of freedom and non-centrality parameter $\lambda$, independent of all other random variables, and $f_{W_\lambda^\ell}$ denotes its probability density function.

A straightforward particularization of the $\mathsf d$-tilted information converse in Theorem \ref{thm:C1sym} leads to the following result.  
\begin{thm}[Converse, GMS-AWGN]
  If there exists a $(k, n, d, \epsilon)$ code, then
  \begin{align}
 \epsilon \geq  \sup_{\gamma \geq 0} \bigg\{ &~ \Prob{  U
 \geq n C(P) - k R(d) + 
  \gamma} 
 - \exp\left(-\gamma\right) \bigg\}\label{eq:C1GMS-AWGN}
\end{align}
where
\begin{equation}
 U = \frac {\log e}{2} \left( W^{k}_0 - k \right) + \frac {\log e}{2} \left(\frac{P}{1+P} W^{n}_{\frac n P} -  n  \right)
\end{equation}
\label{thm:C1GMS-AWGN}
\end{thm}
Observe that the terms to the left of the `$\geq$' sign inside the probability in \eqref{eq:C1GMS-AWGN} are zero-mean random variables whose variances are equal to $k \mathcal V(d)$ and $n V$, respectively. 
\begin{proof}
The spherically-symmetric 
$P_{\bar Y^n} = P_{Y^{ n \star}} = P_{\mathsf Y^\star} \times \ldots \times P_{\mathsf Y^\star}$, where $\mathsf Y^\star \sim \mathcal N(0, \sigma_{\mathsf N}^2(1 + P))$ is the capacity-achieving output distribution, satisfies the symmetry assumption of Theorem \ref{thm:C1sym}. More precisely, it is not hard to show (see \cite[(205)]{polyanskiy2010channel}) that for all $x^n \in \mathcal F (\alpha)$,   $\imath_{X^n;  Y^{n \star}}(x^n; Y^n)$ has the same distribution under $P_{Y^{n \star}|X^n = x^n}$ as
\begin{equation}
 \frac n 2 \log\left(1+P\right) - \frac{\log e}{2}\left( \frac{P}{1 + P} W^{n}_{\frac n P} - n\right)\label{eq:ixyAWGNconditional}
\end{equation}
The $\mathsf d$-tilted information in $s^k$ is given by
\begin{equation}
 \jmath_{S^k}(s^k, d) =\frac k 2 \log \frac{\sigma^2_{\mathsf S}}{d} + \left( \frac{|s^k|^2}{\sigma^2_{\mathsf S}} - k\right) \frac{\log e}{2} \label{eq:dtiltedGMS}
\end{equation}
 Plugging \eqref{eq:ixyAWGNconditional} and \eqref{eq:dtiltedGMS} into \eqref{eq:C1sym}, \eqref{eq:C1GMS-AWGN} follows. 
 \end{proof}
The hypothesis testing converse in Theorem \ref{thm:Chtsym} is particularized as follows.
 \begin{thm}[Converse, GMS-AWGN]
\begin{equation}
k\int_0^\infty r^{k-1} \Prob{ P W^{n}_{n\left(1 + \frac 1 P\right)} + k \frac d {\sigma^2} r^2 \leq n \tau} dr\leq 1 \label{eq:ChtGMS-AWGN}
\end{equation}
where $\tau$ is the solution to
\begin{equation}
 \Prob{ \frac{P}{1+P} W^{n}_{\frac n { P} } + W^{k}_{0}  \leq n\tau} = 1 - \epsilon\label{eq:ChtGMS-AWGNa}
\end{equation}
\label{thm:ChtGMS-AWGN}
\end{thm}
\begin{proof}
As in the proof of Theorem \ref{thm:C1GMS-AWGN}, we let $\bar Y^n \sim Y^{n \star} \sim \mathcal N(0, \sigma_{\mathsf N}^2(1 + P)\mathbf I )$. Under $P_{Y^n| X^n = x^n}$, the distribution of $\imath_{X^n; Y^{n^\star}} (x^n; Y^{n \star} )$ is that of  \eqref{eq:ixyAWGNconditional}, while under $P_{Y^{n \star}}$, it has the same distribution as (cf. \cite[(204)]{polyanskiy2010channel})
\begin{equation}
 \frac n 2 \log(1+P) - \frac {\log e} 2 \left(  P W^{n}_{n\left( 1 + \frac 1 P\right)  } - n\right)\label{eq:ixyAWGNunconditional}
\end{equation}
Since the distribution of  $\imath_{X^n; Y^{n \star}} (x^n; Y^{n \star} )$ does not depend on the choice of $x^n \in \mathbb R^n$ according to either measure, Theorem \ref{thm:Chtsym} applies.
Further, choosing $Q_{S^k}$ to be the Lebesgue measure on $\mathbb R^k$, i.e. $dQ_{S^k} = d s^k$, observe that
\begin{equation}
 \log f_{S^k} (s^k) = \log \frac{dP_{S^k}(s^k)}{d s^k} = -\frac k 2 \log\left( 2 \pi \sigma^2_{\mathsf S} \right) - \frac {\log e}{2 \sigma^2_{\mathsf S} } |s^k|^2 
\end{equation}
Now, \eqref{eq:ChtGMS-AWGN} and \eqref{eq:ChtGMS-AWGNa} are obtained by integrating  
\begin{align}
 1\bigg\{  &\log f_{S^k} (s^k) + \imath_{X^n; Y^{n \star}} (x^n; y^n ) > 
 \notag \\
 &\frac n 2 \log(1 + P) + \frac n 2 \log e - \frac k 2 \log(2 \pi \sigma_{\mathsf S}^2) - \frac {\log e}{2} n\tau \bigg\}
\end{align} 
with respect to $d s^k dP_{Y^{n \star}}(y^n\!) \text{ and } dP_{S^k}(s^k\!)dP_{Y^n | X^n = x^n}(y^n\!)$, respectively. 
\end{proof}

The bound in Theorem \ref{thm:A} can be computed as follows.  
\begin{thm}[Achievability, GMS-AWGN]
There exists a $\left(k, n, d, \epsilon\right)$ code such that
\begin{align}
\epsilon &~\leq \inf_{ \gamma >0} \Bigg\{ \E{ \exp 
\left\{- \left|  U - \log \gamma \right|^+\right\}}
+ e^{1-\gamma}    \Bigg\} \label{eq:AGMS-AWGN}
\end{align}
where 
 \begin{align}
 U &=nC(P) - \frac{\log e}{2}\left( \frac{P}{1 + P} W^{n}_{ \frac n P} - n \right)
- \log \frac{F}{\rho(W^{k}_{0})}\\
 F &= \max_{n \in \mathbb N, t \in \mathbb R^+}  \frac{ f_{W^{n}_{nP}}\left(t\right)  }{ f_{W^{n}_{ 0}}\left(\frac{t}{1 + P}\right) } < \infty \label{eq:AGMS-AWGNderivativebound}
\end{align}
and $\rho: \mathbb R^+ \mapsto [0, 1]$ is defined by
\begin{equation}
 \rho(t) = 
\frac{\Gamma \left( \frac k 2 + 1\right)}{\sqrt \pi k \Gamma \left( \frac{k-1}{2} + 1\right) }\left( 1 - L\left( \sqrt{\frac{t}{k}}\right) \right) ^{\frac{k-1} 2}
\end{equation}
where
\begin{equation}
 L(r) = 
\begin{cases}
0 & r < \sqrt{\frac{d}{\sigma^2_{\mathsf S}}} - \sqrt{1 - \frac{d}{\sigma^2_{\mathsf S}}}\\
1 & \left| r- \sqrt{1 - \frac{d}{\sigma^2_{\mathsf S}}}\right| > \sqrt{\frac{d}{\sigma^2_{\mathsf S}}}\\
  \frac{\left(1 + r^2 - 2\frac d {\sigma^2_{\mathsf S}}\right)^2}{4 \left(1 - \frac d {\sigma^2_{\mathsf S}}\right) r^2} & \text{otherwise}
\end{cases}
\end{equation}
\label{thm:AGMS-AWGN}
\end{thm}

\begin{proof}
We compute an upper bound to \eqref{eq:A} for the specific case of the transmission of a GMS over an AWGN channel. First, we weaken the infimum over $P_{Z^k}$ in \eqref{eq:A} by choosing $P_{Z^k}$  to be the uniform distribution on the surface of the $k$-dimensional sphere with center at $\mathbf 0$ and radius 
$
 r_0 = \sqrt k \sigma \sqrt{1 - \frac{d}{\sigma^2_{\mathsf S}}}
$. 
We showed in \cite[proof of Theorem 37]{kostina2011fixed} (see also \cite{wyner1968communication,sakrison1968geometric}) that
\begin{align}
 P_{Z^k}\left( B_d(s^k) \right)  &\geq \rho\left( |s^k|^2\right) \label{eq:GMSdball}
\end{align}
which takes care of the source random variable in \eqref{eq:A}. 

We proceed to analyze the channel random variable $\imath_{X^n; Y^n}(X^n; Y^n)$. Observe that since $X^n$ lies on the power sphere and the noise is spherically symmetric, $|Y^n|^2 = |X^n + N^n|^2$ has the same distribution as $|x_0^n + N^n|^2$, where $x_0^n$ is an arbitrary point on the surface of the power sphere. Letting $x_0^n = \sigma_{\mathsf N} \sqrt P \left( 1, 1, \ldots, 1 \right)$, we see that $\frac 1 {\sigma_{\mathsf N}}|x_0^n + N^n|^2 = \sum_{i = 1}^n \left(\frac{1}{\sigma_{\mathsf N}^2} N_i+ \sqrt P\right)^2$ has the non-central chi-squared distribution with $n$ degrees of freedom and noncentrality parameter $nP$. To simplify calculations, we express the information density as
\begin{equation}
\imath_{X^n; Y^n}(x_0^n; y^n ) = \imath_{X^n; Y^{n \star}}(x_0^n; y^n )  - \log \frac{dP_{Y^n}}{dP_{Y^{n \star}}}\left( y^n\right) \label{eq:-AGMS_AWGNa}
\end{equation}
where $Y^{n \star} \sim \mathcal N(0, \sigma_{\mathsf N}^2(1 + P)\mathbf I )$. 
The distribution of 
$\imath_{X^n; Y^{n \star}}(x_0^n; Y^n )$
 is the same as \eqref{eq:ixyAWGNconditional}.
 Further, due to the spherical symmetry of both $P_{Y^n}$ and $P_{Y^{n \star}}$, as discussed above, we have (recall that `$\sim$' denotes equality in distribution)
\begin{equation}
\frac{dP_{Y^n}}{dP_{Y^{n \star}}}\left( Y^n \right) \sim \frac{f_{W^{n}_{nP}}\left(W^{n}_{nP}\right)  }{ f_{W^{n}_{0}}\left(\frac{W^{n}_{nP}}{1 + P}\right) }
\end{equation}
which is bounded uniformly in $n$ as observed in \cite[(425), (435)]{polyanskiy2010channel}, thus \eqref{eq:AGMS-AWGNderivativebound} is finite, and  \eqref{eq:AGMS-AWGN} follows. 
\end{proof}
The following result strengthens Theorem \ref{thm:2order} in the special case of the GMS-AWGN.
\begin{thm}[Gaussian approximation, GMS-AWGN]
\label{thm:2orderGMS-AWGN}
The parameters of the optimal $(k, n, d, \epsilon)$ code satisfy \eqref{eq:2order}
 where $R(d)$, $C$, $\mathcal V(d)$, $V$ are given by \eqref{eq:R(d)GMS}, \eqref{eq:CapacityAWGN}, \eqref{eq:DispersionGMS}, \eqref{eq:DispersionAWGN}, respectively, and the remainder term in \eqref{eq:2order} satisfies
\begin{align}
 \bigo{1} &\leq \theta \left( n\right) \label{eq:CremainderGMS-AWGN}\\
 &\leq \log n + \log \log n + \bigo{1} \label{eq:AremainderGMS-AWGN}
\end{align}
\end{thm}

\begin{proof}
 An asymptotic analysis of the converse bound in Theorem \ref{thm:ChtGMS-AWGN} similar to that found in \cite[proof of Theorem 40]{kostina2011fixed} leads to  \eqref{eq:CremainderGMS-AWGN}.  An asymptotic analysis of the achievability bound in Theorem \ref{thm:AGMS-AWGN} similar to \cite[Appendix K]{kostina2011fixed} leads to \eqref{eq:AremainderGMS-AWGN}. 
 \end{proof}

Numerical evaluation of the bounds reveals that JSCC noticeably outperforms SSCC in the displayed region of blocklengths (Fig. \ref{fig_GMS-AWGN}).

\begin{figure}[htbp]
    \epsfig{file=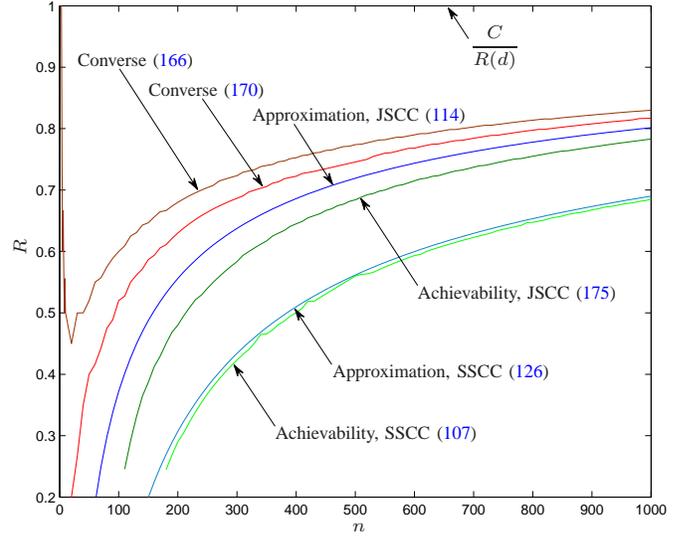,width=1\linewidth}
\caption{ Rate-blocklength tradeoff for the transmission of a GMS with $\frac{d}{\sigma_s^2} = 0.5$ over an AWGN channel with $P= 1$, $\epsilon = 10^{-2}$. }
\label{fig_GMS-AWGN}
\end{figure}

\section{To code or not to code}
\label{sec:tocodeornot}
Our goal in this section is to compare the excess distortion performance of the optimal code of rate $1$ at channel blocklength $n$ with that of the optimal symbol-by-symbol code, evaluated after  $n$ channel uses, leveraging the bounds in Sections  \ref{sec:C} and \ref{sec:A} and the approximation in Section  \ref{sec:2order}. We show certain examples in which symbol-by-symbol coding is, in fact, either optimal or very close to being optimal.  A general conclusion drawn from this section is that even when no coding is asymptotically suboptimal
it can be a very attractive choice for short blocklengths \cite{kostina2012tocodeornot}. 

\subsection{Performance of symbol-by-symbol source-channel codes}

\begin{defn}
An $(n, d, \epsilon, \alpha)$ symbol-by-symbol code is an $(n,n, d, \epsilon, \alpha)$ code $\left(\mathsf f, \mathsf g \right)$ (according to Definition \ref{defn:jscckn}) that satisfies 
\begin{align}
\mathsf f(s^n) &= \left(\mathsf f_1(s_1), \ldots, \mathsf f_1(s_n) \right)\\
\mathsf g(y^n) &= \left(\mathsf g_1(y_1), \ldots, \mathsf g_1(y_n) \right)
\end{align}
for some pair of functions $\mathsf f_1 \colon \mathcal S \mapsto \mathcal A$ and $\mathsf g_1 \colon \mathcal B \mapsto \hat{\mathcal S}$.

The minimum excess distortion achievable with symbol-by-symbol codes at channel blocklength $n$, excess probability $\epsilon$ and cost $\alpha$ is defined by
\begin{equation}
D_1(n, \epsilon, \alpha) =  \inf\left\{ d \colon \exists (n, d, \epsilon, \alpha) \text{ symbol-by-symbol code}\right\}.
 \label{eq:Dstar1ne}
\end{equation}
\label{defn:JSCC1to1code}
\end{defn}

\begin{defn}
The distortion-dispersion function of symbol-by-symbol joint source-channel coding is defined as
\begin{equation}
 \mathscr W_1(\alpha) = \lim_{\epsilon \rightarrow 0} \limsup_{n \rightarrow \infty}
\frac {  n \left(  D\left( C(\alpha)\right) - D_1(n, \epsilon, \alpha)\right)^2}{2 \log_e \frac 1 \epsilon} \label{eq:distortiondispersiondef1}
\end{equation}
where $D(\cdot)$ is the distortion-rate function of the source.  
\end{defn}
As before, if there is no channel input-cost constraint ($\mathsf c^n(x^n) = 0$ for all $x^n \in \mathcal A^n$), we will simplify the notation and write $D_1(n, \epsilon)$ for $D_1(n, \epsilon, \alpha)$ and $\mathscr W_1$ for $ \mathscr W_1(\alpha)$.

In addition to restrictions \eqref{item:ch}--\eqref{item:last} of Section \ref{sec:2order}, we assume that the channel and the source are probabilistically matched in the following sense (cf. \cite{gastpar2003tocodeornot}). 
\begin{enumerate}[(i)]
\setcounter{enumi}{4}
\item There exist $\alpha$, 
$P_{\mathsf X^\star| \mathsf S}$ and $P_{\mathsf Z^\star| \mathsf Y}$
such that $P_{\mathsf X^\star}$ and $P_{\mathsf Z^\star| \mathsf S}$ generated by the joint distribution 
$P_{\mathsf S} P_{\mathsf X^\star| \mathsf S} P_{\mathsf Y| \mathsf X} P_{\mathsf Z^\star| \mathsf Y}$
 achieve the capacity-cost function $C(\alpha)$ and the distortion-rate function $D\left( C(\alpha) \right)$, respectively. 
 \label{item:match}
\end{enumerate}
Condition \eqref{item:match} ensures that symbol-by-symbol transmission attains the minimum average (over source realizations) distortion achievable among all codes of any blocklength. The following results pertain to the full distribution of the distortion incurred at the receiver output and not just its mean. 

\begin{thm}[Achievability, symbol-by-symbol code]
Under restrictions \eqref{item:ch}-\eqref{item:match}, if 
\begin{equation}
 \Prob{\sum_{i = 1}^n \mathsf d(S_i, Z_i^\star) > nd} \leq \epsilon \label{eq:A1}
\end{equation}
 where $P_{ Z^{n \star} | S^n} = P_{\mathsf Z^\star | \mathsf S} \times \ldots \times P_{\mathsf Z^\star | \mathsf S}$, and $P_{\mathsf Z^\star|\mathsf S}$ achieves $D\left( C(\alpha) \right)$, 
 then there exists an $(n, d, \epsilon, \alpha)$ symbol-by-symbol code (average cost constraint). 
 \label{thm:d1}
\end{thm}
\begin{proof}
If \eqref{item:match} holds, then there exist a symbol-by-symbol encoder and decoder such that the conditional distribution of the output of the decoder given the source outcome coincides with distribution $P_{\mathsf Z^\star | \mathsf S}$, so the excess-distortion probability of this symbol-by-symbol code is given by the left side of \eqref{eq:A1}. 
\end{proof}

\begin{thm}[Converse, symbol-by-symbol code]
Under restriction \eqref{item:ch} and separable distortion measure, the parameters of any $(n, d, \epsilon, \alpha)$ symbol-by-symbol code (average cost constraint) must satisfy
\begin{equation}
 \epsilon \geq \inf_{ \substack{P_{\mathsf Z| \mathsf S} \colon\\I(\mathsf S; \mathsf Z) \leq C(\alpha) }}\Prob{ \sum_{i = 1}^n \mathsf d(S_i, Z_i) > nd} \label{eq:Cuncoded}
\end{equation}
where  $P_{Z^n|S^n} = P_{\mathsf Z | \mathsf S} \times \ldots \times P_{\mathsf Z | \mathsf S}$. 
\end{thm}

\begin{proof}
The excess-distortion probability at blocklength $n$, distortion $d$ and cost $\alpha$ achievable among all single-letter codes 
$\left(P_{\mathsf X| \mathsf S},~ P_{\mathsf Z|\mathsf Y} \right)$ 
must satisfy
\begin{align}
 \epsilon &\geq \inf_{ \substack{P_{\mathsf X| \mathsf S}, P_{\mathsf Z|\mathsf Y}  \colon\\ \mathsf S - \mathsf X - \mathsf Y - \mathsf Z\\
 \E{\mathsf c(\mathsf X)} \leq \alpha }}\Prob{ \mathsf d_n(S^n, Z^n) > d }\\
 &\geq \inf_{ \substack{P_{\mathsf X| \mathsf S}, P_{\mathsf Z|\mathsf Y}  \colon\\ \E{\mathsf c(\mathsf X)} \leq \alpha\\ I(\mathsf S; \mathsf Z ) \leq I(\mathsf X; \mathsf Y) }}\Prob{\mathsf d_n(S^n, Z^n) > d } \label{eq:-C1a}
 \end{align}
where \eqref{eq:-C1a} holds since $\mathsf S - \mathsf X - \mathsf Y - \mathsf Z$ implies $ I(\mathsf S; \mathsf Z ) \leq I(\mathsf X; \mathsf Y)$ by the data processing inequality. The right side of \eqref{eq:-C1a} is lower bounded by the right side of \eqref{eq:Cuncoded} because $I(\mathsf X; \mathsf Y) \leq C(\alpha)$ holds for all $P_{\mathsf X}$ with $\E{\mathsf c(\mathsf X)} \leq \alpha$, and the distortion measure is separable. 
 \end{proof}
\begin{thm}[Gaussian approximation, optimal symbol-by-symbol code]
Assume $\E{\mathsf d^3\left( \mathsf S, \mathsf Z^\star\right)} < \infty$. Under restrictions \eqref{item:ch}-\eqref{item:match},
\begin{align}
 D_1(n, \epsilon, \alpha) &= D\left( C(\alpha) \right) + \sqrt \frac{\mathscr W_1(\alpha) }{n} \Qinv{\epsilon} + \frac {\theta_1(n)} n  \label{eq:2order1}\\
\mathscr W_1(\alpha) &= \Var{\mathsf d(\mathsf S, \mathsf Z^\star)}\label{eq:dispersion1}
\end{align}
where
\begin{equation}
 \theta_1(n) \leq \bigo{1} \label{eq:Aremainder1}
\end{equation}
 
 Moreover, if there is no power constraint,
\begin{align}
\theta_1(n) &\geq \frac{D^\prime(R)}{R^2}\theta(n)  \label{eq:Cremainder1}\\
\mathscr W_1 &= \mathscr W(1) \label{eq:V=V1}
\end{align}
where $\theta(n)$ is that in Theorem \ref{thm:2order}.

If $\Var{\mathsf d\left( \mathsf S, \mathsf Z^\star\right)} > 0$ and $\mathcal S$, $\hat {\mathcal {S}}$ are finite, then
 \begin{equation}
 \theta_1(n) \geq \bigo{1}\label{eq:Cremainder1finite}
\end{equation}
 \label{thm:2order1}
\end{thm}
\begin{proof} 
Since the third absolute moment of $\mathsf d(S_i, Z_i^\star)$ is finite, the achievability part of the result, namely, \eqref{eq:2order1} with the remainder satisfying \eqref{eq:Aremainder1},
 follows by a straightforward application of the Berry-Esseen bound to \eqref{eq:A1}, provided that $\Var{\mathsf d(S_i, Z_i^\star)} > 0$. If $\Var{\mathsf d(S_i, Z_i^\star)} = 0$, it follows trivially from \eqref{eq:A1}. 

To show the converse in \eqref{eq:Cremainder1}, observe that since the set of all $(n, n, d, \epsilon)$ codes includes all $(n, d, \epsilon)$ symbol-by-symbol codes, we have $D(n, n, \epsilon) \leq D_1(n, \epsilon)$. Since $\Qinv{\epsilon}$ is positive or negative depending on whether $\epsilon < \frac 1 2$ or $\epsilon > \frac 1 2$, using \eqref{eq:JSCCdispersiondistortion} we conclude that we must necessarily have \eqref{eq:V=V1}, which is, in fact, a consequence of conditions \eqref{item:b}, \eqref{item:c} in Section \ref{sec:prelim} and \eqref{item:match}.  Now, \eqref{eq:Cremainder1} is simply the converse part of \eqref{eq:2orderD}. 

The proof of the refined converse in \eqref{eq:Cremainder1finite} is relegated to Appendix \ref{appx:2order1}. 
\end{proof}

In the absence of a cost constraint, Theorem \ref{thm:2order1} shows that if the source and the channel are probabilistically matched in the sense of \cite{gastpar2003tocodeornot}, then not only does symbol-by-symbol transmission achieve the minimum average distortion, but also the dispersion of JSCC (see \eqref{eq:V=V1}). In other words, not only do such symbol-by-symbol codes attain the minimum average distortion but also the variance of distortions at the decoder's output is the minimum achievable among all codes operating at that average distortion. In contrast, if there is an average cost constraint, the symbol-by-symbol codes considered in Theorem \ref{thm:2order1} probably do not attain the minimum excess distortion achievable among all blocklength-$n$ codes, not even asymptotically. Indeed, as observed in \cite{polyanskiy2010thesis}, for the transmission of an equiprobable source over an AWGN channel under the average power constraint and the average block error probability performance criterion, the strong converse does not hold and the second-order term is of order $n^{-\frac 1 3}$, not $n^{-\frac 1 2}$, as in \eqref{eq:2order1}. 

Two conspicuous examples that satisfy the probabilistic matching condition \eqref{item:match}, so that symbol-by-symbol coding is optimal in terms of average distortion, are the transmission of a binary equiprobable source over a binary-symmetric channel provided the desired bit error rate is equal to the crossover probability of the channel \cite[Sec.11.8]{jelinek1968probabilistic}, \cite[Problem 7.16]{csiszar2011information}, and the transmission of a Gaussian source over an additive white Gaussian noise channel under the mean-square error distortion criterion, provided that the tolerable source signal-to-noise ratio attainable by an estimator is equal to the signal-to-noise ratio at the output of the channel \cite{Goblick1965theoretical}. We dissect these two examples next. After that, we will discuss two additional examples where uncoded transmission is optimal. 

\subsection{Uncoded transmission of a BMS over a BSC}
In the setup of Section \ref{sec:BMS_BSC}, if the binary source is unbiased $\left(p = \frac 1 2\right)$, then
$C = 1 - h(\delta)$,
$R(d) = 1 - h(d)$, 
and $D(C) = \delta$. If the encoder and the decoder are both identity mappings (uncoded transmission), the resulting joint distribution satisfies condition \eqref{item:match}. As is well known, regardless of the blocklength, the uncoded symbol-by-symbol scheme achieves the minimum bit error rate (averaged over source and channel). Here, we are interested instead in examining the excess distortion probability criterion.
For example, consider an application where, if the fraction of erroneously received bits exceeds a certain threshold, then the entire output packet is useless.

Using \eqref{eq:JSCCdispersiondistortion} and \eqref{eq:dispersion1}, it is easy to verify that 
\begin{equation}
 \mathscr W(1) = \mathscr W_1 = \delta (1 - \delta)
\end{equation}
 that is, uncoded transmission is optimal in terms of dispersion, as anticipated in \eqref{eq:V=V1}. Moreover, the uncoded transmission attains the minimum bit error rate threshold $D(n, n, \epsilon)$ achievable among all codes operating at blocklength $n$, regardless of the allowed $\epsilon$, as the following result demonstrates.
 
 \begin{figure}[htbp]
    \epsfig{file=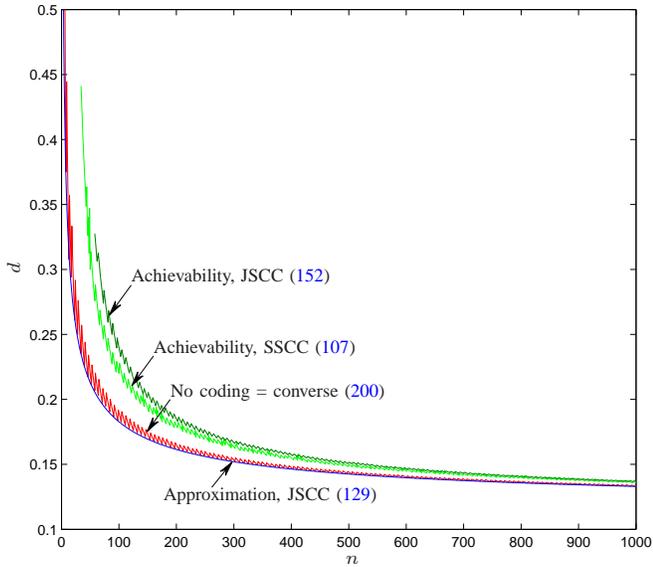,width=1\linewidth}
\caption{ Distortion-blocklength tradeoff for the transmission of a fair BMS over a BSC with crossover probability $\delta = 0.11$ and $R = 1$, $\epsilon = 10^{-2}$. }
\label{fig:DEBMS-BSC}
\end{figure}

\begin{figure}[htbp]
\subfigure[]{
    \epsfig{file=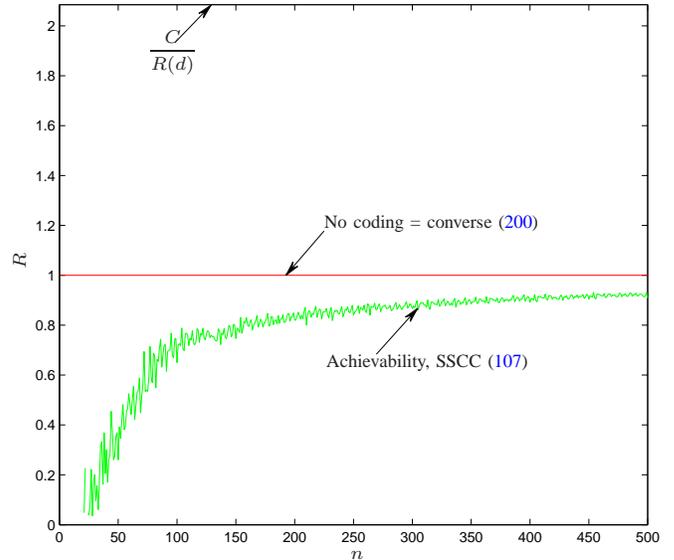, width=1\linewidth}
        \label{subfig:EBMS-BSC_b_R}
}
\subfigure[]{
    \epsfig{file=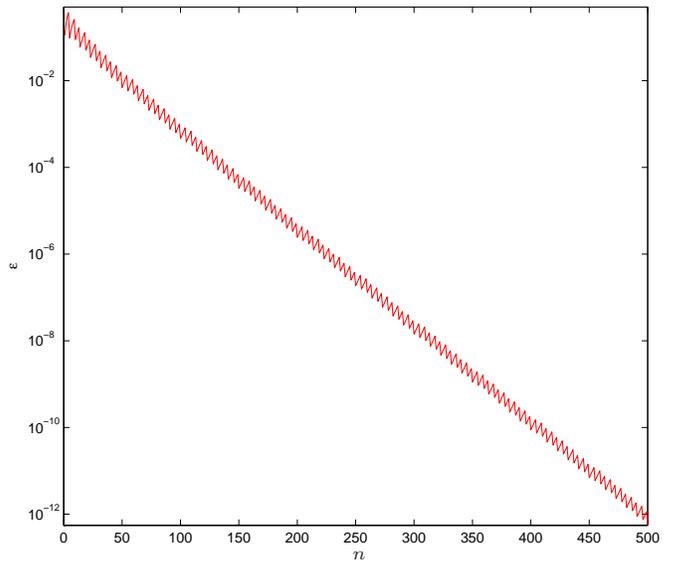, width=1\linewidth}
            \label{subfig:EBMS-BSC_b_E}
}
\caption{ Rate-blocklength  tradeoff \subref{subfig:EBMS-BSC_b_R} for the transmission of a fair BMS over a BSC with crossover probability $\delta = 0.11$ and $d = 0.22$. The excess-distortion probability $\epsilon$ is set to be the one achieved by the uncoded scheme \subref{subfig:EBMS-BSC_b_E}.  }
\label{fig:Special_EBMS-BSC_b}
\end{figure}

\begin{figure}[htbp]
    \epsfig{file=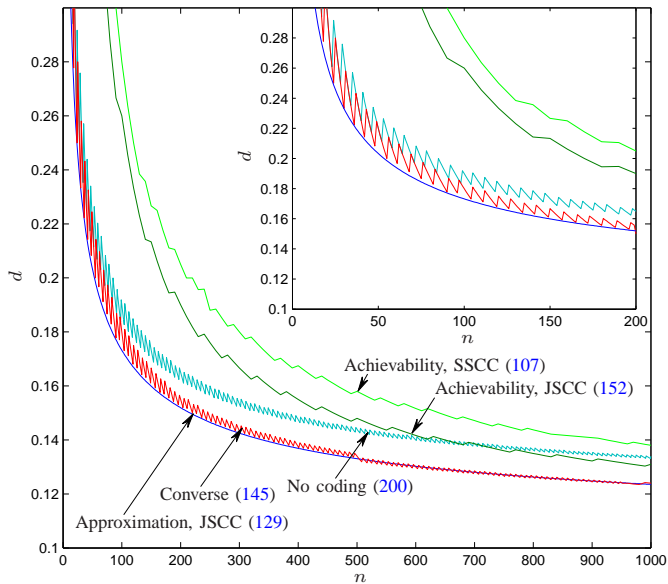,width=1\linewidth}
\caption{ Distortion-blocklength tradeoff for the transmission of a BMS with $p = \frac 2 5$ over a BSC with crossover probability $\delta = 0.11$ and $R = 1$, $\epsilon = 10^{-2}$. }
\label{fig:DBMS-BSC}
\end{figure}

 \begin{thm}[BMS-BSC, symbol-by-symbol code]
 Consider the the symbol-by-symbol scheme which is uncoded if $p \geq \delta$ and whose decoder always outputs the all-zero vector if $p < \delta$. It achieves, at blocklength $n$ and excess distortion probability $\epsilon$, regardless of $0 \leq p \leq \frac 1 2$, $\delta \leq \frac 1 2$,
 \begin{align}
&D_1(n, \epsilon) = \notag \\
 &\min \left\{ d \colon \sum_{t = 0}^{\lfloor nd \rfloor} {n \choose t} \min\{ p, \delta\}^t (1 - \min\{ p, \delta\})^{n - t} \geq 1 - \epsilon \right\} \label{eq:A1_BMS-BSC}
\end{align}
Moreover, if the source is equiprobable $\left(p = \frac 1 2\right)$, 
\begin{equation}
 D_1(n, \epsilon)  = D(n, n, \epsilon) \label{eq:optEBMS}
\end{equation}
\label{thm:A1_BMS-BSC}
\end{thm}
\begin{proof}
Direct calculation yields \eqref{eq:A1_BMS-BSC}. To show \eqref{eq:optEBMS}, let us compare $d^\star = D_1(n,\epsilon)$ with the conditions imposed on $d$ by Theorem \ref{thm:ChtEBMS-BSC}.
Comparing \eqref{eq:A1_BMS-BSC} to \eqref{eq:ChtEBMS-BSCrstar}, we see that either 
\begin{enumerate}[(a)]
 \item equality in \eqref{eq:A1_BMS-BSC} is achieved, $r^\star = nd^\star$, $\lambda = 0$, and (plugging $k = n$ into \eqref{eq:ChtEBMS-BSC})  
\begin{equation}
  \binosum{n}{ nd^\star} \leq \binosum{n}{ \lfloor nd \rfloor }
\end{equation}
 thereby implying that $d \geq d^\star$, or
 \item $r^\star =  nd^\star - 1$, $\lambda > 0$, and \eqref{eq:ChtEBMS-BSC} becomes
 \begin{equation}
 \lambda{ n \choose nd^\star   } + \binosum{n}{nd^\star - 1}  \leq \binosum{n}{ \lfloor nd \rfloor} \label{eq:-nocoding}
\end{equation}
which also implies $d \geq d^\star$. 
To see this, note that $d < d^\star$ would imply
$ \lfloor nd \rfloor \leq n d^\star - 1$
since $nd^\star$ is an integer, which in turn would require (according to \eqref{eq:-nocoding}) that $\lambda \leq 0$, which is impossible.
\end{enumerate}
\end{proof}

For the transmission of the fair binary source over a BSC, Fig. \ref{fig:DEBMS-BSC} shows the distortion achieved by the uncoded scheme, the separated scheme and the JSCC scheme of Theorem \ref{thm:ABMS-BSC} versus $n$ for a fixed excess-distortion probability $\epsilon = 0.01$. The no coding / converse curve in Fig. \ref{fig:DEBMS-BSC} depicts one of those singular cases where the non-asymptotic fundamental limit can be computed precisely. As evidenced by this curve, the fundamental limits need not be monotonic with blocklength. 

Figure  \ref{subfig:EBMS-BSC_b_R} shows the rate achieved by separate coding when $d > \delta$ is fixed, and the excess-distortion probability $\epsilon$, shown in Fig. \ref{subfig:EBMS-BSC_b_E}, is set to be the one achieved by uncoded transmission, namely, \eqref{eq:A1_BMS-BSC}.  Figure \ref{subfig:EBMS-BSC_b_R}  highlights the fact that at short blocklengths (say $n \leq 100$) separate source/channel coding is vastly suboptimal. As the blocklength increases, the performance of the separated scheme approaches that of the no-coding scheme, but according to Theorem \ref{thm:A1_BMS-BSC} it can never outperform it. Had we allowed the excess distortion probability to vanish sufficiently slowly, the JSCC curve would have approached the Shannon limit as $n \to \infty$.
However, in Figure  \ref{subfig:EBMS-BSC_b_R}, the exponential decay in $\epsilon$ is such that there is indeed an asymptotic rate penalty as predicted in \cite{csiszar1982error}.

For the biased binary source with $p = \frac 2 5$ and BSC with crossover probability $0.11$, Figure \ref{fig:DBMS-BSC} plots the maximum distortion achieved with probability $0.99$ by the uncoded scheme, which in this case is asymptotically suboptimal. Nevertheless, uncoded transmission performs remarkably well in the displayed range of blocklengths, achieving the converse almost exactly at blocklengths less than $100$, and outperforming the JSCC achievability result in Theorem \ref{thm:ABMS-BSC} at blocklengths as long as $700$. This example substantiates that even in the absence of a probabilistic match between the source and the channel,  symbol-by-symbol transmission, though asymptotically suboptimal, might outperform SSCC and even our random JSCC achievability bound in the finite blocklength regime.


\subsection{Symbol-by-symbol coding for lossy transmission of a GMS over an AWGN channel}
In the setup of Section \ref{sec:GMS_AWGN}, using \eqref{eq:R(d)GMS} and \eqref{eq:CapacityAWGN}, we find that
\begin{equation}
 D(C(P)) = \frac{\sigma_{\mathsf S}^2}{1 + P} \label{eq:dbarGMS-AWGN}
\end{equation}
The next result characterizes the distribution of the distortion incurred by the symbol-by-symbol scheme that attains the minimum average distortion.
\begin{thm}[GMS-AWGN, symbol-by-symbol code]
The following  symbol-by-symbol transmission scheme in which the encoder and the decoder are the amplifiers:
\begin{align}
 \mathsf  f_1( \mathsf s) &= a \mathsf s, ~ a^2 = \frac{P \sigma_{\mathsf N}^2}{\sigma_{\mathsf S}^2}\\
 \mathsf g_1(\mathsf y) &= b \mathsf y, ~b = \frac{a \sigma_{\mathsf S}^2}{a^2 \sigma_{\mathsf S}^2 + \sigma_{\mathsf N}^2}
\end{align}
is an $(n, d, \epsilon, P)$ symbol-by-symbol code (with average cost constraint) such that
\begin{equation}
 \Prob{W_0^n D(C(P))> n d } = \epsilon \label{eq:d1GMS-AWGN}
\end{equation}
where $W_0^n$ is chi-square distributed with $n$ degrees of freedom. 
\label{thm:GMS-AWGN_d1}
\end{thm}

Note that \eqref{eq:d1GMS-AWGN} is a particularization of \eqref{eq:A1}. Using \eqref{eq:d1GMS-AWGN}, we find that 
\begin{equation}
 \mathscr W_1(P) = 2 \frac{\sigma_{\mathsf S}^4}{\left( 1 + P\right)^2}
\end{equation}
On the other hand, using \eqref{eq:JSCCdispersiondistortion}, we compute
\begin{align}
 \mathscr W(1, P) &= 2 \frac{\sigma_{\mathsf S}^4}{\left( 1 + P\right)^2} \left(2 - \frac 1 {\left( 1+P\right)^2} \right) \\
 &>  \mathscr W_1(P) \label{eq:GMS-AWGN W1>W}
\end{align}
The difference between \eqref{eq:GMS-AWGN W1>W} and \eqref{eq:V=V1}  is due to the fact that the optimal symbol-by-symbol code in Theorem \ref{thm:GMS-AWGN_d1} obeys an average power constraint, rather than the more stringent maximal power constraint of Theorem \ref{thm:2order}, so it is not surprising that for the practically interesting case $\epsilon < \frac 1 2$ the symbol-by-symbol code can outperform the best code obeying the maximal power constraint. Indeed, in the range of blocklenghts displayed in Figure \ref{fig:DGMS_AWGN}, the symbol-by-symbol code even outperforms the converse for codes operating under a maximal power constraint. 
\begin{figure}[htbp]
    \epsfig{file=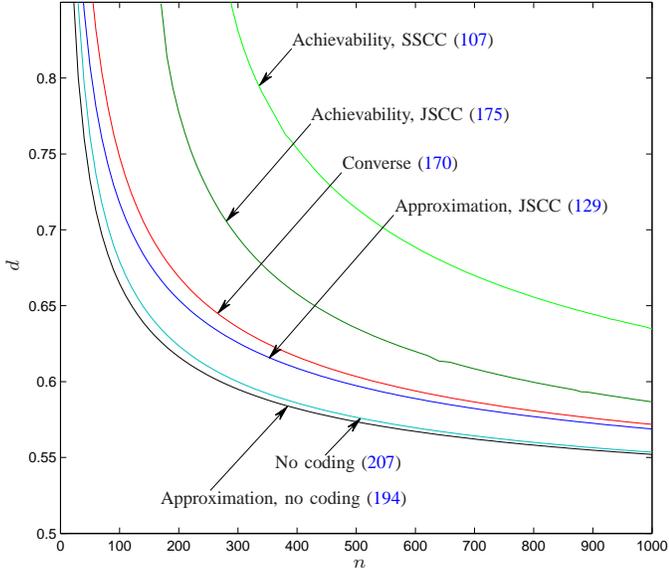,width=1\linewidth}
\caption{ Distortion-blocklength tradeoff for the transmission of a GMS over an AWGN channel with $\frac {P}{\sigma_{\mathsf N}^2} = 1$ and $R = 1$, $\epsilon = 10^{-2}$. }
\label{fig:DGMS_AWGN}
\end{figure}

\subsection{Uncoded transmission of a discrete memoryless source (DMS) over a discrete erasure channel (DEC) under erasure distortion measure}
\label{sec:erasureuncoded}
For a discrete source, the single-letter erasure distortion measure is defined as the following mapping $\mathsf d \colon \mathcal S \times \{\mathcal S, \mathsf e\} \mapsto [0, \infty]$:\footnote{The distortion measure in \eqref{eq:derasure} is a scaled version of the erasure distortion measure found in literature, e.g. \cite{gallager1968information}.}  
\begin{equation}
 \mathsf d(\mathsf s, \mathsf z) = 
\begin{cases}
 0 & \mathsf z = \mathsf s\\
H(\mathsf S) & \mathsf z = \mathsf e\\
 \infty & \text{otherwise} \label{eq:derasure}
\end{cases}
\end{equation}
For any $0 \leq d \leq H(\mathsf S)$, the rate-distortion function of the equiprobable source is achieved by 
\begin{equation}
 P_{\mathsf Z^\star| \mathsf S = \mathsf s}(\mathsf z) = 
\begin{cases}
 1 - \frac{d}{H(\mathsf S)} & \mathsf z = \mathsf s\\
 \frac{d}{H(\mathsf S)} & \mathsf z = \mathsf e \label{eq:PY|Xstarerasure}
\end{cases}
\end{equation}
The rate-distortion function and the $\mathsf d$-tilted information for the equiprobable source with the erasure distortion measure are given by, respectively, 
\begin{align}
  R(d) &= H(\mathsf S) - d \label{eq:Rderasure}\\
 \jmath_{\mathsf S}(\mathsf s, d) &= \imath_{\mathsf S}(\mathsf s) - d \label{eq:dtiltederasure}
\end{align}
Note that, trivially, $\jmath_{\mathsf S}(\mathsf S, d) = R(d) = \log |\mathcal S| - d$ a.s.
The channel that is matched to the equiprobable DMS with the erasure distortion measure is the DEC, whose single-letter transition probability kernel $ P_{\mathsf Y| \mathsf X} \colon \mathcal A \mapsto \{\mathcal A, \mathsf e\}$ is
\begin{equation}
 P_{\mathsf Y| \mathsf X = \mathsf x}(\mathsf y) = 
\begin{cases}
 1 - \delta & \mathsf y = \mathsf x\\
\delta & \mathsf y = \mathsf e 
\end{cases}
\label{eq:PY|Xstarerasure}
\end{equation}
and whose capacity is given by $C = \log |\mathcal A| - \delta$, achieved by  equiprobable $P_{\mathsf X^\star}$. For $P_{\mathsf S}$ equiprobable on $\mathcal S = \mathcal A$, we find that $D(C) = \delta \log \left| \mathcal S\right|$, and 
\begin{equation}
\mathscr W_1 =   \delta \left( 1 - \delta \right) \log^2 \left| \mathcal S\right| \label{eq:W1erasure}
\end{equation}

\subsection{Symbol-by-symbol transmission of a DMS over a DEC under logarithmic loss}
\label{sec:loglossuncoded}
Let the source alphabet $\mathcal S$ be finite, and let the reproduction alphabet $\hat {\mathcal S}$ be the set of all probability distributions on $\mathcal S$. The single-letter logarithmic loss distortion measure
$\mathsf d \colon \mathcal S \times \hat{\mathcal S} \mapsto \mathbb R^+$ 
is defined by \cite{harremoes2007information,courtade2011multiterminal}
\begin{equation}
\mathsf d( s, P_{\mathsf Z}) = \imath_{\mathsf Z}(s)
\end{equation}
 
Curiously, for any $0 \leq d \leq H(\mathsf S)$, the rate-distortion function and the $\mathsf d$-tilted information are given respectively by \eqref{eq:Rderasure} and \eqref{eq:dtiltederasure}, even if the source is not equiprobable.  In fact, the rate-distortion function is achieved by, 
\begin{equation}
 P_{P_{\mathsf Z}^\star | \mathsf S = s}(P_{\mathsf Z}) = 
\begin{cases}
 \frac{d}{H(\mathsf S)}  & P_{\mathsf Z} = P_{\mathsf S}  \\
 1 -  \frac{d}{H(\mathsf S)}  & P_{\mathsf Z} = 1_{\mathsf S}(\mathsf s)
\end{cases}
\end{equation}
and the channel that is matched to the equiprobable source under logarithmic loss is exactly the DEC in \eqref{eq:PY|Xstarerasure}. Of course, unlike Section \ref{sec:erasureuncoded}, the decoder we need is a simple one-to-one function that outputs $P_{\mathsf S}$ if the channel output is $\mathsf e$, and $1_{\mathsf S}(\mathsf y)$ otherwise, where $\mathsf y \neq \mathsf e$ is the output of the DEC. Finally, it is easy to verify that the distortion-dispersion function of symbol-by-symbol coding under logarithmic loss is the same as that under erasure distortion and is given by \eqref{eq:W1erasure}. 

\section{Conclusion}
\label{sec:conclusion}
In this paper we gave a non-asymptotic analysis of joint source-channel coding including several achievability and converse bounds, which hold
in wide generality and are tight enough to determine the dispersion of joint source-channel coding for the transmission of an abstract memoryless source over either a DMC or a Gaussian channel, under an arbitrary fidelity measure.  We also investigated the penalty incurred by separate source-channel coding using both the source-channel dispersion and the particularization of our new bounds to (i) the binary source and the binary symmetric channel with bit error rate fidelity criterion and (ii) the Gaussian source and Gaussian channel under mean-square error distortion. Finally, we showed cases where symbol-by-symbol (uncoded) transmission beats any other known scheme in the finite blocklength regime even when the source-channel matching condition is not satisfied.

The approach taken in this paper to analyze the non-asymptotic fundamental limits of lossy joint source-channel coding is two-fold. Our new achievability and converse bounds apply to  abstract sources and channels and allow for memory, while the asymptotic analysis of the new bounds leading to the dispersion of JSCC is focused on the most basic scenario of transmitting a stationary memoryless source over a stationary memoryless channel. 

The major results and conclusions are the following. 

\begin{enumerate}[1)]
 \item A general new converse bound (Theorem \ref{thm:CT}) leverages the concept of $\mathsf d$-tilted information (Definition \ref{defn:id}), a random variable which corresponds (in a sense that can be formalized \cite{kontoyiannis2000pointwise,kostina2011fixed}) to the number of bits required to represent a given source outcome within distortion $d$ and whose role in lossy compression is on a par with that of information (in \eqref{eq:i}) in lossless compression. 
 \item The converse result in Theorem \ref{thm:Cht} capitalizes on two simple observations, namely, that any $(d, \epsilon)$ lossy code can be converted to a list code with list error probability $\epsilon$, and that a binary hypothesis test between $P_{SXY}$ and an auxiliary distribution on the same space can be constructed by choosing $P_{SXY}$ when there is no list error. We have generalized the conventional notion of list, to allow
the decoder to output a possibly uncountable set of source realizations.
\item As evidenced by our numerical results, the converse result in Theorem \ref{thm:Chtsym}, which applies to those channels satisfying a certain symmetry condition and which is a consequence of the hypothesis testing converse in Theorem \ref{thm:Cht}, can outperform the $\mathsf d$-tilted information converse in Theorem \ref{thm:CT}. Nevertheless, it is Theorem \ref{thm:CT} that lends itself to analysis more easily and that leads to the JSCC dispersion for the general DMC.   
\item Our random-coding-based achievability bound (Theorem \ref{thm:Ag}) provides insights into the degree of separation between the source and the channel codes required for optimal performance in the finite blocklength regime. More precisely, it reveals that the dispersion of JSCC can be achieved in the class of $(M, d, \epsilon)$ JSCC codes (Definition \ref{defn:jsccM}). As in separate source/channel coding, in $(M, d, \epsilon)$ coding the inner channel coding block is connected to the outer source coding block by a noiseless link of capacity $\log M$, but unlike SSCC, the channel (resp. source) code can be chosen based on the knowledge of the source (resp. channel). The conventional SSCC in which the source code is chosen without knowledge of the channel and the channel code is chosen without knowledge of the source, although known to achieve the asymptotic fundamental limit of joint source-channel coding under certain quite weak conditions, is in general suboptimal in the finite blocklength regime.
\item Since 
$ \E{\mathsf d(S, Z)} =\int_0^\infty \Prob{\mathsf d(S, Z) > \xi } d\xi$,
bounds for average distortion can be obtained by integrating our bounds on excess distortion. Note, however, that the code that minimizes $\Prob{\mathsf d(S, Z) > \xi }$ depends on $\xi$.
Since the distortion cdf of any single code does not majorize the cdfs of all possible codes, the converse bound on the average distortion obtained through this approach, although asymptotically tight, may be loose at short blocklengths. Likewise, regarding achievability bounds (e.g. \eqref{eq:A}), the optimization over channel and source random codes, $P_X$ and $P_Z$, must be performed after the integration, so that the choice of code does not depend on the distortion threshold $\xi$.  
\item For the transmission of a stationary memoryless source over a stationary memoryless channel, the Gaussian approximation in Theorem \ref{thm:2order} (neglecting the remainder $\theta(n)$) provides a simple estimate of the maximal non-asymptotically achievable joint source-channel coding rate. Appealingly, the dispersion of joint source-channel coding decomposes into two terms, the channel dispersion and the source dispersion. Thus, only two channel attributes, the capacity and dispersion, and two source attributes, the rate-distortion and rate-dispersion functions, are required to compute the Gaussian approximation to the maximal JSCC rate. 
\item In those curious cases where the source and the channel are probabilistically matched so that symbol-by-symbol coding attains the minimum possible average distortion, Theorem \ref{thm:2order1} ensures that it also attains the dispersion of joint source-channel coding, that is, symbol-by-symbol coding results in the minimum variance of distortions among all codes operating at that average distortion.
\item Even in the absence of a probabilistic match between the source and the channel, symbol-by-symbol transmission, though asymptotically suboptimal, might outperform separate source-channel coding and joint source-channel random coding in the finite blocklength regime.
\end{enumerate}
\section*{Acknowledgement}
The authors are grateful to Dr. Oliver Kosut for offering numerous comments, and,  in particular, suggesting the simplification of the achievability bound in \cite{kostina2012ISITjscc} with the tighter version in Theorem \ref{thm:A}. 
\appendices
\section{The Berry-Esseen theorem}
\label{appx:Berry-Esseen}
The following result is an important tool in the Gaussian approximation analysis.
\begin{thm}[{Berry-Esseen CLT, e.g. \cite[Ch. XVI.5 Theorem 2]{feller1971introduction}}]
\label{thm:Berry-Esseen}
Fix a positive integer $n$. Let $W_i$, $i = 1, \ldots, n$ be independent. Then, for any real $t$
\begin{equation}
\left| \mathbb P \left[ \sum_{i = 1}^n W_i > n \left( D_n + t \sqrt {\frac{V_n}{ n}}\right) \right]  - Q(t) \right| \leq \frac {B_n}{\sqrt n},
\label{eq:BerryEsseen}
\end{equation}
where
\begin{align}
D_n &= \frac 1 n \sum_{i = 1}^n \E{ W_i} \label{eq:BerryEsseenDn} \\
V_n &= \frac 1 n \sum_{i = 1}^n \Var{W_i} \label{eq:BerryEsseenVn} \\
T_n &= \frac 1 n \sum_{i = 1}^n \E{ |W_i - \E{W_i} |^3 } \label{eq:BerryEsseenTn} \\
B_n &=  \frac{c_0 T_n}{V_n^{3/2}} \label{eq:BerryEsseenBn}
\end{align}
and $0.4097 \leq c_0 \leq 0.5600$ ($0.4097 \leq c_0 < 0.4784$ for identically distributed $W_i$). 
\end{thm}

\section{Auxiliary result on the minimization of the information spectrum}
\label{appx:lemmaC}
Given a finite set $\mathcal{A}$, we say that $x^n \in \mathcal A^n$ has type $P_{\mathsf X}$ if the number of times each letter $a \in \mathcal A$ is encountered in $x^n$ is $n P_{\mathsf X}(a)$. Let $\mathcal P$ be the set of all distributions on $\mathcal A$, which is simply the standard $|\mathcal A| -1$ simplex in $\mathbb R^{|\mathcal A|}$. For an arbitrary subset $\mathcal D \subseteq \mathcal P$, denote by $\mathcal D_{[n]}$ the set of distributions in $\mathcal D$ that are also $n$-types, that is, 
\begin{equation}
  \mathcal D_{[n]} = \left\{ P_{\mathsf X} \in \mathcal D \colon \forall  a \in \mathcal A,~ n P_{\mathsf X}(a) \text{ is an integer}  \right\} \label{eq:-Dn}
\end{equation}
Denote by $\Pi(P_{\mathsf X})$ the minimum Euclidean distance approximation of $P_{\mathsf X} \in \mathcal P$ in the set of $n$-types, that is,  
\begin{equation}
\Pi(P_{\mathsf X}) = \arg\min_{\hat P_{\mathsf X} \in \mathcal P_{[n]} } \left|P_{\mathsf X} - \hat P_{\mathsf X}\right| \label{eq:-Pi}
\end{equation}

Let $\mathcal P^\star$ be the set of capacity-achieving distributions:\footnote{In this appendix, we dispose of the assumption \eqref{item:c} in Section \ref{sec:prelim} that the capacity-achieving input distribution is unique. }
\begin{equation}
 \mathcal P^\star = \left\{ P_{\mathsf X} \in \mathcal P \colon I(\mathsf X; \mathsf Y) = C\right\}
\end{equation}
Denote the minimum (maximum) information variances achieved by the distributions in $\mathcal P^\star$ by
\begin{align}
 V_{\min} &= \min_{P_{\mathsf X} \in  \mathcal P^\star} \Var{\imath_{\mathsf X; \mathsf Y}(\mathsf X; \mathsf Y)}\\
  V_{\max} &= \max_{P_{\mathsf X} \in  \mathcal P^\star} \Var{\imath_{\mathsf X; \mathsf Y}(\mathsf X; \mathsf Y)}
\end{align}
and let $\mathcal P^\star_{\min} \subseteq \mathcal P^\star$ be the set of capacity-achieving distributions that achieve the minimum information variance:
\begin{equation}
 \mathcal P^\star_{\min} = \left\{ P_{\mathsf X} \in \mathcal P^\star \colon \Var{\imath_{\mathsf X; \mathsf Y}(\mathsf X; \mathsf Y)} = V_{\min} \right\}
\end{equation}
and analogously $\mathcal P^\star_{\max}$ for the distributions in $\mathcal P^\star$ with maximal variance.
Lemma \ref{lemma:C} below allows to show that in the memoryless case, the infimum inside the expectation in \eqref{eq:CT} with $W = \mathrm{type}\left(X^n\right)$ and $P_{\bar Y^n|W = P_{\mathsf X}} = P_{\mathsf Y} \times \ldots \times P_{\mathsf Y}$, where $P_{\mathsf Y}$ is the output distribution induced by the type $P_{\mathsf X}$,  is approximately attained by those sequences whose type is closest to the capacity-achieving distribution $P_{\mathsf X^\star}$ (if it is non-unique, $P_{\mathsf X^\star}$ is chosen appropriately based on the information variance it achieves). This technical result is the key to proving the converse part of Theorem \ref{thm:2order}. 

\begin{lemma} 
There exist $\bar \Delta > 0$ such that for all sufficiently large $n$: 
\begin{enumerate}
 \item 
 If $V_{\min} > 0$,  then there exists $K > 0$ such that for $\left|\Delta\right| \leq \bar \Delta$,  
\begin{align}
 &~\min_{x^n \in \mathcal A^n} \Prob{\sum_{i = 1}^n \imath_{\mathsf X; \mathsf Y}(x_i; Y_i) \leq  n\left(C - \Delta\right)} 
 \notag\\
 \geq&~   \Prob{\sum_{i = 1}^n \imath_{\mathsf X; \mathsf Y}(x_i^\star; Y_i) \leq  n\left(C - \Delta\right)} -  \frac {K} {\sqrt n} \label{eq:lemmaC}
\end{align}
where \eqref{eq:lemmaC} holds for any $x^{n \star}$ with $\mathrm{type}(x^{n \star}) = \Pi(P_{\mathsf X}^\star)$ for $P_{\mathsf X}^\star \in \mathcal P^\star_{\min}$ if $\Delta \geq 0$ and $P_{\mathsf X}^\star \in \mathcal P^\star_{\max}$ if $\Delta < 0$. 
\item
If $V_{\max} = 0$,  then for all $0 < \alpha < \frac 3 2$ and 
$\Delta \geq \frac{\bar \Delta}{n^{\frac 1 2 + \alpha}}$, 
\begin{equation}
 \min_{x^n \in \mathcal A^n} \Prob{\sum_{i = 1}^n \imath_{\mathsf X; \mathsf Y}(x_i; Y_i) \leq  n\left(C + \Delta\right)} \geq   1  -  \frac {1} {n^{ \frac 1 4 - \frac 3 2 \alpha}} \label{eq:lemmaC0}
\end{equation}
\end{enumerate}
The information densities in the left sides of 
\eqref{eq:lemmaC} and \eqref{eq:lemmaC0}
are computed with $\left\{P_{\mathsf Y| \mathsf X = x_i}, P_{\mathsf Y}\right\}$, where $P_{\mathsf Y}$ is induced by the type of $x^n$, i.e.
$\mathrm{type}(x^{n}) = P_{\mathsf X} \to P_{\mathsf Y| \mathsf X} \to P_{\mathsf Y}$,
and that in the right side of \eqref{eq:lemmaC} is computed with $\left\{P_{\mathsf Y| \mathsf X = x_i^\star}, P_{\mathsf Y}\right\}$, where $P_{\mathsf Y}$ is induced by the type of $x^{n\star}$, i.e.
$\mathrm{type}(x^{n \star}) = P_{\mathsf X} \to P_{\mathsf Y| \mathsf X} \to P_{\mathsf Y}$. The independent random variables $Y_i$ in the left sides of \eqref{eq:lemmaC} and \eqref{eq:lemmaC0} have distribution $P_{\mathsf Y| \mathsf X =x_i}$, while $Y_i$ in the right side of \eqref{eq:lemmaC} have distribution 
$P_{\mathsf Y| \mathsf X = x^\star_i}$. 
 \label{lemma:C}
 \end{lemma}

In order to prove Lemma \ref{lemma:C}, we first show three auxiliary lemmas. The first two deal with approximate optimization of functions. 

If $f$ and $g$ approximate each other, and the minimum of $f$ is approximately attained at $x$, then $g$ is also approximately minimized at $x$, as the following lemma formalizes. 
\begin{lemma}
Fix $\eta > 0$, $\xi > 0$. Let $\mathcal D$ be an arbitrary set, and let $f\colon \mathcal {D} \mapsto \mathbb R$ and $g \colon \mathcal D \mapsto \mathbb R$ be such that
\begin{align}
 \sup_{x \in \mathcal {D}} \left| f(x) - g(x)\right| &\leq \eta \label{eq:lemmafg1}
\end{align}
Further, assume that $f$ and $g$ attain their minima. Then,  
\begin{equation}
 g(x) \leq \min_{y \in \mathcal D} g(y) + \xi + 2\eta \label{eq:lemmafg}
\end{equation}
as long as $x$ satisfies
\begin{equation}
 f(x) \leq \min_{y \in \mathcal D} f(y) + \xi \label{eq:lemmafg2}
\end{equation}
(see Fig. \ref{fig:lemma}).
\label{lemma_fg}
\end{lemma}

\begin{figure}[htbp]
\begin{center}
    \epsfig{file=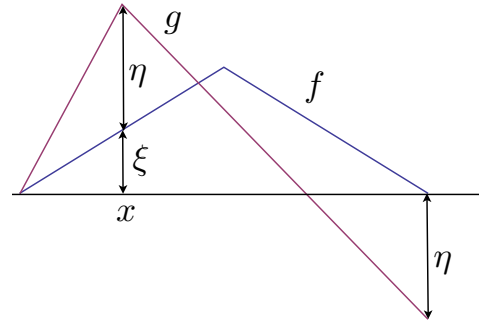,width=.7\linewidth}
\caption{An example where \eqref{eq:lemmafg} holds with equality. }
\label{fig:lemma}
\end{center}
\end{figure}

\begin{proof}[Proof of Lemma \ref{lemma_fg}]
Let $x^\star \in \mathcal D$ be such that $ g(x^\star) = \min_{y \in \mathcal D} g(y)$. Using \eqref{eq:lemmafg1} and \eqref{eq:lemmafg2}, write 
\begin{align}
g(x) &\leq \min_{y \in \mathcal D} f(y) + g(x) - f(x) + \xi\\
&\leq \min_{y \in \mathcal D} f(y)  + \eta + \xi\\
&\leq f(x^\star) + \eta + \xi\\
&= g(x^\star) - g(x^\star) + f(x^\star) + \eta + \xi\\
&\leq g(x^\star) + 2 \eta + \xi
\end{align}
\end{proof}

The following lemma is reminiscent of \cite[Lemma 64]
{polyanskiy2010channel}. 

\begin{lemma}
Let $\mathcal D$ be a compact metric space, and let 
$d \colon \mathcal{D}^2 \to \mathbb R^+$
 be a metric. Fix $f \colon \mathcal D \mapsto \mathbb R$ and $g \colon \mathcal D \mapsto \mathbb R$. Let 
\begin{equation}
 \mathcal D^\star = \left\{x \in \mathcal{D} \colon f(x) = \max_{y \in \mathcal D} f(y) \right\}
\end{equation}
Suppose that for some constants $\ell > 0, L > 0$, we have, for all $(x, x^\star) \in \mathcal D \times \mathcal D^\star$,
\begin{align}
f(x^\star) - f(x) &\geq \ell d^2(x, x^\star) \label{eq:lemmaf}\\
\left| g(x^\star) - g(x)\right| &\leq L d(x, x^\star)\label{eq:lemmag}
\end{align}
Then, for any positive scalars $\varphi, \psi$, 
\begin{equation}
 \max_{x \in \mathcal D} \left\{ \varphi f(x) \pm \psi g(x)\right\} \leq \varphi f(x^\star) \pm \psi g(x^\star) + \frac{L^2 \psi^2}{4 \ell \varphi}\label{eq:lemmamaxsum}
\end{equation}

Moreover, if, instead of \eqref{eq:lemmaf}, $f$ satisfies
\begin{equation}
 f(x^\star) - f(x) \geq \ell d(x, x^\star) \label{eq:lemmaflinear}\
\end{equation}
then, for any positive scalars $\psi$, $\varphi$ such that 
\begin{equation}
 L \psi \leq \ell \varphi \label{eq:lemmaphi}
\end{equation}
we have
\begin{equation}
 \max_{x \in \mathcal D} \left\{ \varphi f(x) \pm \psi g(x)\right\} = \varphi f(x^\star) \pm \psi g(x^\star)\label{eq:lemmamaxsumlinear}
\end{equation}
\label{lemma:maxsum}
\end{lemma}

\begin{proof}[Proof of Lemma \ref{lemma:maxsum}]
Let $x_0$ achieve the maximum on the left side of \eqref{eq:lemmamaxsum}. 
Using \eqref{eq:lemmaf} and \eqref{eq:lemmag}, we have, for all $x^\star \in \mathcal D^\star$, 
\begin{align}
 0 &\leq \varphi \left( f(x_0) - f(x^\star)\right) \pm \psi \left( g(x_0) - g(x^\star)\right) \\
 &\leq -\ell \varphi d^2(x_0, x^\star) + L \psi d(x_0, x^\star) \label{eq:-lemma1}\\
 &\leq \frac{L^2 \psi^2}{4 \ell \varphi} \label{eq:-lemma1a}
\end{align}
where \eqref{eq:-lemma1a} follows because the maximum of \eqref{eq:-lemma1} is achieved at $d(x_0, x^\star) = \frac{L \psi}{2 \ell \varphi}$. 

To show \eqref{eq:lemmamaxsumlinear}, observe using \eqref{eq:lemmaflinear} and \eqref{eq:lemmag} that
\begin{align}
 0 &\leq \varphi \left( f(x_0) - f(x^\star)\right) \pm \psi \left( g(x_0) - g(x^\star)\right) \\
  &\leq \left(-\ell \varphi + L \psi \right) d(x_0, x^\star) \\
 &\leq 0 \label{eq:-lemmalinear}
\end{align}
where \eqref{eq:-lemmalinear} follows from \eqref{eq:lemmaphi}. 

\end{proof}

The following lemma deals with asymptotic behavior of the $Q$-function. 
\begin{lemma}
Fix $a \geq 0$, $b \geq 0$. Then, there exists $q \geq 0$ (explicitly computed in the proof) such that for all $z \geq -\frac{\sqrt n}{2 b}$ and all $n$ large enough, 
\begin{equation}
  Q \left(z - \frac a {\sqrt n}\right)
 -
 Q \left( z + \frac {b}{\sqrt n} z^2   \right)
 \leq
  \frac{q}{\sqrt n}
\end{equation}
\label{lemma:Qbound}
\end{lemma}

\begin{proof}[Proof of Lemma \ref{lemma:Qbound}]
 
$Q(x)$ is convex for $x \geq 0$, and
$
Q^\prime(x) = - \frac 1 {\sqrt{2 \pi}} e^{-\frac {x^2}{2}} 
$,
so for $x \geq 0$, $\xi \geq 0$
\begin{equation}
Q(x + \xi) \geq Q(x) - \frac{\xi}{\sqrt{2 \pi}}  e^{-\frac {x^2}{2}} \label{eq:Q x>0 xi>0}
\end{equation}
while for arbitrary $x$ and $\xi \geq 0$, 
\begin{equation}
Q(x + \xi) \geq Q(x) - \frac{\xi}{\sqrt{2 \pi}} \label{eq:Q xi>0}
\end{equation} 

 If $z \geq \frac{a}{\sqrt n}$,  we use \eqref{eq:Q x>0 xi>0} to obtain
\begin{align}
 &~ Q \left(z - \frac a {\sqrt n}\right)
 -
 Q \left(  z + \frac {b}{\sqrt n} z^2 \right)\\
  \leq &~
  \frac{1}{\sqrt{2 \pi}}  e^{-\frac {\left( z - \frac a {\sqrt n}\right)^2 }{2}} \left( \frac{b}{\sqrt n} z^2 + \frac{a}{\sqrt{n} } \right) \label{eq:-Qa}\\
 \leq &~
  \frac{bz^2}{\sqrt{2 \pi n}}  e^{-\frac {\left( z - \frac a {\sqrt n}\right)^2 }{2}} + \frac{a}{\sqrt{2 \pi n} }\label{eq:-Qa}\\
  \leq &~   \frac{3 b\, e^{-1} + a}{\sqrt{2 \pi n}} \label{eq:-Qb}
\end{align}
where \eqref{eq:-Qb} holds for $n$ large enough because the maximum of \eqref{eq:-Qa} is attained at 
$z =\sqrt{2 + \frac a {4n}} + \frac a {2 \sqrt n}$. 

If $0 \leq z \leq \frac{a}{\sqrt n}$, we
 use \eqref{eq:Q xi>0} to obtain
\begin{align}
 &~Q \left(z - \frac a {\sqrt n}\right)
 -
 Q \left( z + \frac {b}{\sqrt n} z^2   \right)
 \notag\\
 \leq &~ \frac {1} {\sqrt{2 \pi}}\left( \frac{b}{\sqrt n} z^2 + \frac{a}{\sqrt{n} } \right) \\
 \leq&~ \frac {a} {\sqrt{2 \pi n}} \left( 1 + \frac {ab} n \right)
\end{align}
If
$
-\frac{\sqrt n}{2 b} \leq z \leq 0
$,
we use $Q(x) = 1 - Q(-x)$ to obtain
\begin{align}
 &~ Q \left(z - \frac a {\sqrt n}\right)
 -
 Q \left(  z + \frac {b}{\sqrt n} z^2 \right) \notag\\
 = &~ 
  Q \left(  |z| - \frac {b}{\sqrt n} z^2 \right)
  -
  Q \left(|z| + \frac a {\sqrt n}\right)
\\
  \leq &~
 \frac{1}{\sqrt{2 \pi}}  e^{-\frac {z^2 \left(1 - \frac b {\sqrt n} |z|\right)^2}{2}} \left( \frac{b}{\sqrt n} z^2 + \frac{a}{\sqrt{n} } \right) \\
   \leq &~
 \frac{bz^2}{\sqrt{2 \pi n}}  e^{-\frac {z^2 \left(1 - \frac b {\sqrt n} |z|\right)^2}{2}} + \frac a {\sqrt{2 \pi n}} \label{eq:-Qc}\\
 \leq &~  \frac{b z^2}{\sqrt{2 \pi n}} e^{ - \frac {z^2 }{8} } + \frac a {\sqrt{2 \pi n}}  \label{eq:-Qd}\\
 \leq &~  \frac{8 b e^{-1} + a}{\sqrt{2 \pi n}}\label{eq:-Qd1}
\end{align}
where \eqref{eq:-Qd} is due to $\left(1 - \frac b {\sqrt n} |z|\right)^2 \geq \frac 1 4$ in $|z| \leq \frac{\sqrt n}{2 b}$, and \eqref{eq:-Qd1} holds because the maximum of \eqref{eq:-Qd} is attained at $z^2 = 8$.
\end{proof}

We are now equipped to prove Lemma \ref{lemma:C}. 
\begin{proof}[Proof of Lemma \ref{lemma:C}]

Define the following functions $\mathcal P \mapsto \mathbb R_+$:
\begin{align}
I(P_{\mathsf X} ) &= I(\mathsf X; \mathsf Y) = \E{\imath_{\mathsf X; \mathsf Y}(\mathsf X; \mathsf Y) } \label{eq:IPx}\\
V( P_{\mathsf X}) &= \E{\Var{\imath_{\mathsf X; \mathsf Y}(\mathsf X; \mathsf Y) \mid \mathsf X} } \label{eq:VPx}\\
T(P_{\mathsf X}) &= \E{ \left| \imath_{\mathsf X; \mathsf Y}(\mathsf X; \mathsf Y) - \E{\imath_{\mathsf X; \mathsf Y}(\mathsf X; \mathsf Y) | \mathsf X}\right|^3\mid \mathsf X } \label{eq:TPx}
\end{align}
If $P_{\mathsf X} = \mathrm{type}(x^n)$, then for each $a \in \mathcal A$, there are $n P_{\mathsf X}(a)$ occurrences of $P_{\mathsf Y| \mathsf X = a}$ among the $\{ P_{\mathsf Y | \mathsf X = x_i}, i = 1, 2, \ldots, n \}$. In the sequel we will invoke  Theorem \ref{thm:Berry-Esseen} with
$W_i = \imath_{\mathsf X; \mathsf Y}(x_i; Y_i)$
where $x^n$ is a given sequence, and \eqref{eq:BerryEsseenDn}--\eqref{eq:BerryEsseenTn} become
 \begin{align}
 D_{n}
 &= \frac 1 n \sum_{a = 1}^{|\mathcal A|} n P_{\mathsf X}(a)\E{\imath_{\mathsf X; \mathsf Y}(a; \mathsf Y) \mid \mathsf X = a} \\
  &=I(P_{\mathsf X}) \\
 V_{n} &= \frac 1 {n} \left[\sum_{a = 1}^{|\mathcal A|} n P_{\mathsf X}(a)\Var{\imath_{\mathsf X; \mathsf Y}(a; \mathsf Y) \mid \mathsf X = a} \right] \\
  &=  V( P_{\mathsf X}) \\
  T_n &=  \frac 1 {n} \left[\sum_{a = 1}^{|\mathcal A|} n P_{\mathsf X}(a) \left| \imath_{\mathsf X; \mathsf Y}(a; \mathsf Y) - \E{\imath_{\mathsf X; \mathsf Y}(a; \mathsf Y) | \mathsf X = a}\right|^3 \right] \notag\\
  &= T(P_{\mathsf X})
\end{align}

Define the (Euclidean) $\delta$-neighborhood of the set of capacity-achieving distributions $\mathcal P^\star$,
\begin{equation}
 \mathcal P^\star_\delta = \left\{ P_{\mathsf X} \in \mathcal P \colon \min_{P_{\mathsf X^\star} \in \mathcal P^\star} \left| P_{\mathsf X} - P_{\mathsf X^\star}\right| \leq \delta \right\} \label{eq:Pstardelta}
\end{equation}

We split the domain of the minimization in the left side of \eqref{eq:lemmaC} into two sets, $\mathrm{type}(x^n) \in \mathcal P^\star_{\delta, [n]}$ and $\mathrm{type}(x^n) \in \mathcal P_{[n]} \backslash \mathcal P^\star_{\delta}$ (recall notation \eqref{eq:-Dn}), for an appropriately chosen $\delta > 0$.

We now show that \eqref{eq:lemmaC} holds for all $\Delta \leq \frac{\Delta_I}{2}$ if the minimization is restricted to types in $\mathcal P_{[n]} \backslash \mathcal P^\star_{\delta}$, where $\delta > 0$ is arbitrary, and
\begin{equation}
\Delta_I = C - \max_{P_{\mathsf X} \in \mathcal P_{[n]} \backslash \mathcal P^\star_{\delta}  } I(P_{\mathsf X}) > 0 \label{eq:Igap}
\end{equation}
By Chebyshev's inequality, for all $x^n$ whose type belongs to $\mathcal P_{[n]} \backslash \mathcal P^\star_{\delta}$,
\begin{align}
 &~\Prob{\sum_{i = 1}^n \imath_{\mathsf X; \mathsf Y}(x_i; Y_i) > n(C - \Delta)}\\
 =&~  \Prob{\sum_{i = 1}^n \imath_{\mathsf X; \mathsf Y}(x_i; Y_i) - n I(P_{\mathsf X}) > n(C - I(P_{\mathsf X}) )- n\Delta}\\
 \leq &~   \Prob{\sum_{i = 1}^n \imath_{\mathsf X; \mathsf Y}(x_i; Y_i) - n I(P_{\mathsf X}) > \frac{n \Delta_I}{2}} \label{eq:-Ccheb1}\\
 \leq&~ \Prob{\left(\sum_{i = 1}^n \imath_{\mathsf X; \mathsf Y}(x_i; Y_i) - n I(P_{\mathsf X})\right)^2 >\frac{n^2 \Delta_I^2}{4}}\\
 \leq&~ \frac{4 n V(P_{\mathsf X})}{ n^2 \Delta_I^2}\\
 \leq&~ \frac{4 \overline V}{n \Delta_I^2}
\end{align}
where in \eqref{eq:-Ccheb1} we used 
\begin{equation}
 \Delta \leq \frac 1 2 \Delta_I < \Delta_I \leq C - I(P_{\mathsf X})
\end{equation}
and
\begin{equation}
 \overline V = \max_{P_{\mathsf X} \in \mathcal P} V(P_{\mathsf X}) \label{eq:Vub} 
\end{equation}
Note that $\overline V < \infty$ by Property \ref{p:continuous} below.
Therefore,
\begin{align}
 &~ \min_{ 
 \mathrm{type}(x^n) \in \mathcal P_{[n]} \backslash \mathcal P^\star_{\delta}  
 }
 \Prob{\sum_{i = 1}^n \imath_{\mathsf X; \mathsf Y}(x_i; Y_i) \leq n(C - \Delta)} 
 \\
 >&~ 1 -  \frac{4 \overline V}{n \Delta_I^2}\label{eq:-lemmaC1}\\
 \geq&~ \Prob{\sum_{i = 1}^n \imath_{\mathsf X; \mathsf Y}(x^{\star}_i; Y_i) \leq n(C - \Delta)} - \frac{4 \overline V}{n \Delta_I^2} \label{eq:-l2}
\end{align}
We conclude that \eqref{eq:lemmaC} holds if the minimization is restricted to types in $\mathcal P_{[n]} \backslash \mathcal P^\star_{\delta}$. 

Without loss of generality, we assume that all outputs in $\mathcal B$ are accessible (which implies that $P_{\mathsf Y^\star}(\mathsf y) > 0$ for all $\mathsf y \in \mathcal B$) and choose $\delta > 0$ so that for all $P_{\mathsf X} \in  \mathcal P^\star_\delta$ and $\mathsf y \in \mathcal B$,  
\begin{equation}
 P_{\mathsf Y}(\mathsf y) > 0  \label{eq:P_Y>0}
 \end{equation}
 where $P_{\mathsf X} \to P_{\mathsf Y| \mathsf X} \to P_{\mathsf Y} $. 
 We recall the following properties of the functions $I( \cdot )$, $V( \cdot )$ and $T( \cdot )$ from \cite[Appendices E and I]{polyanskiy2010channel}.

\begin{property}
The functions  $I(P_{\mathsf X})$, $V( P_{\mathsf X})$ and $T( P_{\mathsf X})$ are continuous on the compact set $\mathcal P$, and therefore bounded and achieve their extrema. 
\label{p:continuous}
\end{property}

\begin{property}
There exists $\ell_1 > 0$ such that for all $\left( P_{\mathsf X^\star},  P_{\mathsf X} \right) \in \mathcal P^\star \times \mathcal P^\star_\delta$,  
\begin{equation}
C -  I(P_{\mathsf X}) \geq \ell_1 \left| P_{\mathsf X} - P_{\mathsf X^\star}\right|^2 \label{eq:nonflat}
\end{equation}
\label{p:nonflat}
\end{property}

\begin{property}
In $\mathcal P^\star_\delta$, the functions $I(P_{\mathsf X})$, $V( P_{\mathsf X})$ and $T( P_{\mathsf X})$ are infinitely differentiable. 
\label{p:diff} 
\end{property}

\begin{property}
In  $\mathcal P^\star$, 
 $V(P_{\mathsf X}) = \Var{\imath_{\mathsf X; \mathsf Y}(\mathsf X; \mathsf Y)}$. 
\label{p:var} 
\end{property}

Due to Property \ref{p:diff}, there exist nonnegative constants $L_1$ and $L_2$ such that for all $\left(P_{\mathsf X}, P_{\mathsf X^\star}\right) \in \mathcal P_\delta^\star \times  \mathcal P^\star$,
\begin{align}
C - I(P_{\mathsf X}) 
 &\leq 
 L_1 \left| P_{\mathsf X} - P_{\mathsf X^\star}\right| \label{eq:L_1}\\
 \left | V(P_{\mathsf X})  - V(P_{\mathsf X^\star}) \right | &\leq L_2 \left| P_{\mathsf X} - P_{\mathsf X^\star}\right|\label{eq:L_2}
\end{align}

To treat the case  $x^n \in \mathcal P^\star_{\delta, [n]}$, we will need to choose $\delta > 0$ carefully and to consider the cases $V_{\min} > 0$ and $V_{\max}  = 0$ separately. 

\subsection{$V_{\min} > 0$.} 
We decrease $\delta$ until, in addition to \eqref{eq:P_Y>0}, 
\begin{equation}
  V_{\min} \leq 2 \min_{P_{\mathsf X} \in \mathcal P_\delta^\star} V\left(P_{\mathsf X}\right) \label{eq:Vlb}
\end{equation}
is satisfied.

  We now show that \eqref{eq:lemmaC} holds if the minimization is restricted to types in $\mathcal P^\star_{\delta, [n]}$, for all 
  $- \underline \Delta \leq \Delta \leq \frac{ \Delta_I}{2}$
  , for an appropriately chosen $\underline \Delta > 0$. 
Using \eqref{eq:Vlb} and boundedness of $T(P_{\mathsf X})$, write
\begin{equation}
 B = \max_{P_{\mathsf X} \in \mathcal P^\star_\delta} \frac {c_0T(P_{\mathsf X})}{V^\frac 3 2(P_{\mathsf X})} 
 \leq \frac{ 2^{\frac 3 2} c_0 \overline T }{V_{\min}^{\frac 3 2} } < \infty
\end{equation}
where
\begin{equation}
 \overline T =  \max_{P_{\mathsf X} \in \mathcal P^\star_\delta} T(P_{\mathsf X}) < \infty
\end{equation}
Therefore, for any $x^n$ with $\mathrm{type}(x^n) \in \mathcal P^\star_{\delta, [n]}$, the Berry-Esseen bound yields:
\begin{equation}
 \left| \Prob{\sum_{i = 1}^n \imath_{\mathsf X; \mathsf Y}(x_i; Y_i) \leq n(C - \Delta)} - Q\left( \nu(P_{\mathsf X})\right)\right| \leq \frac {B}{\sqrt n} \label{eq:-be}
\end{equation}
where
\begin{equation}
 \nu(P_{\mathsf X}) = \frac{n I(P_{\mathsf X}) - n C + n \Delta }{\sqrt{n V(P_{\mathsf X})}}
 \end{equation}
 
 We now apply Lemma \ref{lemma_fg} with $\mathcal D = \mathcal P^\star_{\delta, [n]}$ and \begin{align}
f(P_{\mathsf X}) &= Q\left( \nu(P_{\mathsf X})\right)\\
g(P_{\mathsf X}) &= \Prob{\sum_{i = 1}^n \imath_{\mathsf X; \mathsf Y}(x_i; Y_i) \leq n(C - \Delta)}
\end{align}
Condition \eqref{eq:lemmafg1} of Lemma \ref{lemma_fg} holds with $\eta = \frac B {\sqrt n}$ due to \eqref{eq:-be}. As will be shown in the sequel, the following version of condition \eqref{eq:lemmafg2} holds:
\begin{equation}
 Q(\nu(\Pi(P_{\mathsf X^\star}))) \leq \min_{P_{\mathsf X} \in \mathcal P^\star_{\delta, [n]}} Q(\nu(P_{\mathsf X})) + \frac{q}{\sqrt n} \label{eq:lemmafg2a}
\end{equation}
where $\Pi(P_{\mathsf X^\star})$, the minimum Euclidean distance approximation of $P_{\mathsf X^\star}$ in the set of $n$-types, is formally defined in \eqref{eq:-Pi}, and $q > 0$ will be chosen later. 
Applying Lemma \ref{lemma_fg},  we deduce from \eqref{eq:lemmafg} that 
\begin{align}
\min_{\mathrm{type}(x^n) \in \mathcal P^\star_{\delta, [n]}} &~\Prob{\sum_{i = 1}^n \imath_{\mathsf X; \mathsf Y}(x_i; Y_i) 
\leq n(C - \Delta)}  \notag \\
\geq &~\Prob{\sum_{i = 1}^n \imath_{\mathsf X; \mathsf Y}(x^{\star}_i; Y_i) \leq n(C - \Delta)}
- \frac{q + 2B}{\sqrt n}\label{eq:-l1}
\end{align}
We conclude that \eqref{eq:lemmaC} holds if minimization is restricted to types in $\mathcal P^\star_{\delta, [n]}$. 

We proceed to show \eqref{eq:lemmafg2a}. As will be proven later, for appropriately chosen $\underline L > 0$ and $\bar L > 0$ we can write
\begin{align}
\frac{\sqrt n \Delta}{\sqrt{V(P_{\mathsf X^\star}) } } - \frac{\underline L}{\sqrt n} \sqrt{|\mathcal A| (|\mathcal A| - 1)} \label{eq:tmaxlower} 
&\leq
\nu(\Pi\left( P_{\mathsf X^\star}\right)) \\
&\leq
 \max_{P_{\mathsf X} \in \mathcal P^\star_{\delta, [n]}} \nu(P_{\mathsf X})\\
 &\leq 
  \max_{P_{\mathsf X} \in \mathcal P^\star_{\delta}} \nu(P_{\mathsf X})\\
&\leq
 \frac {\sqrt n \Delta} { \sqrt{V(P_{\mathsf X^\star}) }  } + \sqrt n \bar L \Delta^2\label{eq:tmaxupper}
\end{align}
where $P_{\mathsf X^\star} \in  \mathcal P^\star_{\min}$ if $\Delta \geq 0$, and $P_{\mathsf X^\star} \in  \mathcal P^\star_{\max}$ if $\Delta < 0$. 

Denote
\begin{align}
a &= \underline L \sqrt{\left|\mathcal A \right| (|\mathcal A| - 1)} \\
 b &= V(P_{\mathsf X^\star}) \bar L \\
 z &= \frac{\sqrt n \Delta}{\sqrt{V(P_{\mathsf X^\star}) } } 
\end{align}
If
\begin{equation}
  \Delta \geq - \frac{1}{2 \bar L \sqrt{ V_{\max} }} = - \underline \Delta
\end{equation}
then $z \geq - \frac{\sqrt n}{2b}$,  and Lemma \ref{lemma:Qbound} applies to $z$. So, using   \eqref{eq:tmaxlower}, \eqref{eq:tmaxupper}, the fact that $Q(
\cdot)$ is monotonically decreasing and  Lemma \ref{lemma:Qbound}, we conclude that there exists $q > 0$ such that 
\begin{align}
 &~ Q \left(\nu(\Pi(P_{\mathsf X^\star}))\right) - \min_{P_{\mathsf X} \in \mathcal P^\star_{\delta, [n]}} Q(\nu(P_{\mathsf X})) \notag\\
 = &~  Q \left(\nu(\Pi(P_{\mathsf X^\star}))\right) -  Q\left(\max_{P_{\mathsf X} \in \mathcal P^\star_{\delta, [n]}} \nu(P_{\mathsf X})\right)\\
 \leq &~
 Q \left(z - \frac a {\sqrt n}\right)
 -
 Q \left(  z + \frac {b}{\sqrt n} z^2 \right) \\
 \leq &~ \frac{q}{\sqrt n}\label{eq:-Qbound}
 \end{align}
which is equivalent to \eqref{eq:lemmafg2a}. 

It remains to prove \eqref{eq:tmaxlower} and \eqref{eq:tmaxupper}. Observing that for $a, b > 0$
\begin{align}
\left| \frac 1 {\sqrt a} - \frac 1 {\sqrt b}\right| &= \frac{\left| a - b\right|}{\sqrt a \sqrt b \left( \sqrt a + \sqrt b \right) }\\
&\leq \frac{\left| a - b\right| }{2 \min \left\{ a, b\right\}^{\frac 3 2 }}
\end{align}
and using \eqref{eq:L_2} and \eqref{eq:Vlb}, we have, for all $\left(P_{\mathsf X}, P_{\mathsf X^\star}\right) \in \mathcal P_\delta^\star \times  \mathcal P^\star$, 
\begin{equation}
 \left| \frac 1 {\sqrt{V(P_{\mathsf X})}} - \frac 1 {\sqrt{V(P_{\mathsf X^\star})}} \right| \leq L \left| P_{\mathsf X} - P_{\mathsf X^\star}\right| \label{eq:LL2}
\end{equation}
where 
\begin{equation}
  L = L_2 \sqrt{ \frac{2}{V_{\min}^3 } } \label{eq:L}
\end{equation}
Thus, recalling \eqref{eq:L_1} and 
denoting $\zeta = |P_{\mathsf X} - P_{\mathsf X^\star}|$, we have
\begin{align}
 &~\frac{C - I(P_{\mathsf X}) - \Delta}{\sqrt{V(P_{\mathsf X}) }} 
 \notag \\
 \leq&~ \frac{L_1\zeta - \Delta}{\sqrt{V(P_{\mathsf X}) }}  \label{eq:tmaxlowera0} \\
 \leq&~ \frac{L_1\zeta - \Delta}{\sqrt{ V(P_{\mathsf X^\star}) } }  +  L \zeta\left( L_1 \zeta + |\Delta| \right) \\
 \leq&~ - \frac{\Delta}{\sqrt{ V(P_{\mathsf X^\star}) }} + \underline L \zeta \label{eq:tmaxlowera} 
 \end{align}
where
\begin{equation}
\underline L = \frac{L_1}{\sqrt {V_{\min}}} + 
L \max\left\{ \underline \Delta, \frac{\Delta_I}{2} \right\} 
+ L L_1 \delta 
\end{equation}
So, \eqref{eq:tmaxlower} follows by observing that for any $P_{\mathsf X} \in \mathcal P$,
\begin{equation}
  \left| P_{\mathsf X} - \Pi\left( P_{\mathsf X}\right) \right| \leq \frac 1 n \sqrt{|\mathcal A| (|\mathcal A| - 1)}  \label{eq:-Pibound}
\end{equation}
and letting $P_{\mathsf X} = \Pi(P_{\mathsf X^\star})$ in \eqref{eq:tmaxlowera0}--\eqref{eq:tmaxlowera}. 

To show \eqref{eq:tmaxupper}, we apply Lemma \ref{lemma:maxsum} with 
\begin{align}
 \mathcal D &= \mathcal P^\star_\delta\\
 \mathcal D^\star &= \mathcal P^\star \\
 \varphi &= \sqrt n\\
  \psi &= \sqrt{n} |\Delta|\\
 f\left( P_{\mathsf X}\right) &=  \frac{I(P_{\mathsf X}) - C} { \sqrt {V(P_{\mathsf X})} } \\
 g\left( P_{\mathsf X}\right) &= \frac{1}{ \sqrt{V(P_{\mathsf X})} }
\end{align}
We proceed to verify that conditions of Lemma \ref{lemma:maxsum} are met. 
Function $g$ satisfies condition \eqref{eq:lemmag} with $L$ defined in \eqref{eq:L}. 
Let us now show that function $f$ satisfies condition \eqref{eq:lemmaf} with
$\ell = \frac{\ell_1}{\sqrt {\overline V}}$,
where $\overline V$ and $\ell_1$ are defined in \eqref{eq:Vub} and \eqref{eq:nonflat}, respectively.  
For any $\left( P_{\mathsf X}, P_{\mathsf X^\star}\right) \in \mathcal P^\star_\delta \times \mathcal P^\star$, write
\begin{align}
 f(P_{\mathsf X^\star})  - f(P_{\mathsf X}) &= \frac{C - I(P_{\mathsf X})}{\sqrt{V(P_{\mathsf X})}}\\
 &\geq \frac{ C - I(P_{\mathsf X}) } {\sqrt{\overline V}} \label{eq:-fa}\\
&\geq \frac{\ell_1}{\sqrt {\overline V}} \left| P_{\mathsf X} - P_{\mathsf X^\star}\right|^2\label{eq:-fb}
\end{align}
where \eqref{eq:-fa} follows from \eqref{eq:Vub}, and \eqref{eq:-fb} applies \eqref{eq:nonflat}.
So, Lemma \ref{lemma:maxsum} applies to
$\nu(P_{\mathsf X}) = \varphi f(P_{\mathsf X}) + \mathrm{sign}(\Delta) \psi g(P_{\mathsf X})$, resulting in  \eqref{eq:tmaxupper} with
\begin{equation}
 \bar L = \frac{L_2^2\sqrt{\overline V}}{2 \ell_1 V_{\min}^3}
\end{equation}
thereby completing the proof of \eqref{eq:-l1}.

Combining \eqref{eq:-l2} and \eqref{eq:-l1}, we conclude that \eqref{eq:lemmaC} holds for all $\Delta$ in the interval
\begin{equation}
- \frac{\ell_1 V_{\min}^3}{L_2^2 \sqrt{\overline V V_{\max}}} \leq \Delta \leq \frac{\Delta_I}{2}
\end{equation}

\subsection{$V_{\max} = 0$.} 
We choose $\delta$ so that \eqref{eq:P_Y>0} is satisfied. The case $\mathrm{type}(x^n) \notin \mathcal P^\star_{\delta, [n]}$ was covered in \eqref{eq:-l2}, so we only need to consider the minimization of the left side of \eqref{eq:lemmaC0} over $\mathcal P^\star_{\delta, [n]}$.  Fix $\alpha < \frac 3 2$. If 
\begin{equation}
 \Delta \geq \left( \frac{L_2^2 }{2^8  \ell_1}\right)^{\frac 1 3} \frac 3 {n^{\frac 1 2 + \alpha}}
\end{equation}
we have
\begin{align}
 &~\Prob{\sum_{i = 1}^n \imath_{\mathsf X; \mathsf Y}(x_i; Y_i) > n(C + \Delta)} \notag\\
 = &~  \Prob{\sum_{i = 1}^n \imath_{\mathsf X; \mathsf Y}(x_i; Y_i) - n I(P_{\mathsf X}) > n(C - I(P_{\mathsf X}) )+ n \Delta} 
 \\
 \leq&~ \frac{V(P_{\mathsf X})}{ n \left( C - I(P_{\mathsf X}) + \Delta \right)^2 } \label{eq:bada}\\
 \leq&~ \frac{L_2 |P_{\mathsf X} - P_{\mathsf X^\star}|} {n \left( \ell_1 |P_{\mathsf X} - P_{\mathsf X^\star}|^2 + \Delta \right)^2} \label{eq:badb}\\
 \leq&~ \frac{3 ^{\frac 3 2} L_2  }{  16 \ell_1^{\frac 1 2} } \frac 1 {n \Delta^{\frac 3 2}} \label{eq:badc}\\
 \leq&~ \frac{1}{n^{\frac 1 4  - \frac 3 2 \alpha}} \label{eq:badd}
\end{align}
where 
\begin{itemize}
\item \eqref{eq:bada} is by Chebyshev's inequality;
\item \eqref{eq:badb} uses \eqref{eq:nonflat}, \eqref{eq:L_2} and $V_{\max} = 0$;
\item \eqref{eq:badc} holds  because the maximum of its left side is attained at $|P_{\mathsf X} - P_{\mathsf X^\star}|^2 = \frac{\Delta}{3\ell_1}$.
\end{itemize}
\end{proof}

\section{Proof of the converse part of Theorem \ref{thm:2order}}
\label{appx:2orderC}

Note that for the converse, restriction \eqref{item:last} can be replaced by the following weaker one: 
\begin{enumerate}
 \item[(iv$^\prime$)] The random variable $\jmath_{\mathsf S}(\mathsf S,d)$ has finite absolute third moment. \label{item:jmoment} 
\end{enumerate}
To verify that \eqref{item:last} implies (iv$^\prime$), observe that by the concavity of the logarithm,
\begin{equation}
0 \leq \jmath_{\mathsf S}(\mathsf s, d) + \lambda^\star d\leq   \lambda^\star \E{\mathsf d(\mathsf s, \mathsf Z^\star)} 
\end{equation}
so
\begin{equation}
\E{ \left| \jmath_{\mathsf S}(\mathsf S, d) + \lambda^\star d \right|^3} \leq \lambda^{\star 3} \E{\mathsf d^3(\mathsf S, \mathsf Z^\star)} \label{eq:dtilted3dmoment}
\end{equation}
We now proceed to prove the converse by showing first
that we can eliminate all rates exceeding 
\begin{equation}
 \frac{k}{n} \geq \frac{C}{R(d) - 3 \tau} \label{eq:-2orderChighrate}
\end{equation}
 for any 
 $0 < \tau < \frac {R(d)} 3$.
  More precisely, we show that the excess-distortion probability of any code having such rate converges to $1$ as $n \to \infty$, and therefore for any $\epsilon < 1$, there is an $n_0$ such that for all $n \geq n_0$, no $(k, n, d, \epsilon)$ code can exist for $k$, $n$ satisfying \eqref{eq:-2orderChighrate}.

We weaken \eqref{eq:C11} by fixing $\gamma =  k \tau$ and choosing a particular output distribution, namely, $P_{\bar Y^n} = P_{Y^{n \star}} = P_{\mathsf Y^\star} \times \ldots \times  P_{\mathsf Y^\star}$. Due to restriction \eqref{item:s} in Section \ref{sec:2order}, $P_{Z^k}^\star = P_{\mathsf Z}^\star \times \ldots \times P_{\mathsf Z}^\star$, and the $\mathsf d-$tilted information single-letterizes, that is, for a.e. $s^k$,
\begin{equation}
\jmath_{S^k}(s^k,d) = \sum_{i = 1}^k \jmath_{\mathsf S}(s_i,d) \label{eq:id_iid}
\end{equation}
Theorem \ref{thm:C1} implies that error probability $\epsilon^\prime$ of every $(k, n, d, \epsilon^\prime)$ code must be lower bounded by
\begin{align}
&~ 
 \E{ \!\min_{x^n \in \mathcal A^n }  \Prob{ \sum_{i = 1}^k \jmath_{\mathsf S}(S_i, d) -  \sum_{j = 1}^n \imath_{\mathsf X; \mathsf Y^\star} (x_i; Y_i) \geq k\tau \mid S^k} \!} 
\notag\\ 
-
&~ 
\exp\left(-k\tau\right) \notag\\
\geq
&~ 
\min_{x^n \in \mathcal A^n }  \Prob{ \sum_{j = 1}^n \imath_{\mathsf X; \mathsf Y^\star} (x_i; Y_i) \leq nC + k\tau } 
\notag\\ 
\cdot
&~ 
\Prob{\sum_{i = 1}^k \jmath_{\mathsf S}(S_i, d) \geq nC + 2k\tau} - \exp\left(-k\tau\right)
 \\
  \geq
&~ 
\min_{x^n \in \mathcal A^n }  \Prob{ \sum_{j = 1}^n \imath_{\mathsf X; \mathsf Y^\star} (x_i; Y_i) \leq n C + n \tau^\prime } 
\notag\\ 
\cdot
&~ 
\Prob{\sum_{i = 1}^k \jmath_{\mathsf S}(S_i, d) \geq k R(d) - k\tau } - \exp\left(-k\tau\right)  \label{eq:Cjsscstrong}
\end{align}
where in \eqref{eq:Cjsscstrong}, we used \eqref{eq:-2orderChighrate} and 
$\tau^\prime = \frac {C\tau} {R(d) - 3 \tau} > 0$.
Recalling \eqref{eq:Ejd} and 
\begin{equation}
\E{\imath_{\mathsf X; \mathsf Y^\star} (\mathsf x; \mathsf Y)| \mathsf X = \mathsf x} \leq C \label{eq:E<Capacity}
\end{equation}
 with equality for $P_{\mathsf X^\star}$-a.e. $\mathsf x$, we conclude using the law of large numbers that \eqref{eq:Cjsscstrong} tends to $1$ as $k, n \to \infty$. 

We proceed to show that for all large enough $k, n$, if there is a sequence of $(k, n, d, \epsilon^\prime)$ codes such that
\begin{align}
 - 3 k \tau &\leq nC - k R(d)  \label{eq:-2orderCnlb}\\
 &\leq \sqrt{n V + k \mathcal V(d)}\Qinv{ \epsilon} + \theta\left(n\right) \label{eq:-2orderCnub}
\end{align}
then
$\epsilon^\prime \geq \epsilon$. 

Note that in general the bound in Theorem \ref{thm:C1} with the choice of $P_{\bar Y^n}$ as above does not lead to the correct channel dispersion term. We first consider the general case, in which we apply Theorem \ref{thm:CT}, and then we show the symmetric case, in which we apply Theorem \ref{thm:C1sym}. 

Recall that $x^n \in \mathcal A^n$ has type $P_{\mathsf X}$ if the number of times each letter $a \in \mathcal A$ is encountered in $x^n$ is $n P_{\mathsf X}(a)$. In Theorem \ref{thm:CT}, we weaken the supremum over $W$ by letting $W$ map $X^n$ to its type, $W = \mathrm{type}(X^n)$. Note that the total number of types satisfies (e.g. \cite{csiszar2011information}) $T \leq (n + 1)^{|\mathcal A| - 1}$. We weaken the supremum over $\bar Y^n$ in \eqref{eq:CT} by fixing $P_{\bar Y^n | W = P_{\mathsf X}} = P_{\mathsf Y} \times \ldots \times P_{\mathsf Y}$, where $P_{\mathsf X} \to P_{\mathsf Y| \mathsf X} \to P_{\mathsf Y}$, i.e. $P_{\mathsf Y}$ is the output distribution induced by the type $P_{\mathsf X}$. In this way, Theorem \ref{thm:CT} implies that the error probability of any $(k, n, d, \epsilon^\prime)$ code  must be lower bounded by
\begin{align}
  \epsilon^\prime \geq &~\E{\min_{x^n \in \mathcal A^n}
  \Prob{ \sum_{i = 1}^k \jmath_{\mathsf S}(S_i, d) - \sum_{i = 1}^n \imath_{\mathsf X; \mathsf Y}(x_i; Y_i) \geq \gamma \mid S^k} } \notag\\
   &~- (n + 1)^{|\mathcal A| - 1} \exp\left( -\gamma \right) \label{eq:-2orderCT}
\end{align}
 
Choose
\begin{align}
\gamma &= \left( |\mathcal A| - \frac 1 2\right) \log(n + 1) \label{eq:2orderCT_gamma}
\end{align}

At this point we consider two cases separately, $V > 0$ and $V = 0$.
\subsection{$V > 0$.}
\label{appx:2orderV>0}
 In order to apply Lemma \ref{lemma:C} in Appendix \ref{appx:lemmaC}, we isolate the typical set of source sequences:
\begin{equation}
 \mathcal T_{k, n} = \left\{s^k \in \mathcal S^k \colon  \left| \sum_{i = 1}^k \jmath_{\mathsf S}(s_i, d) - nC \right| \leq n \bar \Delta - \gamma \right\}
\end{equation}
Observe that 
\begin{align}
 &~\Prob{S^k \notin \mathcal T_{k, n}} 
 \notag\\
 =&~ \Prob{\left| \sum_{i = 1}^k \jmath_{\mathsf S}(S_i, d) - nC \right| > n \bar \Delta - \gamma } \\
\leq&~ \Prob{\left| \sum_{i = 1}^k \jmath_{\mathsf S}(S_i, d) - kR(d) \right| + \left| nC - kR(d)\right| + \gamma > n \bar \Delta } \\
 \leq&~ \Prob{\left| \sum_{i = 1}^k \jmath_{\mathsf S}(S_i, d) - kR(d)  \right| > k \frac{\bar \Delta R(d)}{2C} }\label{eq:-2orderTka}\\\
 \leq&~ \frac{4 C^2}{R^2(d) \bar \Delta^2}\frac {\mathcal V(d)} k \label{eq:-2orderTkd}
\end{align}
where 
\begin{itemize}
\item \eqref{eq:-2orderTka} follows by lower bounding 
\begin{align}
&~n \bar \Delta - \gamma - \left| nC - kR(d)\right| \notag\\
\geq&~ n \bar \Delta - \gamma - 3 k \tau \label{eq:-2orderTka1}\\
\geq &~  n \frac{3 \bar \Delta}{4} - 3 k \tau \label{eq:-2orderTka1a}\\
\geq&~ k \frac{3 \bar \Delta}{4 C}\left( R(d) - 3 \tau \right)  - 3 k \tau \label{eq:-2orderTka2}\\
\geq&~ k\frac{\bar \Delta R(d)}{2C} \label{eq:-2orderTka3}
\end{align}
where
\begin{itemize}
 \item \eqref{eq:-2orderTka1} holds for large enough $n$ due to \eqref{eq:-2orderCnlb} and \eqref{eq:-2orderCnub};
 \item \eqref{eq:-2orderTka1a} holds for large enough $n$ by the choice of $\gamma$ in \eqref{eq:2orderCT_gamma}; 
 \item \eqref{eq:-2orderTka2} lower bounds $n$ using \eqref{eq:-2orderCnlb}; 
  \item \eqref{eq:-2orderTka3} holds for a small enough $\tau > 0$.
\end{itemize}
 \item  \eqref{eq:-2orderTkd} is by Chebyshev's inequality. 
\end{itemize}
Now, we let
\begin{equation}
  \epsilon_{k, n} = \epsilon + \frac {B}{\sqrt{n + k}} + \frac 1{\sqrt{n+1}} 
  + \frac{4 C^2}{R^2(d) \bar \Delta^2}\frac {\mathcal V(d)} k \label{eq:-2orderCepsilonkn}
\end{equation}
where $B > 0$ will be chosen in the sequel, and $k, n$ are chosen so that both \eqref{eq:-2orderCnlb} and the following version of \eqref{eq:-2orderCnub} hold:
\begin{equation}
 nC - k R(d) \leq \sqrt{n V  + k \mathcal V(d) -  \frac {L_2 |\mathcal A| }{n + k } }\Qinv{ \epsilon_{k, n}} - \gamma \label{eq:-2orderCnub1}
\end{equation}
 where $L_2 < \infty$ is defined in \eqref{eq:L_2}. Denote for brevity
\begin{equation}
 r(x^n, y^n, s^k) = \sum_{i = 1}^n \imath_{\mathsf X; \mathsf Y}(x_i; y_i) -  \sum_{i = 1}^k \jmath_{\mathsf S}(s_i, d) 
\end{equation}
Weakening \eqref{eq:-2orderCT} using \eqref{eq:2orderCT_gamma} and Lemma \ref{lemma:C}, we can lower bound $\epsilon^\prime$ by
 \begin{align}
&~ \mathbb E \Bigg[ \min_{x^n \in \mathcal A^n}
  \Prob{  r(x^n, Y^n, S^k) \leq - \gamma \mid S^k} 
  \cdot 1\left\{ S^k  \in \mathcal T_{k, n} \right\} \Bigg] 
  \notag\\& 
  - \frac 1 {\sqrt{n+1}} \notag\\ 
   \geq&~ \mathbb E \Bigg[
  \Prob{ r(x^{n \star}, Y^n, S^k) \leq  - \gamma \mid  S^k}  \cdot 1\left\{ S^k  \in \mathcal T_{k, n} \right\} \Bigg] 
 \notag\\&
   - \frac{K}{\sqrt n} - \frac 1 {\sqrt{n+1}} \label{eq:-2orderCTa}\\
      =&~ 
  \Prob{  r(x^{n \star}, Y^n, S^k) \leq  - \gamma, ~ S^k \in \mathcal T_{k, n} } 
  - \frac{K}{\sqrt n} 
  - \frac 1 {\sqrt{n+1}} \label{eq:-2orderCTc}\\
       \geq&~ 
  \Prob{ r(x^{n \star}, Y^n, S^k) \leq  - \gamma }   
  - \Prob{S^k \notin \mathcal T_{k, n} }
   - \frac{K}{\sqrt n}
     \notag\\ &
    - \frac 1 {\sqrt{n+1}} \label{eq:-2orderCTd}\\  
          \geq&~ 
  \Prob{ r(x^{n \star}, Y^n, S^k) \leq  - \gamma }   
  -  \frac{4 C^2}{R^2(d) \bar \Delta^2}\frac {\mathcal V(d)} k - \frac{K}{\sqrt n}
    \notag\\ &
   - \frac 1 {\sqrt{n+1}} \label{eq:-2orderCTe}\\ 
   \geq&~ \epsilon \label{eq:-2orderCTf}
\end{align}
where \eqref{eq:-2orderCTa} is by Lemma \ref{lemma:C}, and \eqref{eq:-2orderCTd} is by the union bound. To justify \eqref{eq:-2orderCTf}, observe that the quantities in Theorem \ref{thm:Berry-Esseen} corresponding to the sum of independent random variables in \eqref{eq:-2orderCTe} are
\begin{align}
 D_{n+k}
  &= \frac {n}{n+k}I( \Pi( P_{\mathsf X^\star})) - \frac k {n + k} R(d)\\
  &\leq \frac n {n + k} C  - \frac k {n + k} R(d) \label{eq:-2orderCIlb}\\
 V_{n+k} 
  &=  \frac n {n+k} V( \Pi( P_{\mathsf X^\star})) + \frac k {n + k} \mathcal V(d) \\
  &\geq \frac n {n+k} V + \frac k {n + k} \mathcal V(d) - \frac {L_2 |\mathcal A| }{n + k } \label{eq:-2orderCVlb}\\
 T_{n +k} 
  &=  \frac n {n + k} T(\Pi( P_{\mathsf X^\star})) + \frac k {n + k} \E{\left|\jmath_{\mathsf S}(\mathsf S, d) - R(d) \right|^3}
  \end{align}
where the functions $\Pi(\cdot)$, $I(\cdot)$, $V(\cdot)$, $T(\cdot)$ are defined in \eqref{eq:-Pi}, \eqref{eq:IPx}--\eqref{eq:TPx} in Appendix \ref{appx:lemmaC}. 
To show \eqref{eq:-2orderCVlb},  recall that $V(P_{\mathsf X^\star}) = V$ by Property \ref{p:var} in Appendix \ref{appx:lemmaC}, and use \eqref{eq:L_2} and \eqref{eq:-Pibound}.  
Further, $T_{n+k}$ is bounded uniformly in $P_{\mathsf X}$, so \eqref{eq:BerryEsseenBn} is upper bounded by some constant $B >0$. 
Finally, applying \eqref{eq:-2orderCIlb} and \eqref{eq:-2orderCVlb} to \eqref{eq:-2orderCnub1}, we conclude that
\begin{equation}
-\gamma \geq (n+k)D_{n + k} - \sqrt{(n+k)V_{n + k}}\Qinv{ \epsilon_{k, n}} \label{eq_-2ordergammaub}
\end{equation}
which enables us to lower bound the probability in \eqref{eq:-2orderCTe} invoking the Berry-Esseen bound (Theorem \ref{thm:Berry-Esseen}). In view of \eqref{eq:-2orderCepsilonkn}, the resulting bound is equal to $\epsilon$, and the proof of \eqref{eq:-2orderCTf} is complete.

\subsection{$V = 0$.}
\label{appx:2orderV=0}
 Fix $0 < \alpha < \frac 1 6$.
 
If $\mathcal V(d) > 0$, 
we choose $\gamma$ as in \eqref{eq:2orderCT_gamma}, and
\begin{equation}
  \epsilon_{k, n} = \epsilon + \frac {B}{\sqrt{k}} + (n + 1)^{|\mathcal A| - 1} \exp\left( -\gamma \right) +  \frac 1 {n^{\frac 1 4 - \frac 3 2 \alpha}}  \label{eq:-2orderCepsilonkn0}
\end{equation}
where $B > 0$ is the same as in \eqref{eq:-2orderCepsilonkn}, and $k, n$ are chosen so that 
the following version of \eqref{eq:-2orderCnub} hold:
\begin{equation}
 nC - k R(d) \leq \sqrt{k \mathcal V(d) }\Qinv{ \epsilon_{k, n}} - \gamma - \bar \Delta n^{\frac 1 2 - \alpha}\label{eq:-2orderCnub0}
\end{equation}
where $\bar \Delta > 0$ was defined in Lemma \ref{lemma:C}. 
Weakening \eqref{eq:-2orderCT} using \eqref{eq:lemmaC0}, we have
 \begin{align}
  \epsilon^\prime 
  &\geq\min_{x^n \in \mathcal A^n} \Prob{\imath_{\mathsf X; \mathsf Y}(x_i, Y_i) \geq nC + \bar \Delta n^{\frac 1 2 - \alpha}} 
  \notag \\
  &\cdot \Prob{\sum_{i = 1}^k \jmath_{\mathsf S}(S_i, d) \geq nC +  \bar \Delta n^{\frac 1 2 - \alpha} + \gamma }
  \notag \\
  &-(n + 1)^{|\mathcal A| - 1} \exp\left( -\gamma \right)\\
  &\geq \left( 1 - \frac 1 {n^{\frac 1 4 - \frac 3 2 \alpha}}\right)
  \notag \\&
  \cdot \Prob{\sum_{i = 1}^k \jmath_{\mathsf S}(S_i, d) \geq k R(d) + \sqrt{k \mathcal V(d)}\Qinv{\epsilon_{k, n}} }
  \notag \\
  &-(n + 1)^{|\mathcal A| - 1} \exp\left( -\gamma \right) \label{eq:-2orderCTa0}\\
   &\geq 
  \left( 1 - \frac 1 {n^{\frac 1 4 - \frac 3 2 \alpha}}\right) \left( \epsilon_{k, n} - \frac B {\sqrt k}\right) 
  \notag\\&
   -(n + 1)^{|\mathcal A| - 1} \exp\left( -\gamma \right)\label{eq:-2orderCTb0}\\  
          &\geq 
  \epsilon_{k, n} - \frac B {\sqrt k}  
   -  \frac{1}{n^{\frac 1 4 - \frac 3 2 \alpha} }-(n + 1)^{|\mathcal A| - 1} \exp\left( -\gamma \right) \\
   &= \epsilon 
\end{align}
where \eqref{eq:-2orderCTa0} uses \eqref{eq:lemmaC0} and \eqref{eq:-2orderCnub0}, and \eqref{eq:-2orderCTb0} is by the Berry-Esseen bound. 

If $\mathcal V(d) = 0$, which implies $\jmath_{S}( S_i, d) = R(d)$ a.s., we let
\begin{equation}
 \gamma = \left( |\mathcal A|- 1\right) \log (n + 1) - \log \left( 1 - \epsilon - \frac 1 {n^{\frac 1 4 - \frac 3 2 \alpha} } \right) \label{eq:-2ordergamma00}
\end{equation}
and choose $k, n$ that satisfy 
\begin{equation}
 k R(d) - nC \geq \gamma + \bar \Delta n^{\frac 1 2 - \alpha}\label{eq:-2orderCnub00}
\end{equation}
Then, plugging $\jmath_{S}( S_i, d) = R(d)$ a.s. in \eqref{eq:-2orderCT}, we have
\begin{align}
  \epsilon^\prime \geq &~\min_{x^n \in \mathcal A^n}
  \Prob{ \sum_{i = 1}^n \imath_{\mathsf X; \mathsf Y}(x_i; Y_i)  \leq k R(d) - \gamma } 
  \notag\\
   &~
   - (n + 1)^{|\mathcal A| - 1} \exp\left( -\gamma \right) \\
   \geq &~
   \min_{x^n \in \mathcal A^n}
  \Prob{ \sum_{i = 1}^n \imath_{\mathsf X; \mathsf Y}(x_i; Y_i)  \leq nC + \bar \Delta n^{\frac 1 2 - \alpha} } 
  \notag\\
   &~
   - (n + 1)^{|\mathcal A| - 1} \exp\left( -\gamma \right) \label{eq:-CV00a}\\
   \geq &~ 1 - \frac 1 {n^{\frac 1 4 - \frac 3 2 \alpha}} - (n + 1)^{|\mathcal A| - 1} \exp\left( -\gamma \right)\label{eq:-CV00b} \\
   = &~ \epsilon\label{eq:-CV00c}
   \end{align}
   where \eqref{eq:-CV00a} is by the choice of $k, n$ in \eqref{eq:-2orderCnub00},  \eqref{eq:-CV00b} invokes \eqref{eq:lemmaC0}, and \eqref{eq:-CV00c} follows from the choice of $\gamma$ in \eqref{eq:-2ordergamma00}. 

\subsection{Symmetric channel.}
\label{appx:2orderCsym}
We show that if the channel is such that the distribution of 
$\imath_{\mathsf X; \mathsf Y^\star}(\mathsf x; \mathsf Y)$ (according to $P_{\mathsf Y|\mathsf X = \mathsf x}$)
does not depend on the choice $\mathsf x \in \mathcal A$, Theorem \ref{thm:C1sym} leads to a tighter third-order term than \eqref{eq:Cremainder}. 

If either $V > 0$ or $\mathcal V(d) > 0$, let
\begin{align}
 \gamma &= \frac 1 2 \log n \label{eq:-2orderCepsilonknsymgamma}\\
  \epsilon_{k, n} &= \epsilon + \frac {B}{\sqrt{n + k}} + \frac 1{\sqrt{n}} \label{eq:-2orderCepsilonknsym}
\end{align}
where $B > 0$ can be chosen as in \eqref{eq:-2orderCepsilonkn}, and let $k, n$ be such that the following version of \eqref{eq:-2orderCnub} (with the remainder $\theta(n)$ satisfying \eqref{eq:Cremainder} with $\underline c = \frac 1 2$) holds:
\begin{equation}
 nC - k R(d) \leq \sqrt{n V  + k \mathcal V(d)}\Qinv{ \epsilon_{k, n}} - \gamma \label{eq:-2orderCnub1sym}
\end{equation}
Theorem \ref{thm:C1sym} and Theorem \ref{thm:Berry-Esseen} imply that the error probability of every $(k, n, d, \epsilon^\prime)$ code must satisfy, for an arbitrary sequence $x^n \in \mathcal A^n$,
\begin{align}
\epsilon^\prime
\geq
&~ 
 \Prob{ \sum_{i = 1}^k \jmath_{\mathsf S}(S_i, d) -  \sum_{j = 1}^n \imath_{\mathsf X; \mathsf Y^\star} (x_i; Y_i) \geq \gamma}  \label{eq:2orderC1sym}
-
\exp\left(-\gamma \right)\\
\geq
&~
\epsilon  \label{eq:2orderC1syma}
\end{align}

If both $V = 0$ and $\mathcal V(d) = 0$, choose $k, n$ to satisfy
\begin{align}
k R(d) - nC  &\geq \gamma\\
  &=  \log \frac 1 {1 - \epsilon} \label{eq:gammasym00}
\end{align}
Substituting \eqref{eq:gammasym00} and $\jmath_{\mathsf S}(S_i, d) = R(d)$, $\imath_{\mathsf X; \mathsf Y^\star} (x_i; Y_i) = C$ a.s. in \eqref{eq:2orderC1sym}, we conclude that the right side of \eqref{eq:2orderC1sym} equals $\epsilon$, so $\epsilon^\prime \geq \epsilon$ whenever a $(k, n, d, \epsilon^\prime)$ code exists.

\subsection{Gaussian channel}
\label{appx:2orderCgauss}
In view of Remark \ref{rem:Gauss}, it suffices to consider the equal power constraint \eqref{eq:equalP}. 
The spherically-symmetric 
$P_{\bar Y^n} = P_{Y^{ n \star}} = P_{\mathsf Y^\star} \times \ldots \times P_{\mathsf Y^\star}$, where $\mathsf Y^\star \sim \mathcal N(0, \sigma_{\mathsf N}^2(1 + P))$, satisfies the symmetry assumption of Theorem \ref{thm:C1sym}. In fact, for all $x^n \in \mathcal F (\alpha)$,   $\imath_{X^n;  Y^{n \star}}(x^n; Y^n)$ has the same distribution under $P_{Y^{n}|X^n = x^n}$ as (cf. \eqref{eq:ixyAWGNconditional})
\begin{equation}
 G_n = \frac n 2 \log\left(1+P\right) - \frac{\log e}{2}\left( \frac{P}{1 + P} \sum_{i = 1}^n \left( W_i - \frac 1 {\sqrt P} \right)^2 - n\right)\label{eq:ixyAWGNconditionalsum}
\end{equation}
where $W_i \sim \mathcal N\left( \frac 1 {\sqrt P}, 1\right)$, independent of each other. 
Since $G_n$ is a sum of i.i.d. random variables, the mean of $\frac {G_n}{n}$ is equal to $C =  \frac 1 2 \log\left(1+P\right)$ and its variance is equal to \eqref{eq:DispersionAWGN}, the result follows analogously to  \eqref{eq:-2orderCepsilonknsymgamma}--\eqref{eq:2orderC1syma}. 

\section{Proof of the achievability part of Theorem \ref{thm:2order}}
\label{appx:2orderA}

\subsection{Almost lossless coding  ($d = 0$) over a DMC.} 
\label{appx:2orderAlossless}
The proof consists of an asymptotic analysis of the bound in Theorem \ref{thm:Alossless} by means of Theorem \ref{thm:Berry-Esseen}. Weakening \eqref{eq:Alossless} by fixing $P_{X^n} = P_{X^n}^\star = P_{\mathsf X^\star} \times \ldots \times P_{\mathsf X^\star}$, we conclude that there exists a $(k, n, 0, \epsilon^\prime)$ code with
\begin{equation}
 \epsilon^\prime \leq \E{\exp\left( - \left| \sum_{i = 1}^n \imath_{\mathsf X; \mathsf Y}^\star\left( X_i^\star; Y_i^\star \right) - \sum_{i = 1}^k \imath_{\mathsf S}(S_i)\right|^+ \right)} \label{eq:-Alosslessa}
\end{equation}
where $(S^k, {X^n}^\star, {Y^n}^\star)$ are distributed according to $P_{S^k}P_{{X^n}^\star} P_{Y^n|X^n}$. 
The case of equiprobable $\mathsf S$ has been tackled in \cite{polyanskiy2010channel}. Here we assume that $\imath_{\mathsf S}(\mathsf S)$ is not a constant, that is, 
$\Var{\imath_{\mathsf S}(\mathsf S)} > 0$. 

Let $k$ and $n$ be such that
\beq
nC - kH(\mathsf S) \geq \sqrt{n V + k \mathcal V} \Qinv{\epsilon - \frac{B + 1}{\sqrt{n+k}}} + \frac 1 2 \log (n+k) \label{eq:-Alosslessb}
\eeq
where $\mathcal V = \Var{\imath_{\mathsf S}(\mathsf S)}$, and $B$ is the Berry-Esseen ratio \eqref{eq:BerryEsseenBn} for the sum of $n + k$ independent random variables appearing in the right side of \eqref{eq:-Alosslessa}. Note that $B$ is finite due to:
\begin{itemize}
\item  $\Var{\imath_{\mathsf S}(\mathsf S)} > 0$;
\item the third absolute moment of $\imath_{\mathsf S}(\mathsf S)$ is finite;
\item the third absolute moment of $\imath_{\mathsf X; \mathsf Y}^\star(\mathsf X^\star; \mathsf Y^\star)$ is finite, as observed in Appendix \ref{appx:lemmaC}. 
\end{itemize}
Therefore,  \eqref{eq:-Alosslessb} can be written as \eqref{eq:2order} with the remainder therein satisfying \eqref{eq:Aremainderlossless}. So, it suffices to prove that if $k, n$ satisfy \eqref{eq:-Alosslessb}, then the right side of \eqref{eq:-Alosslessa} is upper bounded by $\epsilon$. Let
\begin{align}
&~ \mathcal T_{k, n} =
  \bigg\{
 \left(s^k, x^{n}, y^{n}\right) \in \mathcal S^k \times \mathcal A^n \times \mathcal B^n \colon
 \notag\\
 &~ \sum_{i = 1}^n \imath_{\mathsf X; \mathsf Y}^\star\left( x_i; y_i \right) 
  - \sum_{i = 1}^k \imath_{\mathsf S}(s_i) 
  \notag\\
   \geq&~ nC -  kH(\mathsf S) - \sqrt{n V + k \mathcal V} \Qinv{\epsilon - \frac{B + 1}{\sqrt{n+k}}} 
 \bigg\} \label{eq:-2orderATknlossless}
\end{align}
By the Berry-Esseen bound (Theorem \ref{thm:Berry-Esseen}), 
\begin{equation}
 \Prob{\left(S^k, X^{n \star}, Y^{n \star}\right) \notin \mathcal T_{k, n}} \leq \epsilon - \frac{1}{\sqrt{n + k}}
\end{equation}
We now further upper bound \eqref{eq:-Alosslessa} as
\begin{align}
\epsilon^\prime 
&\leq 
\mathbb E \Bigg[
\exp\left( - \left| \sum_{i = 1}^n \imath_{\mathsf X; \mathsf Y}^\star\left( X_i^\star; Y_i^\star \right) - \sum_{i = 1}^k \imath_{\mathsf S}(S_i)\right|^+  \right) 
\notag\\
&\cdot
1_{\mathcal T_{k, n}}\left(S^k, X^{n \star}, Y^{n \star}\right) 
\Bigg]
+ \Prob{\left(S^k, X^{n \star}, Y^{n \star}\right) \notin \mathcal T_{k, n}} \label{eq:-Alosslessc1}\\
&\leq \frac 1 {\sqrt {n + k}} \Prob{\left(S^k, X^{n \star}, Y^{n \star}\right)  \in T_{k, n}}
+ \epsilon - \frac{1}{\sqrt{n + k}}\label{eq:-Alosslessc}\\
&\leq \epsilon \label{eq:-Alosslessd}
\end{align}
where we invoked \eqref{eq:-Alosslessb} and \eqref{eq:-2orderATknlossless} to upper bound the exponent in the right side of \eqref{eq:-Alosslessc1}. 

\subsection{Lossy coding over a DMC.} 
\label{appx:2orderAlossy}
The proof consists of the asymptotic analysis of the bound in Theorem \ref{thm:A} using Theorem \ref{thm:Berry-Esseen} and Lemma \ref{lemma:aepA} below, which deals with asymptotic behavior of distortion $d$-balls. Note that Lemma \ref{lemma:aepA} is the only step that requires finiteness of the ninth absolute moment of $\mathsf d(\mathsf S, \mathsf Z^\star)$ as required by restriction \eqref{item:last} in Section \ref{sec:2order}. 
\begin{lemma}[{\cite[Lemma 2]{kostina2011fixed}}] Under restrictions \eqref{item:s}--\eqref{item:last}, there exist constants $k_0, c, K > 0$ such that for all $k \geq k_0$, 
\begin{align}
 & \Prob{ \log \frac 1 {P_{Z^{k\star}}(B_d(S^k))} \leq \sum_{i = 1}^k \jmath_{\mathsf S}(S_i, d) + \left( \bar c - \frac 1 2 \right) \log k + c } 
 \notag\\
  &\geq
  1 - \frac K {\sqrt k}
\end{align}
where $\bar c$ is given by \eqref{eq:Cbar}.
\label{lemma:aepA}
\end{lemma}
We weaken \eqref{eq:A} by fixing 
\begin{align}
 P_{X^n} &= P_{X^{n \star}} = P_{\mathsf X^\star} \times \ldots \times P_{\mathsf X^\star}\\
 P_{Z^k} &= P_{Z^{k \star}} = P_{\mathsf Z^\star} \times \ldots \times P_{\mathsf Z^\star}\\
\gamma &= \frac 1 2 \log_e k + 1 \label{eq:-Agamma}
\end{align}
where $\Delta > 0$, so there exists a $(k, n, d, \epsilon^\prime)$ code with error probability $\epsilon^\prime$ upper bounded by
\begin{align}
&
\E{ \exp\left( - \left| \sum_{i = 1}^n \imath_{\mathsf X; \mathsf Y}^\star(X_i^\star; Y_i^\star )  - \log \frac{ \gamma}{P_{Z^{k \star}}(B_d(S^k))} 
\right|^+ \right)}
\notag\\&
+ e^{1 - \gamma}
\label{eq:-2orderAa1}
\end{align}
where $(S^k, {X^n}^\star, {Y^n}^\star, Z^{k \star})$ are distributed according to $P_{S^k}P_{{X^n}^\star} P_{Y^n|X^n} P_{Z^{k \star}}$. 
We need to show that for $k, n$ satisfying \eqref{eq:2order}, \eqref{eq:-2orderAa1} is upper bounded by $\epsilon$. 

We apply Lemma \ref{lemma:aepA} to upper bound  \eqref{eq:-2orderAa1} as follows:
\begin{align}
 \epsilon^\prime 
 \leq&~ \E{\exp\left( - \left|U_{k, n}\right|^+ \right)} + \frac {K+1} {\sqrt k} \label{eq:-2orderA1sta}
\end{align}
with
\begin{align}
U_{k, n} 
&= \sum_{i = 1}^n \imath_{\mathsf X; \mathsf Y}^\star(X_i^\star; Y_i^\star )  - \sum_{i = 1}^k \jmath_{\mathsf S}(S_i, d) - \left( \bar c - \frac 1 2 \right)\log k 
\notag\\
&- \log \gamma  - c \label{eq:-2orderAUkn}
\end{align}

We first consider the (nontrivial) case $\mathcal V(d) + V > 0$. Let $k$ and $n$ be such that
\begin{align}
nC - kR(d) &\geq \sqrt{n V + k \mathcal V(d)} \Qinv{\epsilon_{k, n} } 
\notag\\
&+ \bar c  \log k + \log{\gamma} + c  \label{eq:-2orderAe}\\
\epsilon_{k, n} &= \epsilon - \frac{B}{\sqrt{n+k}} - \frac {K + 2} {\sqrt k}
\end{align}
where constants $c$ and $\bar c$ are defined in Lemma \ref{lemma:aepA}, and $B$ is the Berry-Esseen ratio \eqref{eq:BerryEsseenBn} for the sum of $n + k$ independent random variables appearing in \eqref{eq:-2orderA1sta}. 
Note that $B$ is finite because:
\begin{itemize}
\item  either $\mathcal V(d) > 0$ or $V > 0$ by the assumption;
\item the third absolute moment of $\jmath_{\mathsf S}(\mathsf S, d)$ is finite by restriction \eqref{item:last} as spelled out in \eqref{eq:dtilted3dmoment};
\item the third absolute moment of $\imath_{\mathsf X; \mathsf Y}^\star(\mathsf X^\star; \mathsf Y^\star)$ is finite, as observed in Appendix \ref{appx:lemmaC}. 
\end{itemize}
Applying a Taylor series expansion to \eqref{eq:-2orderAe} with the choice of $\gamma$ in \eqref{eq:-Agamma}, we conclude that \eqref{eq:-2orderAe} can be written as  \eqref{eq:2order} with the remainder term satisfying \eqref{eq:Aremainder}. 

It remains to further upper bound \eqref{eq:-2orderA1sta} using \eqref{eq:-2orderAe}. Let
\begin{align}
 \mathcal T_{k, n} = \bigg\{ 
 &\left(s^k, x^{n}, y^{n}\right) \in \mathcal S^k \times \mathcal A^n \times \mathcal B^n \colon
 \notag\\
  &\sum_{i = 1}^n 
  \imath_{\mathsf X; \mathsf Y}^\star(x_i; y_i )  - \sum_{i = 1}^k \jmath_{\mathsf S}(s_i, d)
  \notag\\ 
  &\geq nC -  kR(d) - \sqrt{n V + k \mathcal V(d)} \Qinv{\epsilon_{k, n}}
 \bigg\}  \label{eq:-2orderAf}
\end{align}
By the Berry-Esseen bound (Theorem \ref{thm:Berry-Esseen}), 
\begin{equation}
 \Prob{\left(S^k, X^{n \star}, Y^{n \star}\right) \notin \mathcal T_{k, n}} \leq \epsilon_{k, n} + \frac{B}{\sqrt{n + k}}
\end{equation}
so the expectation in the right side of \eqref{eq:-2orderA1sta} is upper-bounded as
\begin{align}
  &~\E{ \exp\left( - \left| U_{k, n} \right|^+ \right)} \notag\\
  \leq&~ 
  \E{\exp\left( - \left| U_{k, n}\right|^+ 1\left\{ \left(S^k, X^{n \star}, Y^{n \star}\right) \in \mathcal T_{k, n}\right\} \right)}
\notag\\&
+ \Prob{\left(S^k, X^{n \star}, Y^{n \star}\right) \notin \mathcal T_{k, n}} \label{eq:-2orderAg}\\
 \leq&~\frac 1 {\sqrt k}\Prob{ \left(S^k, X^{n \star}, Y^{n \star}\right) \in \mathcal T_{k, n} } + \epsilon_{k, n} + \frac{B}{\sqrt{n + k}} \label{eq:-2orderA1stb}
\end{align}
where we used \eqref{eq:-2orderAe} and \eqref{eq:-2orderAf} to upper bound the exponent in the right side of \eqref{eq:-2orderAg}.  

Putting \eqref{eq:-2orderA1sta} and \eqref{eq:-2orderA1stb} together, we conclude that $\epsilon^\prime \leq \epsilon$. 

Finally, consider the case $V = \mathcal V(d) = 0$, which implies  $\jmath_{\mathsf S}(\mathsf S, d) = R(d)$ and $\imath_{\mathsf X; \mathsf Y}^\star(X_i^\star; Y_i^\star ) = C$ almost surely, and let $k$ and $n$ be such that
\begin{equation}
nC - kR(d) \geq \left( \bar c - \frac 1 2 \right) \log k + \log{\gamma} + c   + \log \frac 1 {\epsilon - \frac{K + 1}{\sqrt k}} 
\end{equation}
where constants $c$ and $\bar c$ are defined in Lemma \ref{lemma:aepA}. Then
\begin{equation}
  \E{ \exp\left( - \left| U_{k, n} \right|^+ \right)}
  \leq  \epsilon - \frac{K + 1}{\sqrt k}
  \end{equation}
which, together with \eqref{eq:-2orderA1sta}, implies that $\epsilon^\prime \leq \epsilon$, as desired. 

\subsection{Lossy or almost lossless coding over a Gaussian channel}
\label{appx:2orderAgauss}
In view of Remark \ref{rem:Gauss}, it suffices to consider the equal power constraint \eqref{eq:equalP}. As shown in the proof of Theorem \ref{thm:AGMS-AWGN}, for any distribution of $X^n$ on the power sphere, 
\begin{align}
 \imath_{X^n; Y^n}(X^n; Y^n) \geq G_n - F 
\end{align}
where $G_n$ is defined in \eqref{eq:ixyAWGNconditionalsum} (cf. \eqref{eq:ixyAWGNconditional}) and $F$ is a (computable) constant. 

Now, the proof for almost lossless coding in Appendix \ref{appx:2orderAlossless} can be modified to work for the Gaussian channel by adding $\log F$ to the right side of \eqref{eq:-Alosslessb} 
and replacing 
$\sum_{i = 1}^n \imath_{\mathsf X; \mathsf Y}^\star\left( X_i^\star; Y_i^\star \right)$
in \eqref{eq:-Alosslessa} and \eqref{eq:-Alosslessc1} with $G_n - \log F$, and in \eqref{eq:-2orderATknlossless} with $G_n$. 

Similarly, the proof for lossy coding in Appendix \ref{appx:2orderAlossy} is adapted for the Gaussian channel by adding $\log F$ to the right side of \eqref{eq:-2orderAe}
and replacing
$\sum_{i = 1}^n \imath_{\mathsf X; \mathsf Y}^\star\left( X_i^\star; Y_i^\star \right)$
in \eqref{eq:-2orderAa1} and \eqref{eq:-2orderAUkn}
 with $G_n - \log F$, and in \eqref{eq:-2orderAf} with $G_n$.
 
\section{Proof of Theorem \ref{thm:2order1}}
\label{appx:2order1}
Applying the Berry-Esseen bound to \eqref{eq:Cuncoded}, we obtain
\begin{align}
 &~D_1(n, \epsilon, \alpha) 
 \notag\\
 \geq&\! \min_{\substack{P_{\mathsf Z|\mathsf S} \colon\\ I(\mathsf S; \mathsf Z) \leq C(\alpha) }} \left\{\!\E{\mathsf d(\mathsf S, \mathsf Z)} + \sqrt \frac{\Var{\mathsf d(\mathsf S, \mathsf Z)}}{n} \Qinv{\epsilon + \frac{B}{\sqrt n}} \!\right\}\\
 =&~  D(C(\alpha)) + \sqrt \frac{\mathscr W_1(\alpha)}{n} \Qinv{\epsilon+ \frac{B}{\sqrt n}}  \label{eq:-2order1a}
\end{align}
where $B$ is the Berry-Esseen ratio, and  \eqref{eq:-2order1a} follows by the application of Lemma \ref{lemma:maxsum} with 
\begin{align}
 \mathcal D &=\left\{ P_{\mathsf S \mathsf Z} = P_{\mathsf Z| \mathsf S} P_{\mathsf S} \colon I(\mathsf S; \mathsf Z) \leq R(\bar d) \right\} \\
 f(P_{\mathsf S \mathsf Z}) &= - \E{d(\mathsf S, \mathsf Z)}\\
 g(P_{\mathsf S \mathsf Z}) &= - \sqrt{\Var{d(\mathsf S, \mathsf Z)}}  \Qinv{\epsilon+ \frac{B}{\sqrt n}}\\
 \varphi &= 1\\
\psi &= \frac 1 {\sqrt n}
\end{align}
Note that the mean and standard deviation of $d(\mathsf S, \mathsf Z)$ are linear and continuously differentiable in $P_{\mathsf S \mathsf Z}$, respectively, so conditions \eqref{eq:lemmag} and \eqref{eq:lemmaflinear} hold with the metric being the usual Euclidean distance between vectors in $\mathbb R^{|\mathcal S| \times |\hat {\mathcal S}|}$. So, \eqref{eq:-2order1a} follows immediately upon observing that by the definition of the rate-distortion function, $\E{\mathsf d(\mathsf S, \mathsf Z)} \geq \E{\mathsf d(\mathsf S, \mathsf Z^\star)} = D(C(\alpha))$ for all $P_{\mathsf Z|\mathsf S}$ such that $I(\mathsf S; \mathsf Z) \leq C(\alpha)$. 

\bibliographystyle{IEEEtran}
\bibliography{../../ratedistortion}
\newpage
\begin{IEEEbiographynophoto}{Victoria Kostina}(S'12)
received the BachelorÕs degree with honors in applied mathematics
and physics from the Moscow Institute of Physics and Technology,
Russia, in 2004, where she was affiliated with the Institute for Information Transmission Problems of the Russian Academy of Sciences, and the MasterÕs degree in electrical engineering from the University
of Ottawa, Canada, in 2006.

She is currently pursuing a Ph.D. degree in electrical engineering at Princeton University. Her research interests lie in information theory, theory of random processes, coding, and wireless
communications.
\end{IEEEbiographynophoto}

\begin{IEEEbiographynophoto}{Sergio Verd\'{u}}(S'80--M'84--SM'88--F'93)
received the Telecommunications Engineering degree from the
Universitat Polit\`{e}cnica de Barcelona in 1980, and the Ph.D. degree in Electrical Engineering from the
University of Illinois at Urbana-Champaign in 1984. Since 1984 he has been a member of the faculty of
Princeton University, where he is the Eugene Higgins Professor of Electrical Engineering, and is a member
of the Program in Applied and Computational Mathematics.

Sergio Verd\'{u} is  the recipient of the 2007 Claude E. Shannon Award, and
the 2008 IEEE Richard W. Hamming Medal. 
He is a member of the  National Academy of Engineering, 
and was awarded a Doctorate Honoris Causa from the Universitat  Polit\`{e}cnica de Catalunya in 2005.

He is a recipient of several paper awards from the IEEE: 
the 1992 Donald Fink Paper Award, 
the 1998 and 2012 Information Theory  Paper Awards, 
an Information Theory Golden Jubilee Paper Award,
the 2002 Leonard Abraham Prize Award,  
the 2006 Joint Communications/Information Theory Paper Award, 
and the 2009 Stephen O. Rice Prize from IEEE Communications Society.  
In 1998, Cambridge University Press published his book {\em Multiuser Detection,} 
for which he received the 2000 Frederick E. Terman Award from the American Society for Engineering Education. 

Sergio Verd\'{u} served as President of the IEEE Information Theory Society in 1997, and
on its Board of Governors (1988-1999, 2009-present).
He has also served in various editorial capacities for the {\em IEEE Transactions on Information Theory}:
Associate Editor (Shannon Theory, 1990-1993; Book Reviews, 2002-2006),  
Guest Editor of the Special Fiftieth Anniversary Commemorative Issue
(published by IEEE Press as ``Information Theory: Fifty years of discovery"), 
and member of the Executive Editorial Board (2010-2013).
He is the founding Editor-in-Chief of {\em Foundations and Trends in Communications and Information Theory}.
\end{IEEEbiographynophoto}

\end{document}